\documentclass[mnsc,non-blindrev]{informs3}
\OneAndAHalfSpacedXI 

\usepackage{endnotes}
\usepackage[ruled]{algorithm}
\usepackage{algorithmicx,algpseudocode,amsfonts,subfigure,optidef}
\usepackage{helvet}
\usepackage{graphicx,comment,appendix}
\usepackage{hyperref}
\hypersetup{%
    pdfborder = {0 0 0}
}
\usepackage[capitalize,noabbrev]{cleveref}
\usepackage{tikz}
\usepackage{multirow,makecell}
\usepackage{subfigure,caption,microtype,booktabs} 

\usepackage{amsthm,amsmath,amssymb,soul}
\usepackage{comment,enumitem,diagbox}
\usepackage{subcaption}
\setlist{leftmargin=*}
\usepackage{xurl}

\newtheorem{theorem}{Theorem}[section]
\newtheorem{corollary}[theorem]{Corollary}
\newtheorem{lemma}[theorem]{Lemma}
\newtheorem{proposition}[theorem]{Proposition}

\theoremstyle{definition}
\newtheorem{example}[theorem]{Example}

\usepackage{natbib}
\bibpunct[, ]{(}{)}{,}{a}{}{,}%
\def\bibfont{\small}%
\usepackage{bm,comment}

\ECRepeatTheorems

\EquationsNumberedThrough    

\allowdisplaybreaks
\begin{document}


\RUNAUTHOR{}

\RUNTITLE{}

\TITLE{Optimal Selection with Balanced Market Share: Static and Dynamic Assortment Optimization}

\ARTICLEAUTHORS{Omar El Housni \quad \quad Qing Feng \quad \quad Huseyin Topaloglu \\ School of Operations Research and Information Engineering, Cornell Tech, Cornell University  \\ \EMAIL{\{oe46,qf48,ht88\}@cornell.edu} \URL{}}

\ABSTRACT{
Assortment optimization is a critical tool for online retailers aiming to maximize revenue. However, optimizing purely for revenue can lead to unbalanced sales across products, potentially causing a long tail of low-selling products and products with excessively large market shares, both of which could be harmful to the seller. To address these issues, we introduce a market share balancing constraint that limits the disparity in expected sales between any two offered products to a factor of a given parameter $\alpha$. We study both static and dynamic assortment optimization under the multinomial logit (MNL) model with this fairness constraint. In the static setting, the seller selects a distribution over assortments that satisfies the market share balancing constraint while maximizing expected revenue. We show that this problem can be solved in polynomial time, and we characterize the structure of the optimal solution: a product is included if and only if its revenue and preference weight exceed certain thresholds. We further extend our analysis to settings with additional feasibility constraints on the assortment and demonstrate that, given a $\beta$-approximation oracle for the constrained problem, we can construct a $\beta$-approximation algorithm under the fairness constraint.
In the dynamic setting, each product has a finite initial inventory, and the seller implements a dynamic policy to maximize total expected revenue while respecting both inventory limits and the market share balancing constraint in expectation. We design a policy that is asymptotically optimal, with its approximation ratio converging to one as inventories grow large.
}



\maketitle

%

\section{Introduction}

Assortment optimization is one of the fundamental problems in the field of revenue management. In this problem, the decision maker, hereafter referred to as the seller, chooses an assortment, i.e., a subset of products, from the entire universe of products to offer to customers in order to maximize a context-specific objective, such as revenue, profit or total market share. Assortment optimization finds its place in many practical business settings. For example, online retailers decide which products to display to customers in order to maximize their expected revenue. Online advertisers select an effective combination of advertisements in order to maximize user engagement and click-through rates. Motivated by the crucial role of assortment optimization in today's retailing industry, there has been a large stream of work that focuses on developing efficient algorithms that find the optimal or approximately optimal solution to assortment optimization problems.

{  While optimal assortment decisions can significantly help sellers maximize revenue, focusing exclusively on revenue maximization often causes significant imbalance in sales among different products. This is because the optimal assortment often includes a few highly popular products alongside many other less popular ones. Since popular products capture most customer attention, less popular items are often overshadowed and have limited chances to generate sales; as a result, such assortments can lead to highly unbalanced sales across products. Under such unbalanced sales, the seller would be faced with a few highly popular products attaining a large market share, along with a long tail of low-selling products. However, as we will discuss next, both a long tail of low-selling products and a small number of products with excessively large market shares can each lead to undesirable consequences for the seller in various ways.

First, having a long tail of low-selling products can negatively affect the seller, both from the logistics and supplier-relations perspective and from the customer’s perspective. From the aspect of logistics and supplier relationship, maintaining a long tail adds operational burden, requiring the seller to manage more products, hold and replenish inventory, and sustain a broader set of supplier relationships, which in turn increases operational costs. As noted in \cite{byrne2015the}, ``\textit{\dots long tails \dots represent a major financial commitment and disproportional cost burden for make-to-stock manufacturers. Too often, these items provide negative returns, undermining the margins of the entire portfolio}". Low-selling products also tend to exhibit larger demand uncertainty, leading to larger safety stocks needed for the seller and less accurate demand predictions, further worsening the revenue margins of these products. As also noted in \cite{byrne2015the}, ``{\it A finance executive looking at the capital costs would find that the bottom 50 percent of items require 23 times as much inventory per dollar of revenue than the top 10 percent. By looking at net profits, one would find that poor forecast accuracy led to lost revenue from stock-outs, and higher operating costs from transshipments and expedites, both which erode margins}". Furthermore, even if the seller is willing to offer certain low-selling products, it could be challenging for the seller to keep corresponding suppliers engaged in the long run. As shown in \cite{de2001competition}, low-selling products may leave the market due to various reasons, including competition with other products and being unable to cover fixed costs of selling on the platform. All these factors potentially make it challenging for the seller to maintain long-term engagement of suppliers with low-selling products, even if the seller is willing to do so. Finally, from the customer perspective, offering a long tail of low-selling products can also degrade the platform’s customer experience, as customers may become confused or overwhelmed when presented with a large number of low-selling items. For instance, it has been documented in \cite{mims2022the} that Amazon customers are getting increasingly confused and frustrated simply because an excessive long tail of products are offered on the platform.

On the other hand, having a few products with excessively large market shares can also be harmful to the seller. When a small number of products account for a large share of sales, the seller’s revenue becomes overly dependent on those items, leading to the undesirable situation of "revenue concentration". \cite{dealhub2025revenue} discussed several negative impacts of revenue concentration on the seller. One potential impact is that it increases the seller’s risk exposure and limits the ability to expand into new markets. As mentioned in \cite{dealhub2025revenue}, ``{\it If one product makes up the bulk of your sales, your entire business hinges on its continued success. A competitor’s innovation, shifting customer preferences, or a change in technology can instantly erode your market share. This kind of dependence also limits growth. Expanding into new markets or upselling existing customers is harder when you have only one core product to offer.}" Another drawback of revenue concentration is that it puts suppliers of high–market-share products in a strong bargaining position. As a result, the seller’s negotiating power is weakened, as noted by \cite{dealhub2025revenue}, ``{\it partners who represent a large share of your revenue gain leverage over your pricing and terms}". Having one single product or supplier aquiring a large market share also risks violation of the antitrust regulations. For example, in 2004, an antitrust investigation was triggered against Dentsply, a monopolist manufacturer-supplier of dental supplies, which dominates the market with a 75\%-80\% market share\footnote{https://www.justice.gov/atr/case-document/us-v-dentsply-brief-united-states-redacted}.

To address these issues, one effective approach is to introduce additional constraints that enforce balanced sales across the products offered by the seller. While there are multiple potential approaches to enforce such balancing constraints, in this paper, we adopt a {\bf market share balancing constraint}, requiring that the market share of any two offered products may differ by at most a factor of $\alpha$. With a single balancing parameter $\alpha$, the seller is able to maintain a relative balance in sales among offered products, without needing to tune a large number of parameters or know the absolute values of target market shares of products. The seller can manage the tradeoff between revenue and balance in market share by controlling the balancing parameter $\alpha$, where a smaller $\alpha$ brings a higher revenue, while a larger $\alpha$ maintains a stronger balance in market shares. We further numerically study the tradeoff between revenue and balance in market shares in Section~\ref{sec:numerical_tradeoff}, where we show that in practice, the seller can attain a relatively balanced market share without suffering significant revenue loss.

}

\subsection{Main Contributions}

In this paper, we study the assortment optimization problem under the market share balancing constraint. This constraint requires that the market shares of any two products offered by the seller differ by at most a factor of $\alpha$, where $0 < \alpha \leq 1$ is an exogenously specified balancing parameter. We consider both static and dynamic versions of the problem under this constraint. In the static setting, the seller selects a distribution over assortments that satisfies the market share balancing constraint and maximizes expected revenue. In the dynamic setting, the seller starts with a finite initial inventory for each product (without replenishment), and customers arrive and make purchase decisions sequentially over time. The seller must choose a dynamic control policy that maximizes total expected revenue while ensuring that the expected total sales across products satisfy the market share balancing constraint. To the best of our knowledge, this work is the first to introduce fairness constraints into dynamic assortment optimization.
In what follows, we formally define the two problems and present our main algorithmic contributions for each setting.

\vspace{2mm}
\noindent\textbf{Static Assortment Optimization.} We begin with the static assortment optimization problem under balanced market shares. In this setting, customers make purchase decisions according to a multinomial logit (MNL) model, and the seller selects a distribution over assortments that satisfies the market share balancing constraint. Specifically, under the induced distribution, the purchase probabilities of any two offered products must differ by at most a factor of $\alpha$.
We allow for randomized assortments, as they provides the seller with more flexibility, thus help the seller achieve higher revenue while maintaining a balanced market share. Given our focus on randomized assortments, a natural question arises: how much additional revenue can be gained through randomization compared to deterministically offering a single assortment? 
To answer this, we compare the optimal expected revenue of the randomized (static) problem with that of a deterministic variant, in which the seller must offer a single assortment subject to the same market share balancing constraint. Letting $n$ denote the number of products, we show that the optimal revenue in the randomized setting can be up to a factor of $\min\{n, 1/(1 - \alpha)\}$ larger than that of the deterministic variant. This indicates that the benefit of randomization can grow linearly in $n$ as $\alpha$ approaches one.

\noindent\underline{\it Algorithm for the Static Problem.} We develop a polynomial-time algorithm that exactly solves the static problem. Our approach is based on uncovering a novel structure in the optimal solution, characterized by two thresholds: one on revenue and one on preference weight. Specifically, we prove that the set of offered products is nested with respect to both revenue and preference weight, and a product is offered if and only if its revenue and preference weight both exceed certain thresholds (see Theorem~\ref{thm:randomized}). This structural result contrasts sharply with the classical unconstrained MNL assortment optimization problem, where the optimal solution is determined solely by a threshold on revenue.
Leveraging this structure, we solve the problem by enumerating all $O(n^2)$ possible threshold pairs and constructing a candidate solution for each. We also perform a sensitivity analysis of the thresholds with respect to $\alpha$, showing that the revenue threshold increases as $\alpha$ decreases, while the preference weight threshold may either increase or decrease.

\noindent\underline{\it Algorithm for the Static Problem with Additional Constraints.} We extend our framework to handle additional constraints on the set of products that may be offered. In this setting, in additional to the market share balancing constraint, the set of offered products must also conform to other feasibility constraints (e.g., capacity, budget, or structural limitations). We prove that if the classical MNL assortment optimization problem under the same type of constraints admits a $\beta$-approximation oracle, then the constrained version of our problem also admits a $\beta$-approximation algorithm (see Theorem~\ref{thm:randomized_constrained}).
To establish this result, we reduce our constrained static problem to a standard MNL assortment optimization problem with a modified set of preference weights. Applying the $\beta$-approximation oracle to this transformed problem yields a $\beta$-approximate solution for our original constrained setting.

\vspace{2mm}
\noindent\textbf{Dynamic Assortment Optimization.} In the dynamic setting, each product $i$ has an initial (non-replenishable) inventory $c_i$, and customers arrive sequentially over a time horizon of $T$ periods. Each customer's choice behavior is modeled by an MNL model. The seller's decision is a dynamic control policy that maps the sales history to a distribution over assortments at each time period. We impose two constraints: (i) the realized sales of product $i$ must not exceed $c_i$ with probability one, and (ii) the total expected sales across products by the end of the horizon must satisfy the market share balancing constraint.
Due to the high dimensionality of the state space, computing an exact optimal dynamic policy is computationally intractable. As a first step, we introduce an upper bound problem that provides a provable bound on the optimal value of the dynamic problem. We show that this upper bound problem is NP-hard and develop a fully polynomial-time approximation scheme (FPTAS) to solve it approximately.

 \noindent\underline{\it Algorithm for the Dynamic Problem.} We provide an asymptotically optimal policy for the dynamic problem, whose approximation ratio converges to one as the initial inventories get large. We construct a policy in the following form: Each product $i$ is associated with a purchase probability $\hat{x}_i$. In each time period, the seller chooses a distribution over assortments such that the purchase probability of product $i$ is $\hat{x}_i$ if the product has remaining inventory, or zero if the product has no remaining inventory. We establish an algorithm that constructs $\hat{\bm{x}}$ based on an approximate solution of the upper bound problem such that the policy satisfies the market share balancing constraint in expectation. 
 We show that our policy achieves strong approximation guarantees (see Theorem~\ref{thm:inventory_constant_1}). Specifically, for an FPTAS with precision parameter $\epsilon>0$, the policy guarantees a $(1/2 - \epsilon)$-approximation to the dynamic problem, and the approximation ratio approaches $1 - \epsilon$ as the initial inventories get large. 
 
 {  In Appendix~\ref{sec:joint_inventory_assortment}, we further consider the joint inventory stocking and dynamic assortment optimization problem with balanced market share, where besides sequentially offering products, the initial inventory stocking $\bm{c}$ is also part of the decision. Specifically, the seller first chooses the initial inventory $\bm{c}$ such that the total initial inventory $\sum_{i\in \mathcal{N}}c_i$ does not exceed a maximum capacity $K$, then chooses a dynamic policy of offering assortments that satisfies the inventory constraint and the market share balancing constraint. To solve the joint inventory and dynamic assortment optimization problem, we first establish an upper bound problem that provides an upper bound on its optimal value. We establish an FPTAS for the upper bound problem. Based on an approximate solution of the upper bound problem, we establish a policy that attains a constant fraction of the optimal expected revenue. We also establish an alternative policy that is asymptotically optimal, whose approximation ratio converges to one as the maximum capacity $K$ gets large.}

\vspace{2mm}
\noindent\textbf{Numerical Experiments.} {  We conduct numerical experiments using MNL models calibrated from real data. We compare our market share balancing constraint and the absolute market share balancing constraint, which enforces absolute lower and upper bounds on the market share of every offered product, and study how effective our approach is in maintaining a tradeoff between revenue and balanced market shares compared to the alternative approach. The results show that if we measure balance in market share by ratio, then our approach always generates a higher revenue than the alternative approach with absolute market share balancing constraint, often by a significant margin, while maintaining similar level of balance among market shares. Furthermore, if we measure balance in market shares using absolute difference, then our approach is still able to attain a higher or comparable revenue compared to the alternative approach while maintaining a similar level of balance, even though our approach is not optimized for this metric of balance. We further use numerical experiments to study how large the balancing parameter $\alpha$ can get if the seller is willing to lose a certain fraction of revenue, and how the maximum and minimum market shares change accordingly. Our numerical results show that by enforcing the market share balancing constraint, the seller can attain a relatively balanced market share without incurring significant revenue loss. One can also observe that enforcing the market share balancing constraint can effectively increase the minimum nonzero market share and reduce the maximum market share.}

{  Furthermore, in Appendix~\ref{sec:numerical},} we also conduct numerical experiments to evaluate the empirical performance of our policy for the dynamic problem using synthetically generated test instances. We compare our policy against two resolving heuristics that also enforce balanced market shares. These heuristics track the purchase probabilities over all previous time periods and, at each resolution step, solve an optimization problem in which the cumulative purchase probabilities must satisfy the market share balancing constraint. Both heuristics are guaranteed to satisfy the constraint by design.
Our results show that the proposed policy consistently outperforms both resolving heuristics across most instances. Furthermore, we observe that the performance of our policy improves as the initial product inventories increase, which aligns with our theoretical result that the policy is optimal.

\subsection{Related Literature}

Our work is related to three streams of literature: static assortment optimization problem under MNL, dynamic assortment optimization problem with resource constraints, and fairness in assortment planning.

\noindent\textbf{Static Assortment Optimization under MNL:} Static assortment optimization under MNL has been studied extensively in the revenue management literature. \cite{talluri2004revenue} are among the first to consider assortment optimization problem under MNL, and they show that the unconstrained optimal assortment under MNL is revenue-ordered. In other words, the optimal assortment contains every product whose revenue is above a certain threshold. Thus, the unconstrained assortment optimization problem under MNL can be solved in polynomial time by optimizing only over revenue-ordered assortments. \cite{rusmevichientong2010dynamic} show that the cardinality constrained assortment optimization problem under MNL can be solved in polynomial time. As an extension of \cite{rusmevichientong2010dynamic}, \cite{sumida2021revenue} further consider the assortment optimization under MNL with totally unimodular constraints, and show that the problem can be solved in polynomial time via an equivalent linear program. \cite{desir2022capacitated} consider the assortment optimization problem under MNL with capacity constraints. They show that the capacitated assortment optimization under MNL is NP-hard, and they provide an FPTAS for the problem. Other research works that consider assortment optimization problems under MNL include \cite{gao2021assortment}, \cite{el2023joint}, \cite{housni2023maximum}.

\noindent\textbf{Dynamic Assortment Optimization with Resource Constraints:} There is significant work on dynamic assortment optimization with resource constraints. In this problem, the seller sells products over multiple time periods with limited capacities of resources, and selling each product consumes capacities of a combination of resources. The problem is first considered by \cite{gallego1994optimal}, where they consider a single product pricing setting and show that if the expected demand and the resource capacity scale with the same rate $\theta$, then there exists a policy with a $(1-O(1/\sqrt{\theta}))$-approximation ratio. \cite{gallego1997multiproduct} and \cite{talluri1998analysis} generalize the results to the setting of multiple resources and multiple products, and the seller makes an accept or reject decision for each product request. \cite{gallego2004managing} and \cite{liu2008choice} are among the first to incorporate choice models into dynamic assortment optimization. Letting $c_{\min}$ be the smallest resource capacity, \cite{rusmevichientong2020dynamic}~establish a policy that achieves an approximation ratio of $1-O(1/\sqrt[3]{c_{\min}})$. \cite{ma2021dynamic}, \cite{bai2022coordinated}, \cite{feng2024near}, \cite{aouad2023nonparametric}, \cite{bai2023fluid}, \cite{li2024revenue} provide policies with an approximation ratio of $1-O(1/\sqrt{c_{\min}})$. \cite{golrezaei2014real} also consider an online assortment optimization problem with limited inventory for each product, but their main focus is to establish an algorithm with guarantees in competitive ratio in the setting where customer arrivals are adversarially given. In this paper, we consider enforcing the market share balancing constraint in dynamic assortment optimization with limited inventory for each product, and we establish a policy that attains an approximation ratio of $1-O(1/\sqrt{c_{\min}})$.

\noindent\textbf{Fairness in Assortment Planning:} In recent years there has been growing interest in assortment optimization problems with fairness constraints. Some early works in this line of research include \cite{chen2022fair}, \cite{lu2023simple}, \cite{barre2023assortment}, \cite{housni2024assortment} and \cite{zhu2024unified}. These works introduce multiple approaches to define fairness. \cite{chen2022fair} consider a randomized assortment optimization problem in which the difference in certain outcome (e.g., market share, revenue, customer exposure) between products with similar quality should not exceed an exogenously given upper bound. \cite{lu2023simple} and \cite{barre2023assortment} define fairness constraints as minimum customer exposure (also referred to as visibility constraints in \cite{barre2023assortment}). \cite{lu2023simple} consider a randomized assortment optimization problem in which the probability of each product being included in the assortment has to exceed a certain threshold. \cite{barre2023assortment} consider an assortment customization problem over a finite number of customers, where each product needs to be visible to at least a certain number of customers. As an extension, \cite{housni2024assortment} generalize the results by enforcing minimum customer exposure of potentially overlapping product categories instead of individual products. {  All aforementioned works enforce fairness constraints over the entire universe of products. One potential drawback of enforcing fairness constraints over all products is that, it often requires the seller to offer a large number of products, many of which could have a negative impact on the seller's revenue. This could further result in significant revenue loss for the seller, making the fairness constraint undesirable for the seller.} 

{  To the best of our knowledge, \cite{zhu2024unified} is the only prior work that considers enforcing fairness constraints only over offered products rather than the entire universe of products. However, different from our paper, they introduce fairness constraints by requiring the total market share of each product class to be either zero or above a certain threshold. We refer to Section~\ref{sec:model_discussion} for a comparison between our market share balancing constraint and a generalized version of the constraint in \cite{zhu2024unified}, which enforces absolute lower bounds and upper bounds on the market share of each product offered with positive probability.}

In this paper, we define fairness constraint as the market share of any two offered products differing by at most a factor of $\alpha$. {  In our setting, the seller has the flexibility of deciding what products to offer, and balanced market shares are guaranteed only within the set of offered products. Such flexibility of product selection marks the key difference between our approach with existing approaches such as \cite{chen2022fair}, \cite{lu2023simple} and \cite{barre2023assortment}. It makes our approach less restrictive compared to existing approaches, and the seller is able to attain a higher revenue while still enjoying the benefit of balanced market shares.} Furthermore, in contrast to fairness constraints defined in most of the existing works, which usually rely on a large number of parameters, the market share balancing constraint considered in this paper only depends on one parameter $\alpha$, potentially making our model easier to implement in practice. Finally, all aforementioned works on fairness in assortment planning focus on static assortment optimization. To the best of our knowledge, our work is the first to introduce balancing constraints to dynamic assortment optimization settings.

\section{Static Assortment Optimization Problem}

In this section, we formally introduce the static assortment optimization problem with balanced market share. {  Then we provide some discussions on our problem formulation.} Next, we quantify the value of randomization by comparing the optimal value of the static problem~with that of a deterministic variant of the static problem.

\subsection{Problem Formulation}\label{sec:definition_static}

Let $\mathcal{N}=\{1,2,\dots,n\}$ be the universe of products. Each product $i\in\mathcal{N}$ has a revenue $r_i>0$. The option of leaving without purchasing any product is represented as product zero, and referred to as the no-purchase option.

\noindent\textbf{MNL Choice Model:} We assume that customers make choices based on an MNL model. Under this model, each product $i\in\mathcal{N}$ has a preference weight $v_i>0$ capturing the attractiveness of the product. Without loss of generality, we normalize the preference weight of the no-purchase option as $v_0=1$. Under MNL, if the customer is offered assortment $S\subseteq\mathcal{N}$, then for all $i\in S$, the purchase probability of product $i$ is given by
\begin{equation*}
\phi(i,S)=\dfrac{v_i}{1+\sum_{j\in S}v_j}.
\end{equation*}
The probability of no-purchase is given by $\phi(0,S)=1/(1+\sum_{j\in S}v_j)$. For all $i\notin S$ we define $\phi(i,S)=0$. The expected revenue gained from a customer who is offered assortment $S$ is given by
\begin{equation}
R(S)=\sum_{i\in S}r_i\phi(i,S)=\dfrac{\sum_{i\in S}r_iv_i}{1+\sum_{i\in S}v_i}.\label{eq:expected_revenue}
\end{equation}

\noindent\textbf{Problem Definition:} In the static problem, the decision of the seller is to choose a distribution $q(\cdot)$ over assortments in $\mathcal{N}$. Whenever a customer arrives, the seller randomly draws an assortment $S$ with probability $q(S)$ and the assortment $S$ is offered to the customer. Thus the purchase probability of product $i$ is $\sum_{S\subseteq\mathcal{N}}\phi(i,S)q(S)$. The seller is given a balancing parameter $0<\alpha\leq 1$, and we require that the purchase probabilities of two products offered with positive probability differ by at most a factor of $\alpha$. In other words, the seller has the flexibility to select a subset of products to offer and guarantees a balanced market share only within selected products. As long as a product is selected by the seller to offer to customers, the constraint guarantees a purchase probability of the product that is at least $\alpha$ times the maximum purchase probability. We refer to this type of constraint as the {\it market share balancing constraint}, and the problem is referred to as the static assortment optimization problem with balanced market share, or \ref{prob:randomized_0} in short. The problem is formally defined as
\begin{equation}
\label{prob:randomized_0}
\tag{\texttt{BMS}}
\begin{aligned}
&\max_{q(\cdot)} && \sum_{S\subseteq\mathcal{N}}R(S)q(S)\\
&\textup{s.t.} && \sum_{S\subseteq\mathcal{N}}q(S)=1,\\
&&&q(S)\geq 0,\ \forall S\subseteq\mathcal{N},\\
&&&\sum_{S\subseteq\mathcal{N}}\phi(i,S)q(S)\in\{0\}\cup\Big[ \alpha\cdot\max_{j\in\mathcal{N}}\sum_{S\subseteq\mathcal{N}}\phi(j,S)q(S),\infty\Big),\ \forall i\in\mathcal{N}.
\end{aligned}
\end{equation}
Here, $R(S)$ is the expected revenue by offering assortment $S$ given by~\eqref{eq:expected_revenue}, and $q(\cdot)$ is a probability mass function over assortments. The validity of $q(\cdot)$ is guaranteed by the first and second constraints of~\ref{prob:randomized_0}. The last set of constraints is the market share balancing constraint, requiring that the purchase probability of each product should be either zero or at least $\alpha$ of the maximum purchase probability over all products. In this way, we ensure that the market shares of any two products offered with positive probability differ by at most a factor of $\alpha$. 

If the seller deterministically offers any assortment $S\subseteq\mathcal{N}$, then for any $i,j\in S$, the ratio of purchase probability of product $i$ over that of product $j$ is equal to $\phi(i,S)/\phi(j,S)=v_i/v_j\geq v_{\min}/v_{\max}$, where $v_{\min}$ and $v_{\max}$ are the minimum and maximum preference weights respectively. Therefore when $\alpha\leq v_{\min}/v_{\max}$, deterministically offering any assortment is feasible to~\ref{prob:randomized_0}, thus in this case~\ref{prob:randomized_0}~is reduced to the unconstrained assortment optimization problem under MNL. As $\alpha$ increases, we gradually strengthen the market share balancing constraint, and when $\alpha=1$, we are enforcing the strongest market share balancing constraint, requiring all offered products to have exactly the same purchase probability.

We present an equivalent sales-based reformulation of~\ref{prob:randomized_0}. The reformulation is called sales-based because it uses purchase probabilities of different products as decision variables instead of offer probabilities of different assortments. For each $i\in\mathcal{N}$, we use $x_i$ to capture the purchase probability of product $i$, and we use $x_0$ to capture the probability of no-purchase. By definition the corresponding objective function is $\sum_{i\in\mathcal{N}}r_ix_i$, and the purchase probabilities $\bm{x}$ should satisfy $x_0+\sum_{i\in\mathcal{N}}x_i=1$. Furthermore, since $\phi(i,S)\leq v_i\phi(0,S)$ for any $i\in\mathcal{N}$ and any $S\subseteq\mathcal{N}$, under any distribution over assortments the purchase probability of product $i$ is at most $v_i$ times the probability of no-purchase. Therefore we enforce $0\leq x_i\leq v_ix_0$ for all $i\in\mathcal{N}$ to ensure that the purchase probabilities $\bm{x}$ are valid under MNL. Finally, by the market share balancing constraint, we further have that for all $i\in\mathcal{N}$, $x_i$ should be either zero or at least $\alpha\cdot\max_{j\in\mathcal{N}}x_j$. Combining the discussions above we obtain the sales-based reformulation of~\ref{prob:randomized_0}, which is formally defined in~\ref{prob:randomized}, given by
\begin{equation}\label{prob:randomized}
\tag{\texttt{BMS-Sales}}
\begin{aligned}
&\max_{\bm{x}} && \sum_{i\in\mathcal{N}}r_ix_i\\
&\textup{s.t.} && x_0+\sum_{i\in\mathcal{N}}x_i=1,\\
&&&0\leq x_i\leq v_ix_0,\ \forall i\in\mathcal{N},\\
&&&x_i\in\{0\}\cup\Big[\alpha\cdot\max_{j\in\mathcal{N}}x_j,\infty\Big),\ \forall i\in\mathcal{N}.
\end{aligned}
\end{equation}
{  We remark that similar sales-based formulation has also been considered in existing works such as \cite{topaloglu2013joint} and \cite{gallego2015general}.} For any feasible solution $q(\cdot)$ of~\ref{prob:randomized_0}, by setting $x_i=\sum_{S\subseteq\mathcal{N}}q(S)\phi(i,S)$ for all $i\in\mathcal{N}\cup\{0\}$, we obtain a feasible solution $\bm{x}$ to~\ref{prob:randomized} with the same objective value. Furthermore, given any feasible solution $\bm{x}$ of \ref{prob:randomized}, we are able to construct a distribution over assortments $q(\cdot)$ feasible to \ref{prob:randomized_0} with the same objective. Suppose without loss of generality we index the products as $x_1/v_1\geq x_2/v_2\geq\dots\geq x_n/v_n$. Using the assortments $S_i=\{1,2,\dots,i\}$ and $S_0=\varnothing$, we define a distribution over assortments $q(\cdot)$ based on $\bm{x}$ as follows. We set $q(S)=0$ for all $S\notin\{S_0,S_1,\dots,S_n\}$. Defining $v_0=1$ and $x_{n+1}=0$ for notational uniformity, for all $i\in\{0,1,\dots,n\}$, we define $q(S_i)$ as
\begin{equation}\label{eq:sales_to_distribution}
q(S_i)=\Big(\dfrac{x_i}{v_i}-\dfrac{x_{i+1}}{v_{i+1}}\Big)\Big(1+\sum_{j\in S_i}v_j\Big).
\end{equation}
As shown in Theorem 1 of \cite{topaloglu2013joint}, if the purchase probabilities $\bm{x}$ satisfy the first two constraints in \ref{prob:randomized}, then the distribution over assortments $q(\cdot)$ defined in \eqref{eq:sales_to_distribution} is able to attain the purchase probabilities $\bm{x}$. Therefore we conclude that~\ref{prob:randomized}~is equivalent to~\ref{prob:randomized_0}. Note that \ref{prob:randomized} is not a linear program because of the non-linearity in the last set of constraints. In the forthcoming Section \ref{sec:randomized}, we will introduce an algorithm that solves \ref{prob:randomized} in polynomial time.

{ 

\subsection{Model Discussion}\label{sec:model_discussion}

Besides our approach of enforcing market share balancing constraint, an alternative constraint is to impose an absolute upper bound and lower bound on the market share of each product offered by the seller. We refer to this constraint as the {\it absolute market share balancing constraint}. The absolute market share balancing constraint is an extension to the constraint considered in \cite{zhu2024unified}, which requires that the market share of any product offered by the seller should exceed a certain lower bound. While it is difficult to argue that one approach uniformly dominates the other, our approach has the following advantages:

\begin{enumerate}

\item {\bf Single parameter:} Our market share balancing constraint depends on one single parameter $\alpha$. The seller can easily control the tradeoff between revenue and balance in market shares by controlling the parameter $\alpha$. On the other hand, if the seller enforces both an absolute lower bound and upper bound, the constraint will depend on more parameters, making it harder for the seller to tune these parameters in order to maintain a tradeoff between revenue and balance in market shares.

\item {\bf No need for knowledge of demand:} Due to uncertainty in demand and market size, in practice it is hard to accurately estimate the sales volume of each product beforehand. If the seller would like to maintain balanced sales via enforcing absolute lower and upper bounds on market shares, it can be challenging to set appropriate absolute lower and upper bounds on market shares beforehand. Our market share balancing constraint, on the other hand, guarantees relative balance among market shares, which we believe to be easier to control via the parameter $\alpha$ and does not require the seller to know the appropriate market share lower and upper bounds.

\item {\bf Adaptivity to market changes:} In practice, the product universe available to the seller may change from time to time, with either new products entering the market or existing products leaving the market. Such market changes could further result in changes in market shares. Since our market share balancing constraint focuses on relative balance in market shares, even if such market changes occur, the seller can still maintain the same level of relative balance by using the same $\alpha$. On the other hand, if the seller enforces absolute lower bounds and upper bounds on market shares, then as new products enter the market or existing products leave, the market shares may change and previous lower bounds and upper bounds may not be valid, thus the seller may need to accordingly change the market share lower and upper bounds whenever market changes occur.

\item {\bf Tractability of the assortment problem:} As shown in Theorem \ref{thm:randomized}, our problem with market share balancing constraint is polynomial time solvable. Furthermore, as shown in Theorem \ref{thm:randomized_constrained}, even if we enforce additional constraints of the set of offered products, we can still achieve the same approximation ratio as the corresponding standard constrained assortment optimization problem. On the other hand, \cite{zhu2024unified} show that even if the seller only enforces a lower bound on the market share of each offered products, the problem is already NP-hard. If the seller enforces additional constraints on the set of offered products, the problem can be more challenging or even infeasible.

\end{enumerate}

We also refer to Section~\ref{sec:numerical_real} for numerical experiments comparing the two approaches using real-world data.

}

\subsection{Value of Randomization}\label{sec:value_randomization}

We study the value of randomization by comparing the optimal value of~\ref{prob:randomized_0}~with that of a deterministic variant of the static problem. In the deterministic variant, the seller deterministically offers a single assortment such that the market share balancing constraint is satisfied, without the flexibility of randomizing over assortments as in~\ref{prob:randomized_0}. In other words, we enforce $q(S)\in\{0,1\}$ for all $S\subseteq\mathcal{N}$ in the deterministic variant. We refer to the deterministic variant of the static problem as \texttt{BMS-Deterministic}.

We show that the gap between the optimal value of~\ref{prob:randomized_0}~and that of~\texttt{BMS-Deterministic}~is at the order of $\min\{1/(1-\alpha),n\}$.

\vspace{0.5em}

\begin{proposition}\label{thm:value_randomization}
The optimal value $R^*$ of~\ref{prob:randomized_0}~and the optimal value $R^*_{det}$ of~\texttt{BMS-Deterministic}~satisfy
\begin{equation}
R^*\leq \min\Big\{\dfrac{2}{1-\alpha},n\Big\}R^*_{det}.\label{eq:gap_upper_bound}
\end{equation}
Furthermore, for any $0<\alpha\leq 1$ and any $n\in\mathbb{Z}^+$, there exists an instance with $n$ products such that
\begin{equation}
R^*\geq \min\Big\{\dfrac{1}{2(1-\alpha)},\dfrac{n}{2}\Big\}R^*_{det}.\label{eq:gap_lower_bound}
\end{equation}
\end{proposition}

\vspace{0.5em}

Proof of Proposition~\ref{thm:value_randomization} is provided in Appendix~\ref{sec:thm:value_randomization}. We remark that the upper bound provided in \eqref{eq:gap_upper_bound} matches the lower bound provided in \eqref{eq:gap_lower_bound} up to a constant, thus the the gap between optimal value of~\ref{prob:randomized_0}~and that of~\texttt{BMS-Deterministic}~is at the order of $\min\{1/(1-\alpha),n\}$. When $\alpha$ is close to one, the gap is at the order of $n$. As we relax the market share balancing constraint by decreasing $\alpha$, we will also observe smaller gaps between~\ref{prob:randomized_0}~and~\texttt{BMS-Deterministic}.

Another important aspect of randomized assortments is the number of assortments the solution randomizes over, as ideally the seller would like the randomized solution to only randomize over a small number of assortments for ease of implementation. While the optimal solution to~\ref{prob:randomized_0}~could theoretically randomize over an exponential number of assortments, there exists an optimal solution to~\ref{prob:randomized_0}~that randomizes over at most $n$ assortments. The reason for this is that given a set of purchase probabilities $\bm{x}$ that satisfies the first two constraints in \ref{prob:randomized}, a distribution over assortments given by~\eqref{eq:sales_to_distribution} attains purchase probabilities $\bm{x}$ and randomizes over a nested sequence of assortments. Therefore if we use the aforementioned approach to construct a distribution over assortments, the distribution randomizes over a nested sequence of assortments, which consists of at most $n$ assortments. We also provide an example to show that if we follow the aforementioned approach to construct distribution over assortments from purchase probabilities, then for any $\alpha>0$, there exists an instance such that the optimal solution to~\ref{prob:randomized_0}~randomizes over $n$ assortments. Details of the example is provided in Example~\ref{example:value_randomization}~of Appendix~\ref{sec:example:value_randomization}.

\section{Algorithms for the Static Problem}\label{sec:randomized}

In this section, we introduce algorithms for the static problem. As the main contribution of this section, we design a polynomial time algorithm for \ref{prob:randomized_0}. As an extension, we further consider the static problem with additional constraints on the set of offered products, where besides market share balancing constraints, the set of offered products should also satisfy certain additional constraints.

\subsection{Algorithm for~\ref{prob:randomized_0}}\label{sec:alg_randomized}

We prove that~\ref{prob:randomized_0}~can be solved exactly in polynomial time. Since we have shown that \ref{prob:randomized_0}~is equivalent to~\ref{prob:randomized}, we mainly focus on~\ref{prob:randomized}~in this section.

We use $R^*$ to denote the optimal value of~\ref{prob:randomized}. In what follows, we present the main result of this section. We prove that there exists an optimal solution to~\ref{prob:randomized}~such that the products selected to offer with positive probability are nested in revenue and preference weight.

\vspace{0.5em}

\begin{theorem}\label{thm:randomized}
For some $\underline{v}\in\{v_i:i\in\mathcal{N}\}$, there exists an optimal solution $\bm{x}^*$ to~\ref{prob:randomized} such that
\begin{equation*}
\{i\in\mathcal{N}:x_i^*>0\}=\{i\in\mathcal{N}:r_i\geq R^*,\,v_i\geq \underline{v}\}.
\end{equation*}
\end{theorem}

\vspace{0.5em}

By Theorem~\ref{thm:randomized}, under the optimal solution to~\ref{prob:randomized}, a product is selected to offer with positive probability if both its revenue and preference weight are above certain thresholds. The result shows that there are at most $n^2$ candidates for the set of products offered with positive probability. Fixing the set of products that are offered with positive probability, \ref{prob:randomized}~is reduced to a linear program. Thus for each possible set of products offered with positive probability, we are able to construct a candidate solution by solving the corresponding linear program. By returning the candidate solution with maximum objective value, we obtain an optimal solution to~\ref{prob:randomized}. The structure of the optimal solution to~\ref{prob:randomized_0}, determined by two thresholds on revenue and preference weight respectively, is novel in the assortment optimization literature. The structure is in sharp contrast with that of the unconstrained optimal assortment under MNL, which is determined only by a threshold on revenue.

Intuitively, the threshold on revenue is to guarantee that no product offered with positive probability can harm the overall expected revenue. If $r_i<R^*$ for some product $i\in\mathcal{N}$, then product $i$ will not be offered with positive probability in the optimal solution. Otherwise, by setting $x_i=0$ and increasing all other components of $\bm{x}$ proportionally so that the sum of these components equals one, we obtain another feasible solution with strictly larger objective value. The threshold on preference weight is to make sure that each offered product is able to maintain a large enough purchase probability, and avoid posing a strong limit on the maximum purchase probability due to market share balancing constraint. If the preference weight $v_i$ of some product $i\in\mathcal{N}$ is very small, it is likely that the product will not be offered even if its revenue $r_i$ is large. The reason for this is that with $v_i$ being very small, the purchase probability of product $i$ will also be very small even if the product is offered with probability one. The small but nonzero purchase probability of product $i$ could further enforce all other products to have very small purchase probabilities, since under market share balancing constraint purchase probabilities of the offered products can differ by a factor of at most $\alpha$. Eventually this could lead to suboptimal solutions due to small purchase probabilities of all products.

To prove Theorem~\ref{thm:randomized}, we first establish an equivalent form of \ref{prob:randomized}. Recall that $R^*$ is the optimal value of \ref{prob:randomized}. We consider the problem

\begin{equation}\label{prob:randomized_linearized}
\begin{aligned}
&\max_{\bm{w}} && \sum_{i\in\mathcal{N}}(r_i-R^*)w_i\\
&\text{s.t.} && 0\leq w_i\leq v_i,\ \forall i\in\mathcal{N},\\
&&& w_i\in\{0\}\cup\Big[\alpha\cdot\max_{j\in\mathcal{N}}w_j,\infty\Big),\ \forall i\in\mathcal{N}.
\end{aligned}
\end{equation}

In Problem~\eqref{prob:randomized_linearized}, the decision variable $\bm{w}\in\mathbb{R}^n$ can be interpreted as $w_i=x_i/x_0$ for all $i\in\mathcal{N}$. The first set of constraints in Problem~\eqref{prob:randomized_linearized}~is equivalent to the second set of constraints in~\ref{prob:randomized}, and last set of constraints in Problem~\eqref{prob:randomized_linearized}~is equivalent to the market share balancing constraint, i.e., the last set of constraints in~\ref{prob:randomized}. Given an optimal solution $\bm{w}^*$ to Problem~\eqref{prob:randomized_linearized}, we reconstruct $\bm{x}^*$ as
\begin{equation}
x_i^*=\dfrac{w_i^*}{1+\sum_{i\in\mathcal{N}}w_i^*},\ \forall i\in\mathcal{N},\ x_0^*=\dfrac{1}{1+\sum_{i\in\mathcal{N}}w_i^*}.\label{eq:reconstruct_x}
\end{equation}

We first use the following lemma to verify that Problem~\eqref{prob:randomized_linearized}~is equivalent to~\ref{prob:randomized}.

\vspace{0.5em}

\begin{lemma}\label{lemma:randomized_linearized}
The optimal value of Problem~\eqref{prob:randomized_linearized}~is $R^*$. Furthermore, letting $\bm{w}^*$ be an optimal solution to Problem~\eqref{prob:randomized_linearized}, and defining $\bm{x}^*$ using Equation~\eqref{eq:reconstruct_x}, $\bm{x}^*$ is an optimal solution to~\ref{prob:randomized}.
\end{lemma}

\vspace{0.5em}

Proof of Lemma~\ref{lemma:randomized_linearized}~is provided in Appendix~\ref{sec:lemma:randomized_linearized}. With the help of Lemma~\ref{lemma:randomized_linearized}, we are able to complete the proof of Theorem~\ref{thm:randomized}. We first prove that for some $\underline{v}\in\{v_i:i\in\mathcal{N}\}$, there exists an optimal solution $\bm{w}^*$ to Problem~\eqref{prob:randomized_linearized} such that
\begin{equation}\label{eq:optimal_w}
w_i^*=\begin{cases}
\min\{v_i,\underline{v}/\alpha\},\ &\text{if}\ r_i\geq R^*\ \text{and}\ v_i\geq \underline{v},\\
0,\ &\text{otherwise}.
\end{cases}
\end{equation}
Then by reconstructing $\bm{x}^*$ from $\bm{w}^*$ using~\eqref{eq:reconstruct_x}, we verify that the constructed $\bm{x}^*$ is optimal to~\ref{prob:randomized}~by Lemma~\ref{lemma:randomized_linearized}~and satisfies the conditions stated in the theorem.

\begin{proof}[Proof of Theorem~\ref{thm:randomized}]
Consider any optimal solution $\hat{\bm{w}}$ to~Problem~\eqref{prob:randomized_linearized}. Suppose there exists $i_0\in\mathcal{N}$ such that $r_{i_0}<R^*$ and $\hat{w}_{i_0}>0$, then we define $\bm{w}'\in\mathbb{R}^n$ as $w_{i_0}'=0$ and $w_i'=\hat{w}_i$ for all $i\in\mathcal{N}\backslash\{i_0\}$. Then we have $w_i'\leq \hat{w}_i\leq v_i$ for all $i\in\mathcal{N}$. Furthermore, for any $i\in\mathcal{N}$ such that $w_i'>0$, we have $$w_i'=\hat{w}_i\geq \alpha\cdot\max_{j\in\mathcal{N}}\hat{w}_j\geq \alpha\cdot\max_{j\in\mathcal{N}}w_j'.$$
Therefore $\bm{w}'$ is feasible to~Problem~\eqref{prob:randomized_linearized}. Since $w_{i_0}'=0$, we have
\begin{equation*}
\sum_{i\in\mathcal{N}}(r_i-R^*)w'_i=\sum_{i\neq i_0}(r_i-R^*)\hat{w}_i>\sum_{i\in\mathcal{N}}(r_i-R^*)\hat{w}_i,
\end{equation*}
where the last inequality is due to the assumption that $r_{i_0}<R^*$ and $\hat{w}_i>0$. This contradicts the optimality of $\hat{\bm{w}}$. Therefore $\hat{w}_i=0$ for all $i\in\mathcal{N}$ such that $r_i<R^*$.

We define $\underline{v}=\min\{v_i:\hat{w}_i>0\}$. By definition we have $\underline{v}\in\{v_i:i\in\mathcal{N}\}$. We further define $\bm{w}^*$ as~\eqref{eq:optimal_w}. By~\eqref{eq:optimal_w} we have $w_i^*\in\{0\}\cup[\underline{v},\underline{v}/\alpha]$ and $w_i^*\leq v_i$ for all $i\in\mathcal{N}$. Therefore for any $i\in\mathcal{N}$ such that $w_i^*>0$, we have
\begin{equation*}
w_i^*\geq \underline{v}=\alpha\cdot\underline{v}/\alpha\geq \alpha\cdot\max_{j\in\mathcal{N}}w_j^*,
\end{equation*}
and we conclude that $\bm{w}^*$ is feasible to~Problem~\eqref{prob:randomized_linearized}. 

We claim that $w_i^*\geq \hat{w}_i$ for all $i\in\mathcal{N}$. Consider any $i\in\mathcal{N}$. If $w_i^*=0$, then either $r_i<R^*$ or $v_i<\underline{v}$. In the former case we have proven that $\hat{w}_i=0$. In the latter case, by definition of $\underline{v}$ we have $v_i<\min\{v_j:\hat{w}_j>0\}$, which implies $i\notin \{j\in\mathcal{N}:\hat{w}_j>0\}$, or equivalently $\hat{w}_i=0$. Therefore $w_i^*=0$ implies $\hat{w}_i=0$. If $w_i^*>0$, since by definition $w_i^*=\min\{v_i,\underline{v}/\alpha\}$ in this case, it suffices to prove that $v_i\geq \hat{w}_i$ and $\underline{v}/\alpha\geq \hat{w}_i$. Since $\hat{w}_i$ is feasible to~Problem~\eqref{prob:randomized_linearized}, we have $\hat{w}_i\leq v_i$. We also have
\begin{equation*}
\hat{w}_i\leq \max_{j\in\mathcal{N}}\hat{w}_j\leq \min\{\hat{w}_i:\hat{w}_i>0\}/\alpha\leq \min\{v_i:\hat{w}_i>0\}/\alpha=\underline{v}/\alpha,
\end{equation*}
where the second inequality is due to the second set of constraints of~Problem~\eqref{prob:randomized_linearized}, and the third inequality is due to the first set of constraints of~Problem~\eqref{prob:randomized_linearized}. Therefore in this case we also have $w_i^*=\min\{v_i,\underline{v}/\alpha\}\geq \hat{w}_i$. Then we conclude that $w_i^*\geq \hat{w}_i$ for all $i\in\mathcal{N}$. Since $w_i^*=\hat{w}_i=0$ for all $i\in\mathcal{N}$ such that $r_i<R^*$, we get
\begin{equation*}
\sum_{i\in\mathcal{N}}(r_i-R^*)w_i^*=\sum_{i:\,r_i\geq R^*}(r_i-R^*)w_i^*
\geq\sum_{i:\,r_i\geq R^*}(r_i-R^*)\hat{w}_i=\sum_{i\in\mathcal{N}}(r_i-R^*)\hat{w}_i=R^*.
\end{equation*}
Therefore $\bm{w}^*$ is also an optimal solution to~Problem~\eqref{prob:randomized_linearized}. By Lemma~\ref{lemma:randomized_linearized}, an optimal solution $\bm{x}^*$ to~\ref{prob:randomized}~can be reconstructed by~\eqref{eq:reconstruct_x}, and $\bm{x}^*$ satisfies
\begin{equation*}
\{i\in\mathcal{N}:x_i^*>0\}=\{i\in\mathcal{N}:w_i^*>0\}=\{i\in\mathcal{N}:r_i\geq R^*,\,v_i\geq \underline{v}\}.\qedhere
\end{equation*}
\end{proof}

By Theorem~\ref{thm:randomized}, we have the following corollary.

\vspace{0.5em}

\begin{corollary}\label{cor:optimal_x}
For some $\underline{v}\in\{v_i:{  i}\in\mathcal{N}\}$, there exists an optimal solution to~\ref{prob:randomized}~such that $x_i^*=w_i^*/(1+\sum_{j\in\mathcal{N}}w_i^*)$ and $x_0^*=1/(1+\sum_{j\in\mathcal{N}}w_j^*)$, and $\bm{w}^*$ is given by Equation~\eqref{eq:optimal_w}.
\end{corollary}

\vspace{0.5em}

Because $\bm{w}^*$ has a closed-form expression given by \eqref{eq:optimal_w}, Corollary~\ref{cor:optimal_x}~provides a closed-form expression for an optimal solution to~\ref{prob:randomized} depending on $R^*$ and $\underline{v}$. The definition of~\eqref{eq:optimal_w}~remains unchanged if we replace $R^*$ with \mbox{$\underline{r}=\min\{r_i:r_i\geq R^*\}$}. To identify the true $(\underline{r},\underline{v})$, we simply enumerate all $O(n^2)$ possible pairs of $(\hat{r},\hat{v})\in\{(r_i,v_j):i,j\in\mathcal{N}\}$, then for each candidate $(\hat{r},\hat{v})$, we construct a candidate solution $\bm{x}(\hat{r},\hat{v})$ following Corollary~\ref{cor:optimal_x}, and finally return the candidate solution that has the maximum expected revenue. Following the aforementioned steps we establish a polynomial-time algorithm for~\ref{prob:randomized}.

\noindent\textbf{Sensitivity Analysis with Respect to $\alpha$:} We conduct a sensitivity analysis on the thresholds for revenue and preference weights with respect to $\alpha$. By Theorem~\ref{thm:randomized}, we conclude that the threshold for revenue always increases as $\alpha$ decreases. This is because the threshold for revenue is the optimal value $R^*$ of \ref{prob:randomized}, and as $\alpha$ decreases the market share balancing constraint is relaxed, thus the optimal value $R^*$ increases. However, as $\alpha$ decreases the threshold on preference weight $\underline{v}$ can either increase or decrease. We show this by the following example.

\begin{example}
Suppose $\mathcal{N}=\{1,2,3\}$. We pick $\epsilon\in(0,1/270)$, and define
\begin{equation*}
v_1=\epsilon,\ v_2=1,\ v_3=3,\ r_1=\dfrac{7}{9\epsilon},\ r_2=r_3=1.
\end{equation*}
The unconstrained optimal assortment in this instance is $\{1,2,3\}$, and the unconstrained optimal expected revenue is
\begin{equation*}
\dfrac{\sum_{i=1}^3 r_iv_i}{1+\sum_{i=1}^3 v_i}=\dfrac{7/9+1+3}{1+\epsilon+1+3}<1.
\end{equation*}
Therefore under any $0<\alpha\leq 1$, the optimal value $R^*$ of the static problem is strictly smaller than one, and also strictly smaller than the revenue of any product. Then by Theorem~\ref{thm:randomized}, the set of products offered in the optimal solution is nested in preference weight. For any $v,\alpha$, we define $\tilde{R}(v,\alpha)$ as
\begin{equation*}
\tilde{R}(v,\alpha)=\dfrac{\sum_{i:v_i\geq v}r_i\min\{v_i,v/\alpha\}}{1+\sum_{i:v_i\geq v}\min\{v_i,v/\alpha\}}.
\end{equation*}
Since we have proven that $R^*<r_i$ for all $i\in\mathcal{N}$, by Corollary \ref{cor:optimal_x} the optimal value of the static problem under any $\alpha\in(0,1]$ is $\max_{i\in\mathcal{N}}\tilde{R}(v_i,\alpha)$, and the threshold on preference weight is equal to $\underline{v}=\arg\max_{v\in\{v_i:i\in\mathcal{N}\}}\tilde{R}(v,\alpha)$. If $\alpha=1$, we have
\begin{equation*}
\tilde{R}(v_1,1)=\dfrac{7/9+2\epsilon}{1+3\epsilon}>\dfrac{7/9}{1+3\epsilon}>\dfrac{3}{4},\ \tilde{R}(v_2,1)=\dfrac{1+1}{1+1+1}=\dfrac{2}{3},\ \tilde{R}(v_3,1)=\dfrac{3}{1+3}=\dfrac{3}{4}.
\end{equation*}
Therefore when $\alpha=1$, we have $\underline{v}=v_1=\epsilon$. If $\alpha=1/3$, we have
\begin{equation*}
\tilde{R}(v_1,1/3)=\dfrac{7/9+6\epsilon}{1+7\epsilon}<\dfrac{7}{9}+6\epsilon<\dfrac{4}{5},\ \tilde{R}(v_2,1/3)=\dfrac{1+3}{1+1+3}=\dfrac{4}{5},\ \tilde{R}(v_3,1/3)=\dfrac{3}{1+3}=\dfrac{3}{4}.
\end{equation*}
Therefore when $\alpha=1/3$, we have $\underline{v}=v_2=1$. If $\alpha=\epsilon/3$, then deterministically offering the unconstrained optimal assortment $\{1,2,3\}$ becomes feasible, thus in this case $\underline{v}=\epsilon$. Therefore when $\alpha$ decreases from one to $1/3$ then to $\epsilon/3$, the threshold on preference weight $\underline{v}$ changes from $\epsilon$ to one then to $\epsilon$ again. Thus we conclude that as $\alpha$ decreases, the threshold on preference weight $\underline{v}$ can either increase or decrease.
\end{example}

\vspace{0.5em}

Before we proceed, we would like to add a remark that the deterministic variant of the static problem~\texttt{BMS-Deterministic}~is much easier than \ref{prob:randomized_0} and can also be solved in polynomial time. Specifically, when the seller deterministically offers a single assortment $S$, the market share balancing constraint is equivalent to $\min_{i\in S}v_i / \max_{i\in S}v_i\geq \alpha$. Letting $S^*$ be an optimal solution to~\texttt{BMS-Deterministic}, the problem can be solved by enumerating all possible values of $\underline{v}=\min_{i\in S^*}v_i$. For each possible value of $\underline{v}$,  we solve an unconstrained assortment optimization problem where the universe of products is restricted to $\{i\in\mathcal{N}:\underline{v}\leq v_i\leq \underline{v}/\alpha\}$, and finally we return the assortment with the highest expected revenue over all enumerations. 

\subsection{Static Problem with Additional Constraints}\label{sec:randomized_constrained}

Besides market share balancing constraint, in practice sellers are often faced with additional constraints on the set of offered products. For example, the seller may have limited capacity that allows keeping only a certain number of products in stock. In this case the seller may require that at most $K$ products are offered with positive probability. {  Another example is that, due to considerations of product diversity, the seller may require that at least $K$ products should be offered to customers, and in this case the seller may require that at least $K$ products are offered with positive probability. More broadly, the seller may categorize products into multiple product categories $C_1,C_2,\dots,C_K$, and in order to in order to guarantee diversity in the product portfolio and better meet different customers' needs, the seller would like to guarantee a minimum number of products from each category to be offered to customers. In this case, the seller may require that for each product category $C_k\subseteq\mathcal{N}$, at least $\ell_k$ products from $C_k$ are offered with positive probability.}

Motivated by the examples above, in this section we consider the setting where in addition to the market share balancing constraint, we further require that the set of products offered with positive probability should belong to some $\mathcal{X}\subseteq 2^{\mathcal{N}}$. Here we would like to emphasize that different from most existing literature on constrained randomized assortment optimization, we impose constraint on products offered with positive probability rather than assortments randomized over. We refer to the problem as~\ref{prob:randomized_constrained_0} (the ``C" in the name stands for contraints), and its formal definition is given by
\begin{equation}
\label{prob:randomized_constrained_0}
\tag{\texttt{BMS-C}}
\begin{aligned}
&\max_{q(\cdot)} && \sum_{S\subseteq\mathcal{N}}R(S)q(S)\\
&\textup{s.t.} && \sum_{S\subseteq\mathcal{N}}q(S)=1,\\
&&&q(S)\geq 0,\ \forall S\subseteq\mathcal{N},\\
&&&\sum_{S\subseteq\mathcal{N}}\phi(i,S)q(S)\in\{0\}\cup\Big[\alpha\cdot\max_{j\in\mathcal{N}}\sum_{S\subseteq\mathcal{N}}\phi(j,S)q(S),\infty\Big),\ \forall i\in\mathcal{N},\\
&&&\Big\{i\in\mathcal{N}:\sum_{S\subseteq\mathcal{N}}\phi(i,S)q(S)>0\Big\}\in\mathcal{X}.
\end{aligned}
\end{equation}
Here, the first three constraints follow from~\ref{prob:randomized_0}. In the last constraint, a product is offered if and only if its purchase probability $\sum_{S\subseteq\mathcal{N}}\phi(i,S)q(S)$ is strictly positive, thus the last constraint enforces the set of offered products to belong to $\mathcal{X}$.

Similar to~\ref{prob:randomized}, by using the purchase probabilities $\bm{x}$ as the decision variables, we obtain the following equivalent sales based reformulation of~\ref{prob:randomized_constrained_0}. The problem is referred to as \ref{prob:randomized_constrained}, and its formal definition is given by

\begin{equation}\label{prob:randomized_constrained}
\tag{\texttt{BMS-C-Sales}}
\begin{aligned}
\begin{aligned}
&\max_{\bm{x}} && \sum_{i\in\mathcal{N}}r_ix_i\\
&\textup{s.t.} && x_0+\sum_{i\in\mathcal{N}}x_i=1,\\
&&&0\leq x_i\leq v_ix_0,\ \forall i\in\mathcal{N},\\
&&&x_i\in\{0\}\cup\Big[\alpha\cdot\max_{j\in\mathcal{N}}x_j,\infty\Big),\ \forall i\in\mathcal{N},\\
&&&\{i\in\mathcal{N}:x_i>0\}\in\mathcal{X}.
\end{aligned}
\end{aligned}
\end{equation}
Here, the first two constraints guarantees that the purchase probabilities $\bm{x}$ are valid under MNL. The third set of constraints is the market share balancing constraint. The last constraint is equivalent to the last constraint in~\ref{prob:randomized_constrained_0}, which requires that the set of offered products belongs to $\mathcal{X}$. Since~\ref{prob:randomized_constrained_0}~is equivalent to~\ref{prob:randomized_constrained}, for the rest of this section, we mainly focus on~\ref{prob:randomized_constrained}.

To solve \ref{prob:randomized_constrained}, we assume access to an approximation oracle that approximately solves the assortment optimization problem under MNL with constraint $\mathcal{X}$. Specifically, given any $A\subseteq\mathcal{N}$, any set of revenues $\{r_i\}_{i\in\mathcal{N}}$ and any set of preference weights $\{v_i\}_{i\in\mathcal{N}}$, we consider the problem

\begin{equation}\label{prob:constrained_mnl}
\tag{\texttt{MNL-}$\mathcal{X}$}
\begin{aligned}
&\max_{S\subseteq A}&&\dfrac{\sum_{i\in S}r_iv_i}{1+\sum_{i\in S}v_i} & \text{s.t.} && S\in\mathcal{X}.
\end{aligned}
\end{equation}
We assume that there exists an oracle that either returns a $\beta$-approximation to \ref{prob:constrained_mnl} or determines its infeasibility in polynomial time. We define $\beta=1$ if the oracle solves \ref{prob:constrained_mnl} exactly in polynomial time. The oracle exists for a wide variety of constraints considered in the literature such as cardinality constraints \citep{rusmevichientong2010dynamic}, totally unimodular constraints \citep{sumida2021revenue}, capacity constraints \citep{desir2022capacitated} and covering constraints \citep{housni2024assortment}.

We show that given a $\beta$-approximation oracle to \ref{prob:constrained_mnl}, there exists a polynomial-time $\beta$-approximation algorithm for~\ref{prob:randomized_constrained}.

\vspace{0.5em}

\begin{theorem}\label{thm:randomized_constrained}
Given a $\beta$-approximation oracle for \ref{prob:constrained_mnl}, there exists a polynomial-time algorithm that returns a $\beta$-approximation to~\ref{prob:randomized_constrained}.
\end{theorem}

\vspace{0.5em}

Theorem~\ref{thm:randomized_constrained}~states that if the assortment optimization problem under MNL with constraint $\mathcal{X}$ has a polynomial-time $\beta$-approximation oracle, then there is a polynomial-time algorithm that achieves the same approximation ratio for~\ref{prob:randomized_constrained}. In particular, if under constraint $\mathcal{X}$, the constrained assortment optimization problem under MNL is polynomial-time solvable, then~\ref{prob:randomized_constrained}~is also polynomial-time solvable.

We use $R^*_{\mathcal{X}}$ to denote the optimal value of~\ref{prob:randomized_constrained}. Similar to Problem~\eqref{prob:randomized_linearized}, to prove Theorem~\ref{thm:randomized_constrained}, we first establish an equivalent form of \ref{prob:randomized_constrained}. We consider the following problem,

\begin{equation}\label{prob:randomized_constrained_linearized}
\begin{aligned}
&\max_{\bm{w}}&&\sum_{i\in\mathcal{N}}(r_i-R^*_{\mathcal{X}})w_i\\
&\text{s.t.}&&0\leq w_i\leq v_i,\ \forall i\in\mathcal{N},\\
&&&w_i\in\{0\}\cup\Big[\alpha\cdot\max_{j\in\mathcal{N}}w_j,\infty\Big),\ \forall i\in\mathcal{N},\\
&&&\{i\in\mathcal{N}:w_i>0\}\in\mathcal{X}.
\end{aligned}
\end{equation}

We first use the following lemma to prove that Problem~\eqref{prob:randomized_constrained_linearized}~is equivalent to~\ref{prob:randomized_constrained}.

\vspace{0.5em}

\begin{lemma}
The optimal value to Problem~\eqref{prob:randomized_constrained_linearized}~is $R_{\mathcal{X}}^*$. Furthermore, letting $\bm{w}^*$ be an optimal solution to Problem~\eqref{prob:randomized_constrained_linearized} and defining $\bm{x}^*$ using Equation~\eqref{eq:reconstruct_x}, $\bm{x}^*$ is an optimal solution to~\ref{prob:randomized_constrained}.\label{lemma:randomized_constrained_linearized}
\end{lemma}

\vspace{0.5em}

Proof of Lemma~\ref{lemma:randomized_constrained_linearized}~is similar to that of Lemma~\ref{lemma:randomized_linearized}, and we leave the proof to Appendix~\ref{sec:lemma:randomized_constrained_linearized}. Next, we use the following lemma to reveal the structure of the optimal solution to Problem~\eqref{prob:randomized_constrained_linearized}.

\vspace{0.5em}

\begin{lemma}\label{lemma:w_structure}
For some $\underline{v}\in\{v_i:i\in\mathcal{N}\}\cup\{\alpha v_i:i\in\mathcal{N}\}$ and $S^*\in\mathcal{X}$ such that $v_i\geq \underline{v}$ for all $i\in S^*$, an optimal solution $\bm{w}^*$ to Problem~\eqref{prob:randomized_constrained_linearized} is given by
\begin{equation}
w_i^*=\begin{cases}
\min\{v_i,\underline{v}/\alpha\},\ &\text{if}\ i\in S^*\ \text{and}\ r_i\geq R_{\mathcal{X}}^*,\\
\underline{v},\ &\text{if}\ i\in S^*\ \text{and}\ r_i<R_{\mathcal{X}}^*,\\
0,\ &\text{if}\ i\notin S^*.
\end{cases}\label{eq:w_structure}
\end{equation}
\end{lemma}

\vspace{0.5em}

In what follows, we provide a proof sketch of Lemma \ref{lemma:w_structure}. Consider any optimal solution $\tilde{\bm{w}}$ and we set $S^*=\{i\in\mathcal{N}:\tilde{w}_i>0\}$. By the last set of constraints in Problem~\eqref{prob:randomized_constrained_linearized}~we get $S^*\in\mathcal{X}$. If we fix the set of offered products as $S^*$, then Problem \eqref{prob:randomized_constrained_linearized} is reduced to a linear program, and let $\hat{\bm{w}}$ be an optimal basic feasible solution. We further define $\underline{v}=\min\{\hat{w}_i:\hat{w}_i>0\}$, and we prove that $\underline{v}$ can be limited to a set of at most $2n$ candidate values. Finally, Fixing the set of offered products as $S^*$ and enforcing $\underline{v}\leq w_i\leq \underline{v}/\alpha$, we obtain an optimal solution $\bm{w}^*$. By the first and second sets of constraints in Problem~\eqref{prob:randomized_constrained_linearized}~we have that $\underline{v}\leq w_i^*\leq \min\{v_i,\underline{v}/\alpha\}$ for all $i\in S^*$. Then we show that an optimal $\bm{w}^*$ follows the structure stated in~\eqref{eq:w_structure}. The reason for this is that if $r_i\geq R^*_{\mathcal{X}}$, the coefficient for $w_i$ is nonnegative, and in this case $w_i$ is optimal when setting it as the upper bound $\min\{v_i,\underline{v}/\alpha\}$; otherwise if $r_i<R^*_{\mathcal{X}}$, then the coefficient for $w_i$ is negative, and in this case $w_i$ is optimal when setting it as the lower bound $\underline{v}$. Proof of Lemma~\ref{lemma:w_structure}~is provided in Appendix~\ref{sec:lemma:w_structure}.

Suppose $\underline{v}$ and $R^*_{\mathcal{X}}$ are known, and we define $S^+(\underline{v})=\{i\in\mathcal{N}:v_i\geq \underline{v}\}$. Since by Lemma~\ref{lemma:w_structure}, we have $v_i\geq \underline{v}$ for all $i\in S^*$ and $w_i^*=0$ for all $i\notin S^*$, we get $w_i^*=0$ for all $i\notin S^+(\underline{v})$. Thus we can restrict the universe of products to $S^+(\underline{v})$. We define a new set of preference weights $\bm{v}'$ as follows: for each $i\in S^+(\underline{v})$, we define
\begin{equation*}
v_i'=\begin{cases}
\underline{v},\ &\text{if}\ r_i<R_{\mathcal{X}}^*,\\
\min\{v_i,\underline{v}/\alpha\},\ &\text{if}\ r_i\geq R_{\mathcal{X}}^*.
\end{cases}
\end{equation*}
By Lemma~\ref{lemma:w_structure}, there exists an optimal solution $\bm{w}^*$ to Problem~\eqref{prob:randomized_constrained_linearized}~such that $w_i^*\in\{0,v_i'\}$ for all $i\in\mathcal{N}$. Therefore Problem~\eqref{prob:randomized_constrained_linearized}~can be equivalently rewritten as
\begin{equation}\label{prob:randomized_constrained_linearized_transformed}
\begin{aligned}
&\max_{S\subseteq S^+(\underline{v})}&&\sum_{i\in S}(r_i-R^*_{\mathcal{X}})v_i'&\text{s.t.}&& S\in\mathcal{X}.
\end{aligned}
\end{equation}

By Lemma~\ref{lemma:randomized_constrained_linearized}, the optimal value of Problem~\eqref{prob:randomized_constrained_linearized_transformed}~is also $R^*_{\mathcal{X}}$. Then we have \mbox{$\sum_{i\in S}(r_i-R^*_{\mathcal{X}})v_i'\leq R^*_{\mathcal{X}}$} for all $S\in\mathcal{X}$, and equality holds when $S=S^*$. Equivalently we have \mbox{$\sum_{i\in S}r_iv_i'/(1+\sum_{i\in S}v_i')\leq R^*_{\mathcal{X}}$} for all $S\in\mathcal{X}$, and equality holds when $S=S^*$. Therefore we conclude that Problem~\eqref{prob:randomized_constrained_linearized_transformed}~is equivalent to a constrained assortment optimization problem under MNL given by

\begin{equation}\label{prob:transformed_mnl}
\tag{\texttt{Transformed-MNL}}
\begin{aligned}
&\max_{S\subseteq S^+(\underline{v})}&&\dfrac{\sum_{i\in S}r_iv_i'}{1+\sum_{i\in S}v_i'}
&\text{s.t.}&&S\in\mathcal{X}.
\end{aligned}
\end{equation}
Note that in \ref{prob:transformed_mnl}, the universe of products $S^+(\underline{v})$ and the preference weights $\bm{v}'$ depend on the parameter $\underline{v}$. We remark that similar transformation between a constrained assortment optimization problem and a binary linear program is also used in \cite{rusmevichientong2010dynamic}~and~\cite{sumida2021revenue}. In observance of these, we are able to use the following lemma to prove that Problem~\eqref{prob:randomized_constrained_linearized}~is equivalent to~\ref{prob:transformed_mnl}, and $S^*$ is also optimal to~\ref{prob:transformed_mnl}. Since we have proven in Lemma~\ref{lemma:randomized_constrained_linearized}~that~\ref{prob:randomized_constrained}~is equivalent to Problem~\eqref{prob:randomized_constrained_linearized}, this also implies that \ref{prob:randomized_constrained}~is equivalent to~\ref{prob:transformed_mnl}.

\vspace{0.5em}

\begin{lemma}\label{lemma:transformed_mnl}
{  For some $\underline{v}\in\{v_i:i\in\mathcal{N}\}\cup \{\alpha v_i:i\in\mathcal{N}\}$}, the optimal value of~\ref{prob:transformed_mnl}~under parameter $\underline{v}$ is $R^*_{\mathcal{X}}$.
\end{lemma}

\vspace{0.5em}

Proof of Lemma~\ref{lemma:transformed_mnl}~is provided in Appendix~\ref{sec:lemma:transformed_mnl}.

By Lemma~\ref{lemma:transformed_mnl}, if $R^*_{\mathcal{X}}$ and $\underline{v}$ are known, by obtaining a $\beta$-approximation $\hat{S}$ and setting $\bm{x}$ as the purchase probabilities under preference weights $\bm{v}'$ and assortment $\hat{S}$, we obtain a $\beta$-approximation to~\ref{prob:randomized_constrained}. With unknown $R^*_{\mathcal{X}}$ and $\underline{v}$, we can replace $R^*_{\mathcal{X}}$ with \mbox{$\underline{r}=\min\{r_i:r_i\geq R^*_{\mathcal{X}}\}$} in the definition of $\bm{v}'$. Since we have proven in Lemma~\ref{lemma:w_structure}~that $\underline{v}\in\{v_i:i\in\mathcal{N}\}\cup\{\alpha v_i:i\in\mathcal{N}\}$, and by definition $\underline{r}\in\{r_i:i\in\mathcal{N}\}$, there are $O(n^2)$ possible combinations of $(\underline{r},\underline{v})$. We enumerate all possible combinations of $(\underline{r},\underline{v})$, and generate a candidate solution for each guess of $(\underline{r},\underline{v})$. By doing this we are able to establish an approximation algorithm for~\ref{prob:randomized_constrained}. The complete algorithm is provided in Algorithm~\ref{alg:randomized_constrained}.

\begin{algorithm}[t]
\SingleSpacedXI
\caption{Algorithm for~\ref{prob:randomized_constrained}}
\label{alg:randomized_constrained}
\begin{algorithmic}
\For{$\hat{r}\in\{r_i:i\in\mathcal{N}\}$}
\For{$\hat{v}\in\{v_i:i\in\mathcal{N}\}\cup\{\alpha v_i:i\in\mathcal{N}\}$}
\State Set $v_i'=\hat{v}$ if $r_i<\hat{r}$, $v_i'=\min\{v_i,\hat{v}/\alpha\}$ if $r_i\geq \hat{r}$. 
\State Solve~\ref{prob:transformed_mnl}~under the $\bm{v}'$ defined above approximately. If~\ref{prob:transformed_mnl}~is infeasible then set $S{(\hat{r},\hat{v})}=\varnothing$, otherwise set $S{(\hat{r},\hat{v})}$ as a $\beta$-approximation to~\ref{prob:transformed_mnl}.
\State Set $\bm{x}{(\hat{r},\hat{v})}$ as
\begin{equation*}
x_i{(\hat{r},\hat{v})}=\dfrac{v_i'}{1+\sum_{j\in S{(\hat{r},\hat{v})}}v_j'},\ \forall i\in S{(\hat{r},\hat{v})},\ x_i{(\hat{r},\hat{v})}=0,\ \forall i\notin  S{(\hat{r},\hat{v})},\ x_0(\hat{r},\hat{v})=\dfrac{1}{1+\sum_{j\in S(\hat{r},\hat{v})} v_j'}.
\end{equation*}
\State Set $R{(\hat{r},\hat{v})}=\sum_{i\in\mathcal{N}}r_ix_i{(\hat{r},\hat{v})}$.
\EndFor
\EndFor
\State Set $(r^*,v^*)=\argmax_{(\hat{r},\hat{v})}R{(\hat{r},\hat{v})}$, and return $\bm{x}{(r^*,v^*)}$.
\end{algorithmic}
\end{algorithm}

Now we are ready to finish the proof of Theorem~\ref{thm:randomized_constrained}.

\begin{proof}[Proof of Theorem~\ref{thm:randomized_constrained}]

We first verify that the output of the algorithm is feasible to~\ref{prob:randomized_constrained}. It suffices to show that for all $(\hat{r},\hat{v})$, the corresponding candidate solution $\bm{x}(\hat{r},\hat{v})$ is feasible to~\ref{prob:randomized_constrained}. If \ref{prob:transformed_mnl}~is infeasible for some $(\hat{r},\hat{v})$, then by definition $S(\hat{r},\hat{v})=\varnothing$, thus $\bm{x}(\hat{r},\hat{v})=0$, which is feasible to~\ref{prob:randomized_constrained}. Suppose \ref{prob:transformed_mnl}~is feasible under some $(\hat{r},\hat{v})$. It is easy to verify that $\bm{x}(\hat{r},\hat{v})$ satisfies the first constraint in~\ref{prob:randomized_constrained}. Since $v_i'\leq v_i$ for all $i\in S^+(\hat{v})$, we have that for all $i\in S^+(\hat{v})$
\begin{equation*}
x_i(\hat{r},\hat{v})\leq v_i'x_0(\hat{r},\hat{v})\leq v_ix_0(\hat{r},\hat{v}).
\end{equation*}
Therefore $\bm{x}(\hat{r},\hat{v})$ satisfies the second set of constraints in~\ref{prob:randomized_constrained}. Since by definition \mbox{$\hat{v}\leq v_i'\leq \hat{v}/\alpha$} for all $i\in S^+(\hat{v})$, for any $i\in S(\hat{r},\hat{v})$, we have
\begin{equation*}
x_i(\hat{r},\hat{v})=v_i'x_0(\hat{r},\hat{v})\geq \hat{v}x_0(\hat{r},\hat{v})=\alpha\cdot\hat{v}/\alpha\cdot x_0(\hat{r},\hat{v})\geq \alpha\cdot\max_{j\in\mathcal{N}} v_j'x_0(\hat{r},\hat{v})\geq \alpha\cdot\max_{j\in\mathcal{N}}x_j(\hat{r},\hat{v}).
\end{equation*}
Furthermore, we have $x_i(\hat{r},\hat{v})=0$ for all $i\notin S(\hat{r},\hat{v})$. Therefore $\bm{x}(\hat{r},\hat{v})$ satisfies the third set of constraints in~\ref{prob:randomized_constrained}. Since $S(\hat{r},\hat{v})\in \mathcal{X}$ as long as~\ref{prob:transformed_mnl}~is feasible, we have \mbox{$\{i\in\mathcal{N}:x_i(\hat{r},\hat{v})>0\}=S(\hat{r},\hat{v})\in\mathcal{X}$,} thus $\bm{x}(\hat{r},\hat{v})$ satisfies the last constraint in~\ref{prob:randomized_constrained}. Therefore we conclude that for all $(\hat{r},\hat{v})$, $\bm{x}(\hat{r},\hat{v})$ is feasible to~\ref{prob:randomized_constrained}, thus the output of the algorithm is always feasible.

We prove the approximation ratio of the algorithm. By Lemma \ref{lemma:transformed_mnl}, there exists some \mbox{$\underline{v}\in\{v_i:i\in\mathcal{N}\}\cup\{\alpha v_i:i\in\mathcal{N}\}$} such that \ref{prob:transformed_mnl} under parameter $\underline{v}$ is equivalent to \ref{prob:randomized_constrained}. We further define $\underline{r}=\min\{r_i:r_i\geq R_{\mathcal{X}}^*,i\in\mathcal{N}\}$, then $\underline{r}\in\{r_i:i\in\mathcal{N}\}$. The objective value of the output is greater than or equal to the objective value of $\bm{x}(\underline{r},\underline{v})$. Therefore it suffices to verify that $\bm{x}(\underline{r},\underline{v})$ is a $\beta$-approximation to~\ref{prob:randomized_constrained}. Since $S(\underline{r},\underline{v})$ is a $\beta$-approximation to~\ref{prob:transformed_mnl}, by Lemma~\ref{lemma:transformed_mnl}~we have that the optimal value of~\ref{prob:transformed_mnl}~is $R^*_{\mathcal{X}}$, therefore
\begin{equation*}
\sum_{i\in\mathcal{N}}r_ix_i(\underline{r},\underline{v})=\dfrac{\sum_{i\in S(\underline{r},\underline{v})}r_iv_i'}{1+\sum_{i\in S(\underline{r},\underline{v})}v_i'}\geq \beta R^*_{\mathcal{X}}.
\end{equation*}
Therefore $\bm{x}(\underline{r},\underline{v})$ is a $\beta$-approximation to~\ref{prob:randomized_constrained}, and we conclude that the algorithm returns a $\beta$-approximation to~\ref{prob:randomized_constrained}.
\end{proof}

\section{Dynamic Assortment Optimization}\label{sec:dynamic}

In this section, we consider the dynamic assortment optimization problem with market share balancing constraint. We first introduce the formal definition of the dynamic problem. Then, to solve the dynamic problem, we formulate an optimization problem that provides an upper bound to the optimal value of the dynamic problem. We prove that the upper bound problem is NP-hard, and establish an FPTAS to approximately solve the upper bound problem. The approximate solution returned by the FPTAS ultimately allows us to construct a policy for the dynamic problem.

\subsection{Problem Formulation}

We formally define the dynamic assortment optimization problem with market share balancing constraint. The seller has an initial inventory of $c_i\in\mathbb{Z}^+$ units of product $i$ without replenishment for each $i\in\mathcal{N}$. The selling horizon consists of $T$ time periods indexed by \mbox{$t\in\{1,2,\dots,T\}$.} In each time period $t$, the seller observes the history of sales and offers an assortment $S_t\subseteq\mathcal{N}$ sampled from $q_t(\cdot)$, and a customer makes a purchase $i_t\in\mathcal{N}\cup\{0\}$ based on MNL. Here we set $i_t=0$ if the customer arriving at time period $t$ leaves without a purchase.

We define $\bm{H}_t=\{i_\tau\}_{\tau\in\{1,2,\dots,t\}}$ as the history of sales up to time period $t$. The decision of the seller is to choose a dynamic control policy $\pi$. In each time period $t$, the policy $\pi$ takes the history of sales $\bm{H}_{t-1}$ by time period $t-1$ as input, and returns a distribution over assortments $q_t(\cdot)$ to offer to customers at time $t$.  We use $\Pi$ to denote all dynamic control policies in this setting, and use the random variable $X_{it}^\pi$ to denote the realized sales of product $i$ by time period $t$ under policy $\pi$. The policy $\pi$ has to satisfy the following two constraints. First, we require that for all $i\in\mathcal{N}$, the realized sales $X_{iT}^\pi$ of product $i$ throughout the entire time horizon should not exceed the initial inventory $c_i$ with probability one. Second, we also require the market share balancing constraint to be satisfied in expectation, which means the expected sales of any two products offered with positive probability should differ by a factor of at most $\alpha$. We denote the dynamic problem as~\ref{prob:inventory}, and the problem is given by

\begin{equation}\label{prob:inventory}
\tag{\texttt{BMS-Dynamic}}
\begin{aligned}
&\max_{\pi\in\Pi}&&{\sum_{i\in\mathcal{N}}r_i\mathbb{E}[X_{iT}^\pi]}\\
&\text{s.t.}&&{\mathbb{P}(X_{iT}^\pi\leq c_i)=1,\ \forall i\in\mathcal{N},}\\
&&&{\mathbb{E}[X_{iT}^\pi]\in\{0\}\cup \Big[\alpha\cdot \max_{j\in\mathcal{N}}\mathbb{E}[X_{jT}^\pi],\infty\Big),\ \forall i\in\mathcal{N}.}
\end{aligned}
\end{equation}

\subsection{Upper Bound Problem}\label{sec:upper_bound}

The dynamic control policy depends on the history of sales which is high-dimensional, thus it is computationally intractable to solve for the exact optimal policy to~\ref{prob:inventory}. In this section, we establish an optimization problem that provides an upper bound for the optimal value of~\ref{prob:inventory}. The upper bound problem serves two purposes: first, we make use of an approximate solution of the upper bound problem to construct our policy; second, we use the optimal value of the upper bound problem as a benchmark to prove a performance guarantee of our policy. We further prove in this section that the upper bound problem is NP-hard, and design an FPTAS for the upper bound problem, which is able to return a $(1-\epsilon)$-approximation for the problem within polynomial time in input size and $1/\epsilon$ for any $\epsilon>0$. 

We consider the following problem,

\begin{equation}\label{prob:inventory_upper_bound}
\tag{\texttt{Upper-Bound}}
\begin{aligned}
&\max_{\bm{x}}&&{T\sum_{i\in\mathcal{N}}r_ix_i}\\
&\text{s.t.}&&{x_0+\sum_{i\in\mathcal{N}}x_i=1},\\
&&&{0\leq x_i\leq v_ix_0,\ \forall i\in\mathcal{N}},\\
&&&{x_i\leq c_i/T,\ \forall i\in\mathcal{N}},\\
&&&{x_i\in\{0\}\cup\Big[\alpha\cdot\max_{j\in\mathcal{N}}x_j,\infty\Big),\ \forall i\in\mathcal{N}.}
\end{aligned}
\end{equation}
Here, $x_i$ can be interpreted as the average expected sales of product $i$ per time period, i.e.,  $x_i=\mathbb{E}[X_{iT}^\pi]/T$, and $x_0$ can be interpreted as the average expected number of no-purchase per time period. The objective is equal to $\sum_{i\in\mathcal{N}}r_i\mathbb{E}[X_{iT}^\pi]=T\sum_{i\in\mathcal{N}}r_ix_i$. We relax the first set of constraints of~\ref{prob:inventory}~as $\mathbb{E}[X_{iT}^\pi]\leq c_i$ for all $i\in\mathcal{N}$, and the relaxed constraints are equivalent to the third set of constraints in~the \ref{prob:inventory_upper_bound} problem. The last set of constraints of~the \ref{prob:inventory_upper_bound} problem~is equivalent to the last set of constraints of~\ref{prob:inventory}. We use the following lemma to verify that the optimal value of~the \ref{prob:inventory_upper_bound} problem~is indeed an upper bound on that of~\ref{prob:inventory}.

\vspace{0.5em}

\begin{lemma}\label{lemma:inventory_upper_bound}
The optimal value of~\ref{prob:inventory}~is upper bounded by that of~the \ref{prob:inventory_upper_bound} problem.
\end{lemma}

\vspace{0.5em}

Proof of Lemma~\ref{lemma:inventory_upper_bound}~is provided in Appendix~\ref{sec:lemma:inventory_upper_bound}. If there is no inventory constraints (i.e., $c_i\geq T$ for all $i\in\mathcal{N}$), the \ref{prob:inventory_upper_bound} problem~is equivalent to~\ref{prob:inventory}. The reason for this is that if $c_i\geq T$ for all $i\in\mathcal{N}$, then~the optimal value of the \ref{prob:inventory_upper_bound} problem~is equal to $T$ times that of \ref{prob:randomized}, or equivalently $T$ times that of \ref{prob:randomized_0}. Suppose $\bm{x}^*$ and $q^*(\cdot)$ are optimal solutions to~the \ref{prob:inventory_upper_bound} problem~and~\ref{prob:randomized_0}~respectively, always sampling assortments from $q^*(\cdot)$ is feasible to~\ref{prob:inventory} if $c_i\geq T$ for all $i\in\mathcal{N}$, and the corresponding expected revenue is $T\sum_{i\in\mathcal{N}}r_ix_i^*$, which matches the optimal value of~the \ref{prob:inventory_upper_bound} problem. Therefore we conclude that the upper bound problem~the \ref{prob:inventory_upper_bound} problem~is tight when $c_i\geq T$ for all $i\in\mathcal{N}$.

We would like to remark that the \ref{prob:inventory_upper_bound} problem is not a linear program because the last set of constraints in the \ref{prob:inventory_upper_bound} problem are not linear constraints. In fact, we show that~the \ref{prob:inventory_upper_bound} problem~is NP-hard.

\vspace{0.5em}

\begin{proposition}\label{thm:inventory_upper_bound}
The \ref{prob:inventory_upper_bound} problem~is NP-hard.
\end{proposition}

\vspace{0.5em}

{  We would like to remark that the NP-hardness of the \ref{prob:inventory_upper_bound} problem is due to the interplay between the inventory constraint (the third constraint in the \ref{prob:inventory_upper_bound} problem) and the market share balancing constraint (the fourth constraint in the \ref{prob:inventory_upper_bound} problem). While the inventory constraint or the market share balancing constraint along does not make the problem NP-hard, the problem becomes NP-hard when the problem has both of the two constraints.} Proof of Proposition~\ref{thm:inventory_upper_bound}~is provided in Appendix~\ref{sec:thm:inventory_upper_bound}. We show the NP-hardness of~the \ref{prob:inventory_upper_bound} problem~by using a reduction from a cardinality constrained version of the subset sum problem.

We use the following proposition to show that the~\ref{prob:inventory_upper_bound}~admits an FPTAS.

\begin{proposition}\label{prop:upper_bound_fptas}
The \ref{prob:inventory_upper_bound} problem admits an FPTAS, and its runtime for obtaining a $(1-\epsilon)$-approximate solution is
\begin{equation*}
O\Big(\dfrac{n^2}{\epsilon^4}\log\Big(\dfrac{1}{\alpha}\Big)\log(nv_{\max}+1)\log\Big(T+\dfrac{nv_{\max}+1}{v_{\min}}\Big)\Big).
\end{equation*}
\end{proposition}

Proof of Proposition \ref{prop:upper_bound_fptas} is provided in Appendix \ref{sec:prop:upper_bound_fptas}. With an approximate solution of the \ref{prob:inventory_upper_bound} problem, we are able to construct our policy for the dynamic problem, which we will discuss in the next section.

\section{Algorithm for the Dynamic Problem}\label{sec:dynamic_inventory}

In this section, we introduce an asymptotically optimal policy for~\ref{prob:inventory}, whose approximation ratio converges to one as initial inventories get large. We focus on the case of $\alpha<1$ throughout this section. For the special case of $\alpha=1$, we have a simpler policy that is also asymptotically optimal, and we present the latter policy in Appendix~\ref{sec:policy_alpha=1}.

\noindent\textbf{Description of the Policy:} We consider policies in the following form: Each product $i\in\mathcal{N}$ is associated with a purchase probability $\hat{x}_i$. In each time period, the seller chooses a distribution over assortments such that the purchase probability of any product $i\in\mathcal{N}$ is $\hat{x}_i$ if the product has remaining inventory, or zero if the product has no remaining inventory. In other words, the purchase probability of product $i$ is always $\hat{x}_i$ as long as product $i$ does not run out of inventory. To ensure that the purchase probabilities are valid under MNL, we require that the purchase probabilities $\hat{\bm{x}}$ should satisfy $\hat{x}_i\leq v_i(1-\sum_{j\in\mathcal{N}}\hat{x}_j)$ for all $i\in\mathcal{N}$. Then a distribution over assortments defined in Equation~\eqref{eq:sales_to_distribution} is able to attain the desired purchase probabilities $\hat{\bm{x}}$.

For any $p\in(0,1)$ and $c\in\mathbb{Z}^+$, we define $G(p,c)$ as the expected sales of a product if its purchase probability is $p$ when there is remaining inventory and its initial inventory is $c$. By definition $G(p,c)=\mathbb{E}[\min\{Z(T,p),c\}]$, where $Z(T,p)$ is a binomial random variable with $T$ trials and success probability being $p$. Then under the aforementioned policy, the expected sales of product $i$ is $G(\hat{x}_i,c_i)$.

\noindent\textbf{Definition of Purchase Probabilities $\hat{\bm{x}}$:} It is tempting to directly use an approximate solution $\tilde{\bm{x}}$ of~the \ref{prob:inventory_upper_bound} problem~as purchase probabilities $\hat{\bm{x}}$ in the policy. However, we note that directly using $\tilde{\bm{x}}$ as purchase probabilities in the policy may not satisfy the market share balancing constraint. The reason for this is that in the presence of inventory constraints, the expected sales $G(\tilde{x}_i,c_i)$ of product $i$ may not be proportional to $\tilde{x}_i$. If we directly use $\tilde{\bm{x}}$ as purchase probabilities, the expected sales $\{G(\tilde{x}_i,c_i)\}_{i\in\mathcal{N}}$ may not satisfy the market share balancing constraint even though nonzero elements of $\tilde{\bm{x}}$ are within a ratio of $\alpha$ from each other.

Therefore we use an alternative approach to construct purchase probabilities $\hat{\bm{x}}$ such that the market share balancing constraint is satisfied. As the first step, we still obtain a \mbox{$(1-\epsilon_1)$-approximate} solution $\tilde{\bm{x}}$ of~the \ref{prob:inventory_upper_bound} problem. We define \mbox{$\bar{S}=\{i\in\mathcal{N}:\tilde{x}_i>0\}$}, and for any $i\notin \bar{S}$, the purchase probability $\hat{x}_i$ is set to be zero. Then we define \mbox{$G_{\min}=\min\{G(\tilde{x}_i,c_i):i\in\bar{S}\}$} as the minimum nonzero expected sales if we directly use $\tilde{\bm{x}}$ as purchase probabilities in the policy. To satisfy the market share balancing constraint, it suffices to find a set of purchase probabilities $\hat{\bm{x}}$ such that \mbox{$G_{\min}\leq G(\hat{x}_i,c_i)\leq G_{\min}/\alpha$} for all $i\in\bar{S}$. For each product $i\in\bar{S}$ such that $G(\tilde{x}_i,c_i)\leq G_{\min}/\alpha$, by setting $\hat{x}_i=\tilde{x}_i$ we already have \mbox{$G_{\min}\leq G(\hat{x}_i,c_i)=G(\tilde{x}_i,c_i)\leq G_{\min}/\alpha$.} Because $G(p,c)$ is continuous and monotonically increasing in $p$, for each product $i\in\bar{S}$ such that $G(\tilde{x}_i,c_i)>G_{\min}/\alpha$, we have $0=G(0,c_i)<G_{\min}/\alpha<G(\tilde{x}_i,c_i)$, then there exists $0<y_i<\tilde{x}_i$ such that $G(y_i,c_i)=G_{\min}/\alpha$. Since it is challenging to solve for the exact $y_i$, we search for an approximate solution by running a bisection search so that the last equality is satisfied within a tolerance of $1-\epsilon_2$, where $\epsilon_2$ is the precision of the bisection search, and we set $\hat{x}_i$ as the approximate solution.

\noindent\textbf{Algorithm and Performance Guarantee:} The complete algorithm is provided in Algorithm~\ref{alg:policy_1}. In the last step of Algorithm~\ref{alg:policy_1}, we can construct a distribution over assortments from purchase probabilities following Equation~\eqref{eq:sales_to_distribution}. We provide a more detailed discussion on the existence of such distribution over assortments in the proof of the forthcoming Theorem~\ref{thm:inventory_constant_1}.

\begin{algorithm}[t]
\SingleSpacedXI
\caption{Asymptotically Optimal Policy for \ref{prob:inventory} when $\alpha<1$}
\label{alg:policy_1}
\begin{algorithmic}
\State Obtain a $(1-\epsilon_1)$-approximation $\tilde{\bm{x}}$ to~the \ref{prob:inventory_upper_bound} problem.
\State Set $\bar{S}=\{i\in\mathcal{N}:\tilde{x}_i>0\}$, $G_{\min}=\min_{i\in\bar{S}} G(\tilde{x}_i,c_i)$.
\For{$i\in\bar{S}$}
\If{$G(\tilde{x}_i,c_i)\leq G_{\min}/\alpha$}
\State Set $\hat{x}_i=\tilde{x}_i$.
\Else
\State Search for $\hat{x}_i$ that satisfies $(1-\epsilon_2)G_{\min}/\alpha\leq G(\hat{x}_i,c_i)\leq G_{\min}/\alpha$ using bisection search.
\EndIf
\EndFor
\State \textbf{Policy:} For each time period $t$, given realized sales $\bm{X}_{t-1}$ by the end of time period $t-1$, set the purchase probabilities $\bm{p}_t$ for time period $t$ as $p_{it}=\hat{x}_i$ if $X_{i,t-1}<c_i$, and $p_{it}=0$ if $X_{i,t-1}=c_i$. Compute a distribution over assortments $\pi_t$ such that the purchase probabilities are $\bm{p}_t$, and offer an assortment randomly drawn from $\pi_t$ at time period $t$.
\end{algorithmic}
\end{algorithm}

We define $c_{\min}=\min\{c_i:i\in\mathcal{N}\}$ as the minimum initial inventory. We prove that with sufficiently high precision of the bisection search, the policy obtained from Algorithm~\ref{alg:policy_1}~is feasible to \ref{prob:inventory}; furthermore, the policy always attains a constant factor guarantee for \ref{prob:inventory}, and is asymptotically optimal as initial inventories get large.

\vspace{0.5em}

\begin{theorem}\label{thm:inventory_constant_1}
For any $\epsilon>0$, by setting $\epsilon_1=\epsilon/2$ and $\epsilon_2=\min\{\epsilon/2,1-\alpha\}$, Algorithm~\ref{alg:policy_1}~returns a $\max\{1/2-\epsilon,1-{1}/{\sqrt{c_{\min}}}-\epsilon\}$-approximation to~\ref{prob:inventory}. The running time of Algorithm~\ref{alg:policy_1} is polynomial in the input size, $1/\epsilon$ and $\log(1/(1-\alpha))$.
\end{theorem}

\vspace{0.5em}

In Theorem~\ref{thm:inventory_constant_1}, we use higher precision for the bisection search as $\alpha$ approaches one. The reason for this is that as $\alpha$ gets close to one, a stronger market share balancing constraint is enforced, requiring that expected sales of offered products to be closer to each other. In order to satisfy the market share balancing constraint, we need higher precision for the bisection search. Since the precision of bisection search is upper bounded by $1-\alpha$, the number of iterations for the bisection search is polynomially dependent on $\log(1/(1-\alpha))$, thus the runtime of Algorithm~\ref{alg:policy_1} is also polynomially dependent on $\log(1/(1-\alpha))$.

Theorem~\ref{thm:inventory_constant_1}~states that Algorithm~\ref{alg:policy_1}~always returns a constant-factor approximation to~\ref{prob:inventory}. Furthermore, as the minimum initial inventory $c_{\min}$ increases, the approximation ratio converges to $1-\epsilon$, where $\epsilon>0$ is the precision of the FPTAS for~the \ref{prob:inventory_upper_bound} problem and the bisection search.

Before we proceed to the proof of Theorem~\ref{thm:inventory_constant_1}, we first provide an upper bound on the runtime for the bisection search.

\vspace{0.5em}

\begin{lemma}\label{lemma:binary_search}
If $G(\tilde{x}_i,c_i)>G_{\min}/\alpha$, the number of iterations of the bisection search for $\hat{x}_i$ is upper bounded by
\begin{equation*}
O\Big(\log(T)+\log\Big(\dfrac{1}{\epsilon_2}\Big)+\log\Big(\dfrac{1+nv_{\max}}{v_{\min}}\Big)+\log\Big(\dfrac{1}{\alpha}\Big)\Big).
\end{equation*}
\end{lemma}

\vspace{0.5em}

To prove Lemma~\ref{lemma:binary_search}, We first prove that $G_{\min}\geq \alpha v_{\min}/4(1+nv_{\max})$. Furthermore, after $\lceil\log_2(T/\epsilon_2G_{\min})\rceil$ iterations of bisection search, if we set $\hat{x}_i$ as the left endpoint of the interval from the bisection search, the distance between the output $\hat{x}_i$ and the exact solution $y_i$ is at most \mbox{$2^{-M(\epsilon_2)}\leq \epsilon_2G_{\min}/T$.} We further prove that $G(p,c)$ is $T$-Lipschitz in $p$ for any $c\in\mathbb{Z}^+$, which further implies \mbox{$|G(\hat{x}_i,c_i)-G_{\min}/\alpha|\leq T|\hat{x}_i-y_i|\leq \epsilon_2G_{\min}$,} thus the inequality in the lemma is verified. Proof of Lemma~\ref{lemma:binary_search}~is provided in Appendix~\ref{sec:lemma:binary_search}.

Now we are able to finish the proof of Theorem~\ref{thm:inventory_constant_1}.

\begin{proof}[Proof of Theorem~\ref{thm:inventory_constant_1}]
To prove the theorem, we first verify the feasibility of the proposed policy, then analyze its runtime, and finally prove its performance guarantee. 

\noindent\underline{\bf Feasibility of the Policy:} 
We first verify that in the policy defined in Algorithm~\ref{alg:policy_1}, there always exists a distribution over assortments that attains the desired purchase probabilities. Theorem 1 of \cite{topaloglu2013joint} states that purchase probabilities $\bm{x}$ can be attained by a distribution over assortments if and only if $x_i\leq v_i(1-\sum_{j\in\mathcal{N}}x_j)$ for all $i\in\mathcal{N}$. Therefore it suffices to prove that $p_{it}\leq v_i(1-\sum_{j\in\mathcal{N}}p_{jt})$ for all $t$. By definition of $\bm{p}_t$ we have $p_{it}\leq \hat{x}_i$ for all $i\in\mathcal{N}$. We further prove that $\hat{x}_i\leq \tilde{x}_i$ for all $i\in\mathcal{N}$. Consider any $i\in\mathcal{N}$. If $i\notin \bar{S}$ then $\hat{x}_i=\tilde{x}_i=0$. If $i\in \bar{S}$ and $G(\tilde{x}_i,c_i)\leq G_{\min}/\alpha$ then $\hat{x}_i=\tilde{x}_i$. If $G(\tilde{x}_i,c_i)>G_{\min}/\alpha$, then by definition of $\hat{x}_i$ in this case we have $G(\hat{x}_i,c_i)\leq G_{\min}/\alpha<G(\tilde{x}_i,c_i)$. Since $G(p,c)$ is monotonically increasing in $p$, we have $\hat{x}_i\leq \tilde{x}_i$. Therefore $\hat{x}_i\leq \tilde{x}_i$ for all $i\in\mathcal{N}$, which implies $p_{it}\leq \tilde{x}_i$ for all $i\in\mathcal{N}$. By the first and second set of constraints in~the \ref{prob:inventory_upper_bound} problem~we get $\tilde{x}_i\leq v_i(1-\sum_{j\in\mathcal{N}}\tilde{x}_j)$ for all $i\in\mathcal{N}$. Therefore for any $i\in\mathcal{N}$ we have
\begin{equation*}
p_{it}\leq \tilde{x}_i\leq v_i\Big(1-\sum_{j\in\mathcal{N}}\tilde{x}_j\Big)\leq v_i\Big(1-\sum_{j\in\mathcal{N}}p_{jt}\Big).
\end{equation*}
Thus, there always exists a distribution over assortments that attains purchase probabilities $\bm{p}_t$.

We verify that the proposed policy satisfies the market share balancing constraint. Under the proposed policy, the purchase probability of any product $i\in\mathcal{N}$ is $\hat{x}_i$ as long as there is remaining inventory for the product. Therefore for any $i\in\mathcal{N}$, the distribution of realized sales of product $i$ is the same as that of $\min\{Z(\hat{x}_i,T),c_i\}$. Therefore for any $i\in\mathcal{N}$ we have
\begin{equation*}
\mathbb{E}[X_{iT}^\pi]=\mathbb{E}\Big[\min\Big\{Z(\hat{x}_i,T),c_i\Big\}\Big]=G(\hat{x}_i,c_i).
\end{equation*}
It is easy to verify that $\mathbb{E}[X_{iT}^\pi]=0$ for all $i\notin \bar{S}$. Therefore to verify that $\pi$ satisfies the market share balancing constraint, it suffices to prove that $G_{\min}\leq G(\hat{x}_i,c_i)\leq G_{\min}/\alpha$ for all $i\in\bar{S}$. Consider any $i\in\bar{S}$. If $G(\tilde{x}_i,c_i)\leq G_{\min}/\alpha$, then by definition $\hat{x}_i=\tilde{x}_i$ and $G_{\min}=\min_{j\in\bar{S}}G(\tilde{x}_j,c_j)$, thus \mbox{$G_{\min}\leq G(\hat{x}_i,c_i)=G(\tilde{x}_i,c_i)\leq G_{\min}/\alpha$.} If $G(\tilde{x}_i,c_i)>G_{\min}/\alpha$, then by definition of $\hat{x}_i$ we get \mbox{$(1-\epsilon_2)G_{\min}/\alpha\leq G(\hat{x}_i,c_i)\leq G_{\min}/\alpha$.} Since by definition $\epsilon_2\leq 1-\alpha$, we have $(1-\epsilon_2)G_{\min}/\alpha\geq G_{\min}$, thus $G_{\min}\leq G(\hat{x}_i,c_i)\leq G_{\min}/\alpha$. Therefore we conclude that $G_{\min}\leq G(\hat{x}_i,c_i)\leq G_{\min}/\alpha$ for all $i\in\bar{S}$, thus the policy satisfies the market share balancing constraint, and is a feasible policy to~\ref{prob:inventory}.

\noindent\underline{\bf Runtime:} In Proposition~\ref{thm:inventory_upper_bound}, we prove that there exists an FPTAS for~the \ref{prob:inventory_upper_bound} problem. Therefore the runtime of obtaining a $(1-\epsilon_1)$-approximation of~the \ref{prob:inventory_upper_bound} problem is polynomial in the input size and $1/\epsilon_1$. By definition $\epsilon_1=\epsilon/2$, thus the runtime of obtaining a $(1-\epsilon_1)$-approximate solution of the \ref{prob:inventory_upper_bound} problem is also polynomial in the input size and $1/\epsilon$. Furthermore, by Lemma~\ref{lemma:binary_search}, the maximum number of iterations is polynomial in the input size and $\log(1/\epsilon_2)$. Since we set $\epsilon_2=\min\{\epsilon/2,1-\alpha\}$, we get $\epsilon_2\geq (1-\alpha)\epsilon/2$, thus $\log(1/\epsilon_2)\leq \log(1/(1-\alpha))+\log(2/\epsilon)$. Therefore by Proposition~\ref{prop:upper_bound_fptas} and Lemma~\ref{lemma:binary_search}, the total runtime of Algorithm~\ref{alg:policy_1} is bounded by 
\begin{align*}
&O\Big(\dfrac{n^2}{\epsilon^4}\log\Big(\dfrac{1}{\alpha}\Big)\log(nv_{\max}+1)\log\Big(T+\dfrac{nv_{\max}+1}{v_{\min}}\Big)\Big)\\
&+n\cdot O\Big(\log(T)+\log\Big(\dfrac{2}{\epsilon}\Big)+\log\Big(\dfrac{1}{1-\alpha}\Big)+\log\Big(\dfrac{1+nv_{\max}}{v_{\min}}\Big)+\log\Big(\dfrac{1}{\alpha}\Big)\Big)\\
\leq&O\Big(\dfrac{n^2}{\epsilon^4}\log\Big(\dfrac{1}{\alpha}\Big)\log(nv_{\max}+1)\log\Big(T+\dfrac{nv_{\max}+1}{v_{\min}}\Big)+n\log\Big(\dfrac{1}{1-\alpha}\Big)\Big).
\end{align*}

\noindent\underline{\bf Performance Guarantee:} We define $\theta(c_{\min})=\max\{{1}/{2},1-{1}/{\sqrt{c_{\min}}}\}$. We prove that the expected revenue of the proposed policy is at least $\theta(c_{\min})-\epsilon$ fraction of the optimal value of~\ref{prob:inventory}. We first prove that $G(\hat{x}_i,c_i)\geq (1-\epsilon_2)\theta(c_{\min})T\tilde{x}_i$ for all $i\in\bar{S}$. Recall that we use $Z(T,p)$ to denote the binomial random variable with $T$ trials and success probability being $p$. Lemma E.1 of \cite{bai2022coordinated}~states that if $Tp\leq c$ for some $c\in\mathbb{Z}^+$, then we have
\begin{equation*}
\mathbb{E}[(Z(T,p)-c)^+]\leq \min\Big\{\dfrac{1}{2},\dfrac{1}{\sqrt{c}}\Big\}Tp.
\end{equation*}
Consider any $i\in\bar{S}$. By the third set of constraints in the \ref{prob:inventory_upper_bound} problem we have $T\tilde{x}_i\leq c_i$. Thus by Lemma E.1 of  \cite{bai2022coordinated} we get
\begin{equation*}
G(\tilde{x}_i,c_i)=\mathbb{E}[\min\{Z(T,\tilde{x}_i),c_i\}]=T\tilde{x}_i-\mathbb{E}[(Z(T,\tilde{x}_i)-c_i)^+]\geq \max\Big\{\dfrac{1}{2},1-\dfrac{1}{\sqrt{c_i}}\Big\}T\tilde{x}_i\geq \theta(c_{\min})T\tilde{x}_i.
\end{equation*}
We define $i_{\min}=\arg\min_{i\in\bar{S}}G(\tilde{x}_i,c_i)$. In particular, by the inequality above we have \mbox{$G(\tilde{x}_{i_{\min}},c_{i_{\min}})\geq \theta(c_{\min})T\tilde{x}_{i_{\min}}$}, which implies $G_{\min}\geq \theta(c_{\min})T\tilde{x}_{i_{\min}}$. Consider any $i\in\bar{S}$. If $G(\tilde{x}_i,c_i)\leq G_{\min}/\alpha$, then we get $\tilde{x}_i=\hat{x}_i$, and $G(\hat{x}_i,c_i)=G(\tilde{x}_i,c_i)\geq \theta(c_{\min})T\tilde{x}_i\geq (1-\epsilon_2)\theta(c_{\min})T\tilde{x}_i$. If $G(\tilde{x}_i,c_i)>G_{\min}/\alpha$, then by definition of $\hat{x}_i$ we get
\begin{equation*}
G(\hat{x}_i,c_i)\geq (1-\epsilon_2)G_{\min}/\alpha\geq (1-\epsilon_2)\theta(c_{\min})T\tilde{x}_{i_{\min}}/\alpha\geq (1-\epsilon_2)\theta(c_{\min})T\tilde{x}_i.
\end{equation*}
Here the last inequality holds because $\tilde{\bm{x}}$ is feasible to~the \ref{prob:inventory_upper_bound} problem, and the last set of constraints in~the \ref{prob:inventory_upper_bound} problem implies \mbox{$\tilde{x}_{i_{\min}}\geq \alpha\cdot\max_{j\in\mathcal{N}}\tilde{x}_j\geq \alpha \tilde{x}_i$.} Therefore we conclude that for all $i\in\bar{S}$, we have $G(\hat{x}_i,c_i)\geq (1-\epsilon_2)\theta(c_{\min})T\tilde{x}_i$. In this case, we have that the total expected revenue of the policy satisfies
\begin{equation*}
\sum_{i\in\bar{S}}r_iG(\hat{x}_i,c_i)\geq (1-\epsilon_2)\theta(c_{\min})T\sum_{i\in\bar{S}}r_i\tilde{x}_i=(1-\epsilon_2)\theta(c_{\min})\sum_{i\in\mathcal{N}}r_i\tilde{x}_i.
\end{equation*}

Since $\tilde{\bm{x}}$ is a $(1-\epsilon_1)$-approximation to~the \ref{prob:inventory_upper_bound} problem, the optimal value of~\ref{prob:inventory}~is upper bounded by $\sum_{i\in\mathcal{N}}r_i\tilde{x}_i/(1-\epsilon_1)$. Therefore the ratio between the expected revenue of the proposed policy over the optimal value of~\ref{prob:inventory} is at least
\begin{equation*}
(1-\epsilon_1)(1-\epsilon_2)\theta(c_{\min})\geq (1-\epsilon)\max\Big\{\dfrac{1}{2},1-\dfrac{1}{\sqrt{c_{\min}}}\Big\}\geq\max\Big\{\dfrac{1}{2}-\epsilon,1-\dfrac{1}{\sqrt{c_{\min}}}-\epsilon\Big\}.
\end{equation*}
Therefore the algorithm gives a $\max\{{1}/{2}-\epsilon,1-{1}/{\sqrt{c_{\min}}}-\epsilon\}$-approximation to~\ref{prob:inventory}.
\end{proof}

{  Finally, we would like to remark that Theorem \ref{thm:inventory_constant_1} can be generalized to any choice model that satisfies substitutability, i.e., $\phi(i,S)\leq \phi(i,S')$ for all $i\in S'$ and $S\supseteq S'$. Specifically, for a general choice model, using similar arguments as in Proposition~\ref{thm:inventory_upper_bound}~we have that the optimal value of the dynamic problem can be upper bounded by that of the following problem,
\begin{equation*}
\begin{aligned}
&\max_{h,\bm{x}}&&T\sum_{i\in\mathcal{N}}r_ix_i\\
&\text{s.t.}&&x_i=\sum_{S\subseteq\mathcal{N}}h(S)\phi(i,S),\ \forall i\in\mathcal{N}\\
&&&x_i\leq c_i/T,\ \forall i\in\mathcal{N},\\
&&&x_i\in\{0\}\cup\Big[\alpha\cdot \max_{j\in\mathcal{N}}x_j,\infty\Big),\ \forall i\in\mathcal{N},\\
&&&\sum_{S\subseteq\mathcal{N}}h(S)=1,\\
&&&h(S)\geq 0,\ \forall S\subseteq\mathcal{N}.
\end{aligned}
\end{equation*}
Assume that we can solve the upper bound problem optimally and let $(h^*,\bm{x}^*)$ be an optimal solution. Based on $\bm{x}^*$, we can obtain $\hat{\bm{x}}$ following Algorithm~\ref{alg:policy_1}. Then in each time period, we sample assortments from a distribution over assortments such that the purchase probability of each product $i$ with remaining inventory is equal to $\hat{x}_i$. Here, the distribution over assortments can be constructed using the sub-assortment sampling in \cite{feng2024near}, which we refer to Procedure 1 in \cite{feng2024near} for details. Following similar arguments in the proof of Theorem~\ref{thm:inventory_constant_1}, we can prove that the policy returned by the algorithm is a $\max\{1/2-\epsilon,1-1/\sqrt{c_{\min}}-\epsilon\}$-approximation to the dynamic problem, where $\epsilon$ is the precision of the bisection search when constructing $\hat{\bm{x}}$.
} 

{ 

\section{Numerical Experiments on Real-World Data}\label{sec:numerical_real}

In this section, we conduct numerical experiments using MNL models calibrated from real data. We study the following two questions: First, how effective is our approach in maintaining a tradeoff between revenue and balanced market shares compared to the alternative approach with absolute market share lower and upper bounds? Second, how large can the balancing parameter $\alpha$ get if the seller is willing to lose a certain fraction of revenue, and how will the maximum and minimum market shares among offered products change accordingly?

\subsection{Data}

We use the dataset \cite{ECommerceData} for our numerical experiments. The dataset contains purchase records from a large online store from April 2020 to November 2020. Each purchase record consists of the date and time of the purchase, the type of the purchased product (e.g., smartphones, tablet, notebook, printer) as well as its price. We conduct our numerical experiments on selected product types that appear in the dataset. 

For each of the selected product type, we calibrate a separate MNL model. The MNL model for each product type are calibrated as follows: We first divide the time horizon covered by the dataset into time intervals of two weeks, and assume that the offered assortment is fixed during each time interval. The assortment offered in each time interval is assumed to be the set of products that have been sold at least once during the period. Since the dataset does not record customers leaving without a purchase, which is crucial for calibrating the MNL model, we add artificial no-purchase records to supplement the purchase records. Specifically, the number of no-purchase records within each time interval is set to be 5 percent of the total number of purchases during the time interval. Finally, we calibrate an MNL model for the corresponding product type using the maximum likelihood estimation, based on the assortment offered in each time interval, the number of purchases of each offered product as well as the no-purchase option under the corresponding assortment.

\subsection{Comparison with the Absolute Market Share Balancing Constraint}

In this section, we conduct numerical experiments to compare the market share balancing constraint with the absolute market share balancing constraint introduced in the model discussion in Section \ref{sec:model_discussion}. Recall that the absolute market share balancing constraint requires that the market share of any product offered with positive probability should exceed a certain absolute lower bound and should not exceed a certain absolute upper bound. Specifically, we vary the balancing parameter $\alpha$ from $0$ to $0.95$ with a stepsize of $0.05$. For each value of $\alpha$, we compare the optimal expected revenue under our market share balancing constraint with the one under absolute market share balancing constraint, where the market share of any product offered by the seller should be within the range of $[\theta,\theta/\alpha]$ for some parameter $\theta$. In the case of $\alpha=0$, the absolute market share balancing constraint only enforces that the market share of each offered product should be at least $\theta$, and no upper bound is enforced. In this way, we use the absolute market share balancing constraint to guarantee that the market shares of any two products offered by the seller differ by a factor of at most $\alpha$. We vary $\theta\in\{0.002,0.004,0.008,0.016\}$.

The optimal expected revenues under our setting and under absolute market share balancing constraint with different values of $\theta$ and $\alpha$ are provided in Figure~\ref{fig:real_data_1}. Specifically, each sub-figure in Figure~\ref{fig:real_data_1} presents the result of a product type, and in each sub-figure, the $x$-axis is the parameter $\alpha$, and each plot shows how the optimal expected revenue change with $\alpha$ under either our setting or a value of $\theta$. One can observe from Figure~\ref{fig:real_data_1} that with each value of $\alpha$, the optimal expected revenue under our market share balancing constraint always dominates that under the absolute market share balancing constraint, which is not surprising because our market share balancing constraint is weaker than the absolute market share balancing constraint for any value of $\theta$. However, it is worth noting that the optimal expected revenue under our setting often exceeds that under the absolute market share balancing constraint by a significant margin. For example, for tablets, clocks, kettles and printers, the absolute market share balancing constraint with $\theta=0.002$ and $\theta=0.016$ only attains less than half of the optimal expected revenue under our setting when $\alpha$ exceeds $0.8$. Furthermore, for most product types, the maximum revenue gap compared to our setting when $\theta=0.004$ and $\theta=0.008$ across all values of $\alpha$ compared to our setting either exceeds or is close to $10\%$, with the only exception being $\theta=0.008$ for clocks. All the above suggests that if we use measure balance in market shares by ratio, then with the same level of balance, our approach is able to generate a higher expected revenue compared to using the absolute market share balancing constraint, and the margin of improvement is often significant. Another observation is that the optimal expected revenue under the absolute market share balancing constraint can be highly sensitive in the value of $\theta$, and small changes in $\theta$ can result in significant change in revenue. For instance, for tablets, refrigerators, clocks, kettles and TVs, changing $\theta$ from $0.004$ to $0.002$ results in a decrease in revenue by more than $20\%$ when $\alpha$ is above $0.8$. Similar phenomena can also be observed when changing $\theta$ from $0.008$ to $0.016$, which could cause a decrease in revenue by more than $30\%$ when $\alpha$ is above $0.8$ for many product types. Such sensitivity of revenue in $\theta$ also highlights the challenge of fine-tuning lower and upper bounds when enforcing the absolute market share balancing constraint.

\begin{figure}[t]
\centering
\subfigure[Tablet]{\includegraphics[height=4cm]{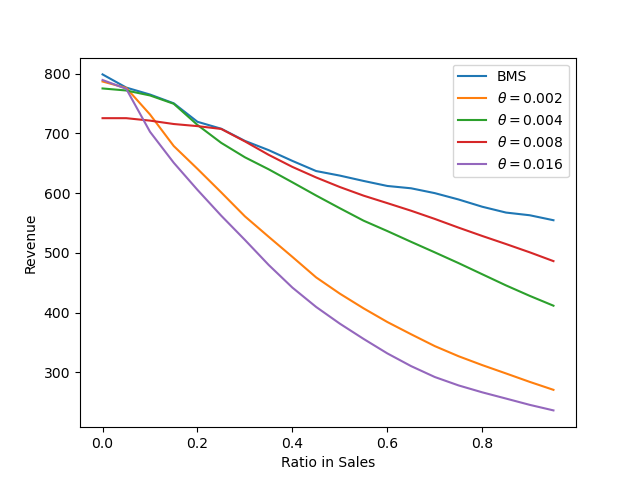}}
\subfigure[Mouse]{\includegraphics[height=4cm]{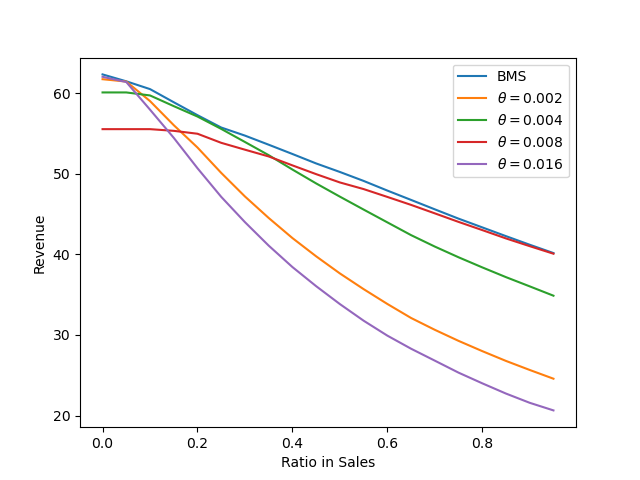}}
\subfigure[Notebook]{\includegraphics[height=4cm]{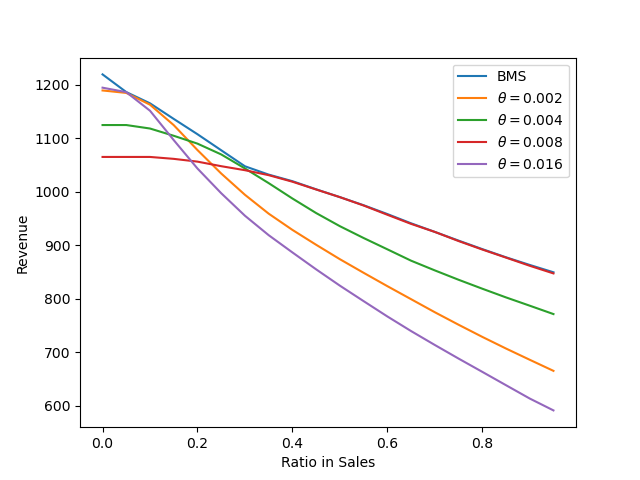}}
\subfigure[Refrigerator]{\includegraphics[height=4cm]{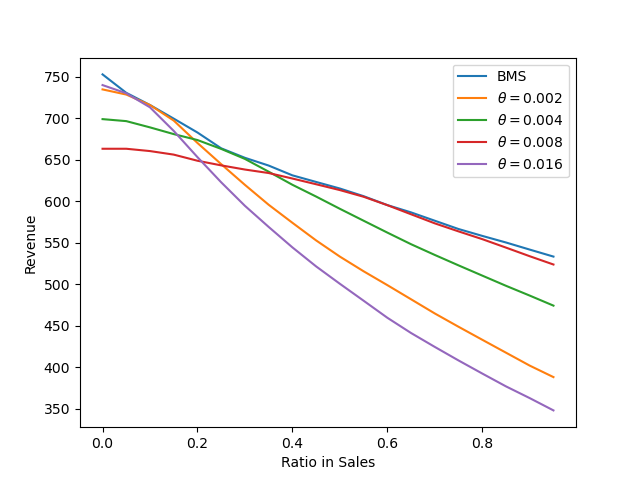}}
\subfigure[Clock]{\includegraphics[height=4cm]{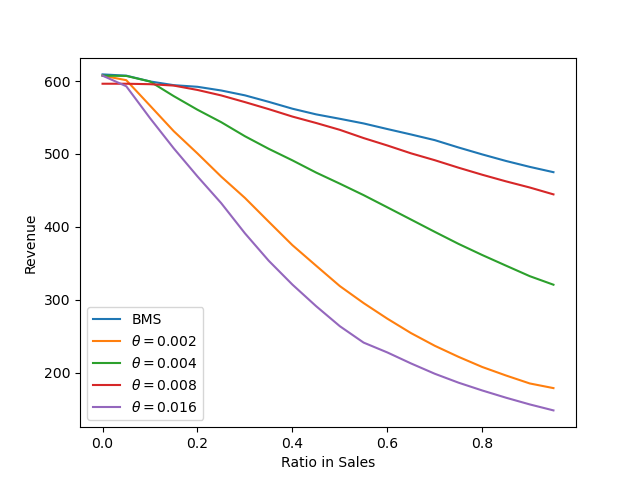}}
\subfigure[Kettle]{\includegraphics[height=4cm]{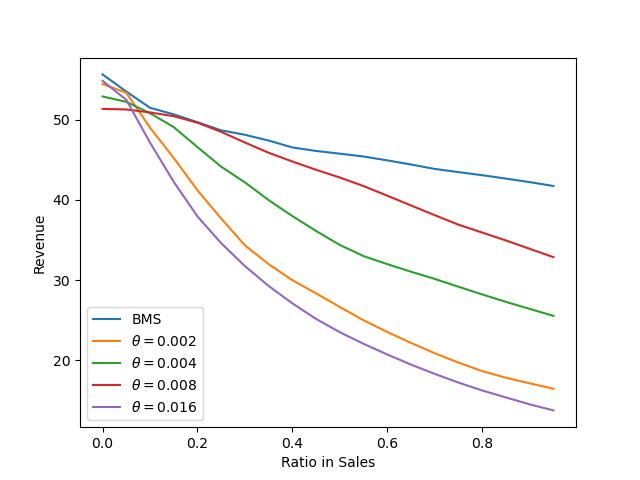}}
\subfigure[Printer]{\includegraphics[height=4cm]{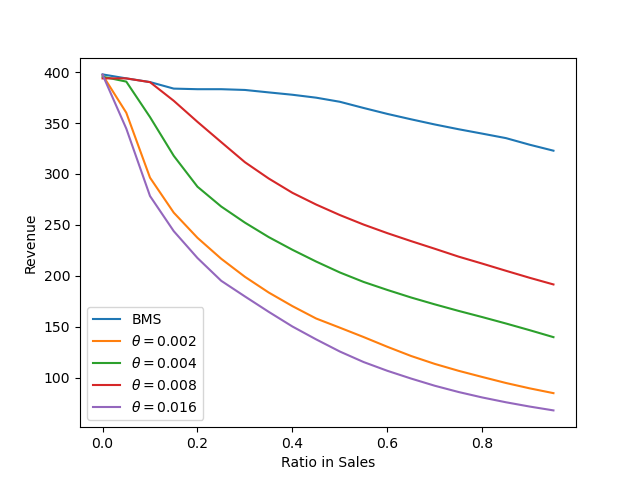}}
\subfigure[TV]{\includegraphics[height=4cm]{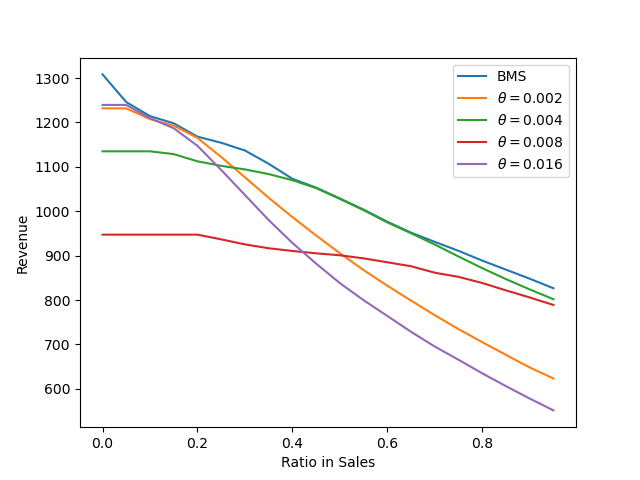}}
\caption{Optimal expected revenue with market share balancing constraint and absolute market share balancing constraint under different $\alpha$}
\label{fig:real_data_1}
\end{figure}

Besides ratio between market shares, another commonly used metric of balanced market shares is the absolute difference between market shares of offered products. Since our approach is not optimized over this metric of balance in market shares, it is unclear how our approach performs under this metric compared to the absolute market share balancing constraint. In what follows, we compare the performance of our market share balancing constraint with that under the absolute market share balancing constraint when balance in market shares is measured by absolute difference. Specifically, we still vary $\alpha$ from $0$ to $0.95$ with a stepsize of $0.05$. For each $\alpha$, we compute an optimal solution of our setting and under each value of $\theta$, and draw a point on the revenue-difference plot, whose $x$-coordinate is the difference between the maximum market share and the minimum nonzero market share, and $y$-coordinate is the expected revenue of the solution. Then for our setting and for each value of $\theta$, we connect the points corresponding to different values of $\alpha$. By connecting these dots and forming the revenue-difference plots, one can observe how the optimal expected revenue and absolute difference in market shares change as $\alpha$ changes. With the revenue-difference plots, one can also directly compare the optimal expected revenue of different approaches under similar absolute differences in market shares. The revenue-difference plot of different product types are provided in Figure~\ref{fig:real_data_2}.

\begin{figure}[t]
\centering
\subfigure[Tablet]{\includegraphics[height=4cm]{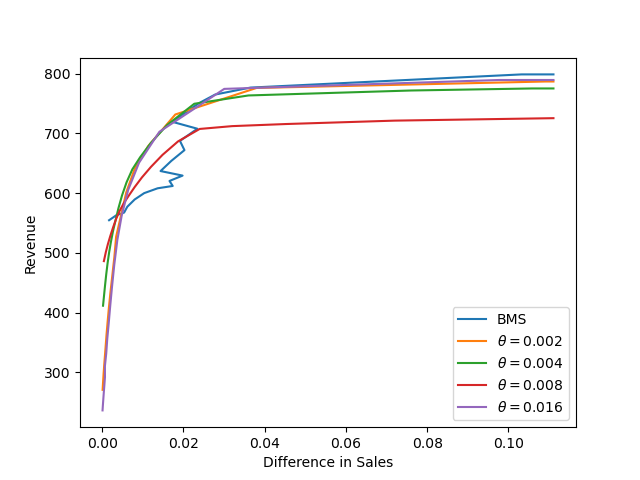}}
\subfigure[Mouse]{\includegraphics[height=4cm]{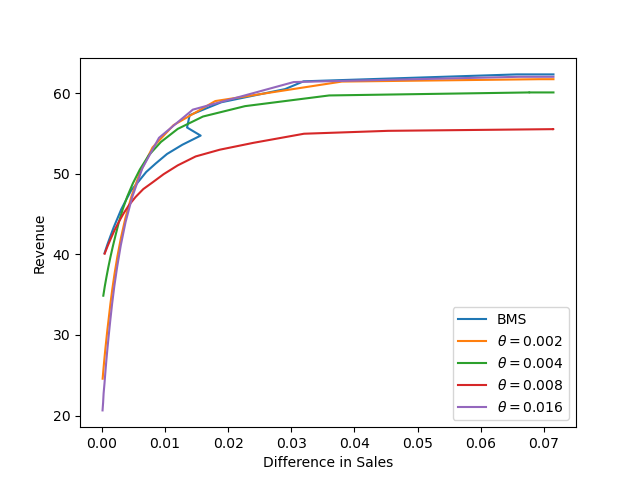}}
\subfigure[Notebook]{\includegraphics[height=4cm]{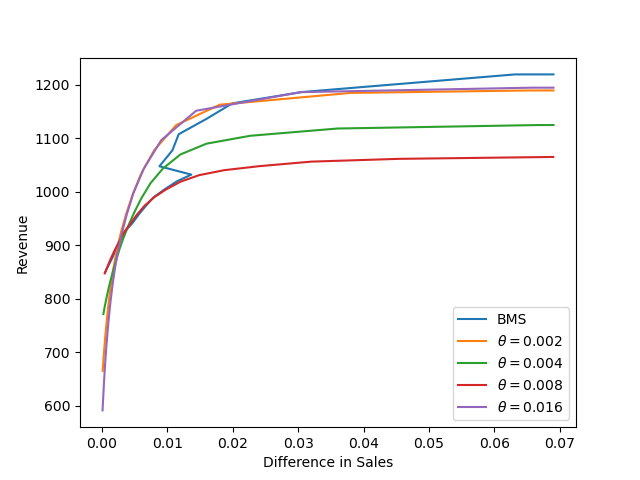}}
\subfigure[Refrigerator]{\includegraphics[height=4cm]{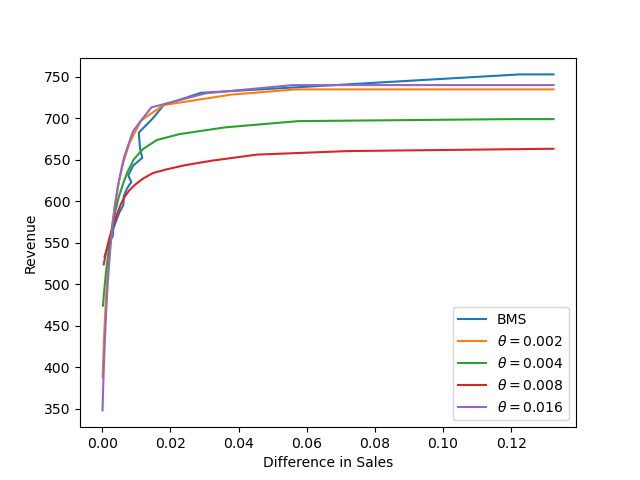}}
\subfigure[Clock]{\includegraphics[height=4cm]{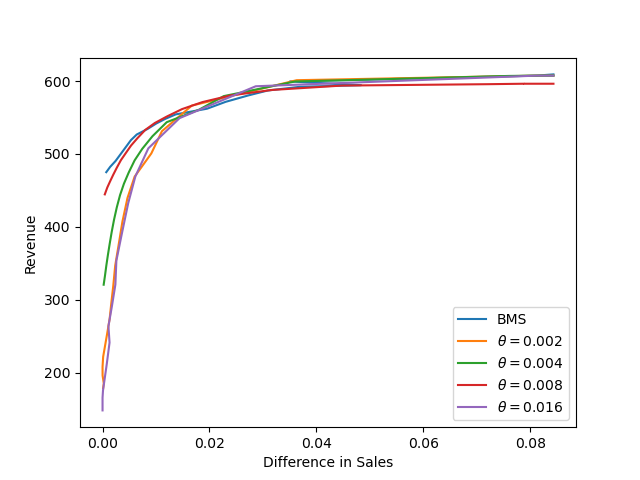}}
\subfigure[Kettle]{\includegraphics[height=4cm]{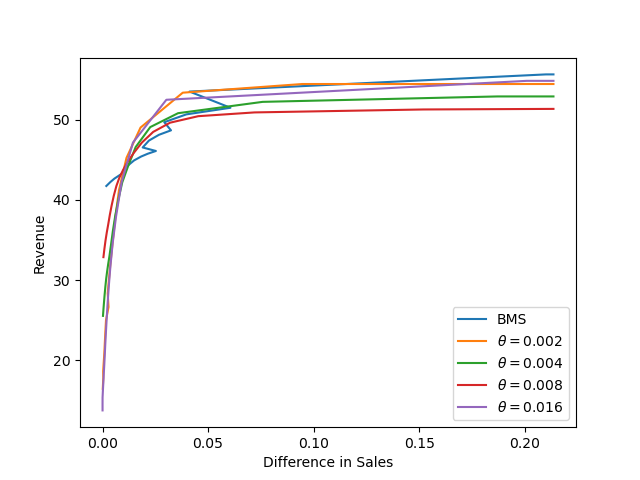}}
\subfigure[Printer]{\includegraphics[height=4cm]{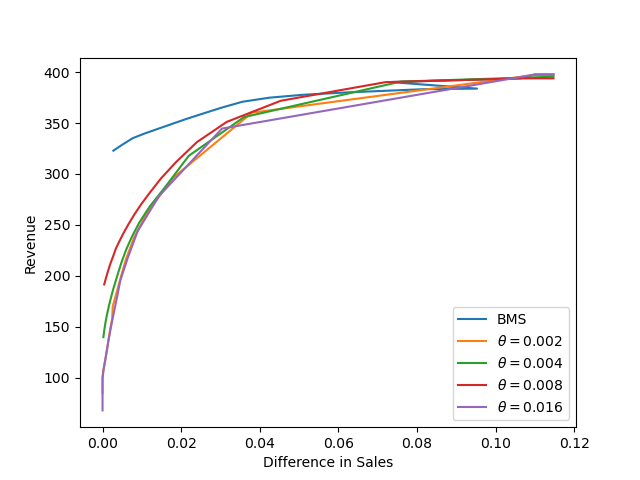}}
\subfigure[TV]{\includegraphics[height=4cm]{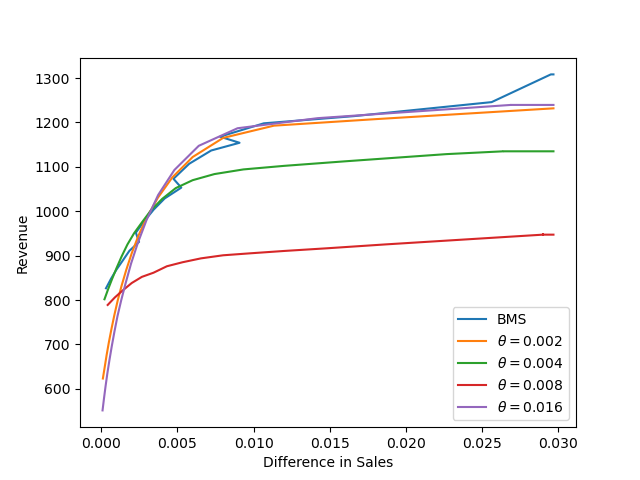}}
\caption{Revenue-difference plot under our setting and the absolute market share balancing constraint for different product types}
\label{fig:real_data_2}
\end{figure}

We would like to remark that here the revenue-difference plots of our setting follow a zigzagging pattern because in our setting, the absolute difference in market shares may not be monotonically decreasing as $\alpha$ increases. Nevertheless, one can observe from Figure~\ref{fig:real_data_2}~that with a similar absolute difference in market shares, our approach either generates an expected revenue either higher than all four values of $\theta$, or close to the best-performing value of $\theta$. This observation suggests that if balance in market shares is measured by absolute difference, our approach still is able to attain a higher or comparable revenue under the same level of balance. This observation further highlights the strength of our approach, as our market share balancing constraint can still perform well when balance in market shares is measured by absolute difference, even though this metric gives more advantage to the absolute market share balancing constraint.

\subsection{Tradeoff between Revenue and Balance in Market Shares}\label{sec:numerical_tradeoff}

In this section, we conduct numerical experiments to study the tradeoff between revenue and balanced market shares. Specifically, we would like to answer the following questions: First, how large can the balancing parameter $\alpha$ get if the seller is willing to lose a certain fraction of revenue? 
Second, how do the maximum and minimum market shares of the offered products change accordingly?

To answer these questions, we conduct numerical experiments on selected product types, and for each product type, we compute the largest $\alpha$ of the market share balancing constraint such that the optimal value of \ref{prob:randomized} is greater than or equal to $(1-\gamma)$ times the unconstrained optimal expected revenue, where $0<\gamma<1$ is the relative revenue loss compared to the unconstrained optimal revenue. Since the optimal expected revenue of \ref{prob:randomized} is monotonically decreasing in $\alpha$, we can use bisection search to find the value of $\alpha$ corresponding to different revenue loss $\gamma$. We vary $\gamma\in\{0,0.02,0.05,0.1,0.15,0.2\}$. For each value of $\gamma$, we also compute the maximum and minimum market shares among the offered products under the corresponding $\alpha$. The values of $\alpha$ corresponding to different relative revenue loss are provided in Table~\ref{table:real_data_1}. The maximum and minimum market shares under different relative revenue loss are provided in Tables~\ref{table:real_data_2} and~\ref{table:real_data_3}, respectively.

\begin{table}[t]
\centering
{\scriptsize
\begin{tabular}{|c|c c c c c c|}
\hline
$\gamma$&0&0.02&0.05&0.1&0.15&0.2\\
\hline
Tablet&0.0069&0.0298&0.0767&0.1479&0.2456&0.3755\\
Mouse&0.0066&0.0522&0.0864&0.1938&0.3013&0.3970\\
Notebook&0.0039&0.0171&0.0425&0.1040&0.2310&0.3550\\
Refrigerator&0.0011&0.0103&0.0396&0.1147&0.2173&0.3472\\
Clock&0.0162&0.1157&0.2896&0.4995&0.7056&0.8677\\
Kettle&0.0011&0.0142&0.0308&0.0874&0.1909&0.3306\\
Printer&0.0108&0.0444&0.0620&0.1323&0.2261&0.3218\\
TV&0.0114&0.0190&0.0425&0.1177&0.2456&0.4097\\
\hline
\end{tabular}
}
\caption{Balancing parameter $\alpha$ for different product types under different revenue loss $\gamma$}
\label{table:real_data_1}
\end{table}

\begin{table}[t]
\centering
{\scriptsize
\begin{tabular}{|c|c c c c c c|}
\hline
$\gamma$&0&0.02&0.05&0.1&0.15&0.2\\
\hline
Tablet&0.1040&0.0367&0.0247&0.0208&0.0152&0.0117\\
Mouse&0.0662&0.0305&0.0196&0.0150&0.0105&0.0091\\
Notebook&0.0634&0.0237&0.0142&0.0103&0.0080&0.0066\\
Refrigerator&0.1226&0.0308&0.0149&0.0098&0.0076&0.0063\\
Clock&0.0857&0.0429&0.0336&0.0235&0.0176&0.0161\\
Kettle&0.2103&0.0484&0.0307&0.0223&0.0220&0.0176\\
Printer&0.1112&0.0686&0.0540&0.0467&0.0351&0.0311\\
TV&0.0299&0.0181&0.0113&0.0088&0.0078&0.0063\\
\hline
\end{tabular}
}
\caption{Maximum market share for different product types under different revenue loss $\gamma$}
\label{table:real_data_2}
\end{table}

\begin{table}[t]
\centering
{\scriptsize
\begin{tabular}{|c|c c c c c c|}
\hline
$\gamma$&0&0.02&0.05&0.1&0.15&0.2\\
\hline
Tablet&0.0007&0.0011&0.0019&0.0031&0.0037&0.0044\\
Mouse&0.0004&0.0016&0.0017&0.0029&0.0032&0.0036\\
Notebook&0.0003&0.0004&0.0006&0.0011&0.0018&0.0024\\
Refrigerator&0.0001&0.0003&0.0006&0.0011&0.0017&0.0022\\
Clock&0.0014&0.0050&0.0097&0.0118&0.0124&0.0140\\
Kettle&0.0002&0.0007&0.0009&0.0020&0.0042&0.0058\\
Printer&0.0012&0.0030&0.0033&0.0062&0.0079&0.0100\\
TV&0.0003&0.0003&0.0006&0.0010&0.0019&0.0026\\
\hline
\end{tabular}
}
\caption{Minimum nonzero market share for different product types under different revenue loss $\gamma$}
\label{table:real_data_3}
\end{table}

In Table~\ref{table:real_data_1}, the first column corresponds to the ratio of the minimum nonzero market shares over maximum market shares under the unconstrained optimal assortment. One can see from the first column that if the seller does not enforce the market share balancing constraint, the market share are significantly unbalanced for all tested product types. Specifically, for all tested product types, the market share of the best selling product is more than 60 times that of the least selling product; for refrigerators and kettles, the ratio is even more than 900. However, if the seller enforces market share balancing constraint at the cost of a fraction of revenue, the balance in market shares among offered products can be significantly improved. As shown in Table~\ref{table:real_data_1}, if the seller is willing to lose 5\% of the unconstrained optimal expected revenue, then the balancing parameter $\alpha$ can exceed or get close to $0.05$ for most tested product types, meaning that the maximum market share would be around 20 times the minimum nonzero market share or even less. For clocks, with a 5\% revenue loss, the balancing parameter can even get as large as $0.28$, meaning that the maximum market share is less than four times the minimum nonzero market share. If the seller is willing to lose 20\% of the unconstrained optimal expected revenue, the balancing parameter can further exceed $1/3$ for most of the tested product types, meaning that the maximum market share is less than 3 times the minimum market share. With these observations, we conclude that by enforcing the market share balancing constraint, the seller can significantly improve the balance in market shares while only incurring moderate revenue loss. 

One can also see from Tables~\ref{table:real_data_2} and~\ref{table:real_data_3}~that as we enforce stronger market share balancing constraint, the maximum market share tends to decrease and the minimum nonzero market share tends to increase, although not being strictly monotone. The decrease in maximum market share and the increase in minimum nonzero market share are often significant when the revenue loss $\gamma$ exceeds $0.1$. For example, if the seller chooses $\alpha$ that corresponds to 15\% revenue loss, then for all tested product types, the maximum market share under the corresponding $\alpha$ is reduced to less than $1/3$ of the maximum market share in the unconstrained setting; for five out of the eight tested product types, the maximum market share is reduced to less than $1/6$ of that in the unconstrained setting. Similarly, for all tested product types, the minimum nonzero market share under the corresponding $\alpha$ increases to at least $5$ times that of the minimum nonzero market share in the unconstrained setting, and for four out of the eight tested product types, the minimum nonzero market share increases to more than $8$ times that in the unconstrained setting. The observations mentioned above suggest that enforcing the market share balancing constraint can effectively reduce the maximum market share and increase the minimum nonzero market share.

}

\section{Conclusion}

In this paper, we propose a novel single parameter approach to define fairness within selected products, which we refer to as the market share balancing constraint. We have considered both the static and dynamic assortment optimization with market share balancing constraint. For the static problem, we prove that the problem is polynomial-time solvable by proving that the set of products offered with positive probability in the optimal solution is nested in revenue and preference weight. We also extend our results to the setting with additional constraints on the set of offered products. For the dynamic problem, we first introduce an upper bound integer program, and establish an asymptotically optimal policy for the dynamic problem based on an approximate solution of upper bound integer program. We also use numerical experiments to test the performance of our policy for the dynamic problem.

For future studies, we suggest exploration in the following directions. First, all results in this paper rely on MNL. Future research could explore the application of market share balancing constraint to other choice models, such as the general attraction model, nested logit model and Markov chain choice model. Second, sellers in practice often faces additional constraints on the each assortment offered to customers such as cardinality constraints or capacity constraints. This requires the seller to randomize only over assortments that satisfy certain constraints, and one potential area for future research is to consider assortment optimization with market share balancing constraint where the assortments randomized over has to satisfy additional constraints. We would like to point out that here the constraints are different from the static problem with additional constraints we discuss in Section \ref{sec:randomized_constrained}, where constraints are imposed to the set of products offered with positive probability instead of assortments the seller randomizes over.



\vspace{-2mm}

{\SingleSpacedXI\bibliographystyle{ormsv080}
\bibliography{ref}

@article{bai2022coordinated,
  title={Coordinated inventory stocking and assortment customization},
  author={Bai, Yicheng and El Housni, Omar and Rusmevichientong, Paat and Topaloglu, Huseyin},
  journal={Operations Research},
  volume={73},
  number={6},
  pages={2953--2971},
  year={2025},
  publisher={INFORMS}
}

@article{talluri2004revenue,
  title={Revenue management under a general discrete choice model of consumer behavior},
  author={Talluri, Kalyan and Van Ryzin, Garrett},
  journal={Management Science},
  volume={50},
  number={1},
  pages={15--33},
  year={2004},
  publisher={INFORMS}
}

@article{topaloglu2013joint,
  title={Joint stocking and product offer decisions under the multinomial logit model},
  author={Topaloglu, Huseyin},
  journal={Production and Operations Management},
  volume={22},
  number={5},
  pages={1182--1199},
  year={2013},
  publisher={SAGE Publications Sage CA: Los Angeles, CA}
}

@book{garey1979computers,
  title={Computers and Intractability},
  author={Garey, Michael R and Johnson, David S},
  volume={174},
  year={1979},
  publisher={Freeman San Francisco}
}

@article{sumida2021revenue,
  title={Revenue-utility tradeoff in assortment optimization under the multinomial logit model with totally unimodular constraints},
  author={Sumida, Mika and Gallego, Guillermo and Rusmevichientong, Paat and Topaloglu, Huseyin and Davis, James},
  journal={Management Science},
  volume={67},
  number={5},
  pages={2845--2869},
  year={2021},
  publisher={INFORMS}
}

@article{desir2022capacitated,
  title={Capacitated assortment optimization: Hardness and approximation},
  author={D{\'e}sir, Antoine and Goyal, Vineet and Zhang, Jiawei},
  journal={Operations Research},
  volume={70},
  number={2},
  pages={893--904},
  year={2022},
  publisher={INFORMS}
}

@misc{housni2024assortment,
  title={Assortment Optimization under the Multinomial Logit Model with Covering Constraints. Available at arXiv 2411.10310},
  author={El Housni, Omar and Feng, Qing and Topaloglu, Huseyin},
  year={2024}
}

@article{bansal2004improved,
  title={Improved fully polynomial time approximation scheme for the 0-1 multiple-choice knapsack problem},
  author={Bansal, Mukul S and Venkaiah, Vadlamudi China},
  journal={International Institute of Information Technology Tech Report},
  pages={1--9},
  year={2004}
}

@misc{chen2022fair,
  title={Fair assortment planning. Available at arXiv 2208.07341},
  year={2022},
  author={Chen, Qinyi and Golrezaei, Negin and Susan, Fransisca},
}

@article{barre2023assortment,
      title={Assortment Optimization with Visibility Constraints}, 
      author={Theo Barre and Omar El Housni and Marouane Ibn Brahim and Andrea Lodi and Danny Segev},
      journal={Mathematical Programming},
      year={2025},
      pages={1--44},
      publisher={Springer}
}

@misc{lu2023simple,
  title={A simple way towards fair assortment planning: Algorithms and welfare implications. Available at SSRN 4514495},
  author={Lu, Wentao and Sahin, Ozge and Wang, Ruxian},
  year={2023}
}

@article{rusmevichientong2010dynamic,
  title={Dynamic assortment optimization with a multinomial logit choice model and capacity constraint},
  author={Rusmevichientong, Paat and Shen, Zuo-Jun Max and Shmoys, David B},
  journal={Operations Research},
  volume={58},
  number={6},
  pages={1666--1680},
  year={2010},
  publisher={INFORMS}
}

@article{gao2021assortment,
  title={Assortment optimization and pricing under the multinomial logit model with impatient customers: Sequential recommendation and selection},
  author={Gao, Pin and Ma, Yuhang and Chen, Ningyuan and Gallego, Guillermo and Li, Anran and Rusmevichientong, Paat and Topaloglu, Huseyin},
  journal={Operations Research},
  volume={69},
  number={5},
  pages={1509--1532},
  year={2021},
  publisher={INFORMS}
}

@article{el2023joint,
  title={Joint assortment optimization and customization under a mixture of multinomial logit models: On the value of personalized assortments},
  author={El Housni, Omar and Topaloglu, Huseyin},
  journal={Operations Research},
  volume={71},
  number={4},
  pages={1197--1215},
  year={2023},
  publisher={INFORMS}
}

@article{housni2023maximum,
  title={Maximum load assortment optimization: Approximation algorithms and adaptivity gaps},
  author={El Housni, Omar and Ibn Brahim, Marouane and Segev, Danny},
  journal={Operations Research},
  volume={74},
  number={1},
  pages={408--429},
  year={2026},
  publisher={INFORMS}
}

@article{gallego1994optimal,
  title={Optimal dynamic pricing of inventories with stochastic demand over finite horizons},
  author={Gallego, Guillermo and Van Ryzin, Garrett},
  journal={Management Science},
  volume={40},
  number={8},
  pages={999--1020},
  year={1994},
  publisher={INFORMS}
}

@article{gallego1997multiproduct,
  title={A multiproduct dynamic pricing problem and its applications to network yield management},
  author={Gallego, Guillermo and Van Ryzin, Garrett},
  journal={Operations Research},
  volume={45},
  number={1},
  pages={24--41},
  year={1997},
  publisher={INFORMS}
}

@article{talluri1998analysis,
  title={An analysis of bid-price controls for network revenue management},
  author={Talluri, Kalyan and Van Ryzin, Garrett},
  journal={Management Science},
  volume={44},
  number={11-part-1},
  pages={1577--1593},
  year={1998},
  publisher={INFORMS}
}

@misc{gallego2004managing,
  title={Managing flexible products on a network. Available at SSRN 3567371},
  author={Gallego, Guillermo and Iyengar, Garud and Phillips, Robert and Dubey, Abhay},
  year={2004}
}

@article{liu2008choice,
  title={On the choice-based linear programming model for network revenue management},
  author={Liu, Qian and Van Ryzin, Garrett},
  journal={Manufacturing \& Service Operations Management},
  volume={10},
  number={2},
  pages={288--310},
  year={2008},
  publisher={INFORMS}
}

@article{rusmevichientong2020dynamic,
  title={Dynamic assortment optimization for reusable products with random usage durations},
  author={Rusmevichientong, Paat and Sumida, Mika and Topaloglu, Huseyin},
  journal={Management Science},
  volume={66},
  number={7},
  pages={2820--2844},
  year={2020},
  publisher={INFORMS}
}

@article{ma2021dynamic,
  title={Dynamic pricing (and assortment) under a static calendar},
  author={Ma, Will and Simchi-Levi, David and Zhao, Jinglong},
  journal={Management Science},
  volume={67},
  number={4},
  pages={2292--2313},
  year={2021},
  publisher={INFORMS}
}

@inproceedings{aouad2023nonparametric,
  title={A Nonparametric Framework for Online Stochastic Matching with Correlated Arrivals},
  author={Aouad, Ali and Ma, Will},
  booktitle={Proceedings of the 24th ACM Conference on Economics and Computation},
  pages={114--114},
  year={2023}
}

@article{li2024revenue,
  title={Revenue management with calendar-aware and dependent demands: Asymptotically tight fluid approximations},
  author={Li, Weiyuan and Rusmevichientong, Paat and Topaloglu, Huseyin},
  journal={Operations Research},
  volume={73},
  number={3},
  pages={1260--1272},
  year={2025},
  publisher={INFORMS}
}

@article{bai2023fluid,
  title={Fluid approximations for revenue management under high-variance demand},
  author={Bai, Yicheng and El Housni, Omar and Jin, Billy and Rusmevichientong, Paat and Topaloglu, Huseyin and Williamson, David P},
  journal={Management Science},
  volume={69},
  number={7},
  pages={4016--4026},
  year={2023},
  publisher={INFORMS}
}

@article{golrezaei2014real,
  title={Real-time optimization of personalized assortments},
  author={Golrezaei, Negin and Nazerzadeh, Hamid and Rusmevichientong, Paat},
  journal={Management Science},
  volume={60},
  number={6},
  pages={1532--1551},
  year={2014},
  publisher={INFORMS}
}

@misc{zhu2024unified,
  title={A Unified Framework to Impose Market Share Constraints for Selected Product Classes: Randomized and Deterministic Assortments under the Multinomial Logit Model. Available at SSRN 4971517},
  author={Zhu, Wenchang and Rusmevichientong, Paat and Topaloglu, Huseyin},
  year={2024}
}

@article{mims2022the,
title={The Surprising Reason Your Amazon Searches Are Returning More Confusing Results than Ever},
author={Mims, Christopher},
url={https://www.wsj.com/business/retail/the-surprising-reason-your-amazon-searches-are-returning-more-confusing-results-than-ever-11656129643},
year={2022}
}

@article{dealhub2025revenue,
title={Revenue Concentration},
author={DealHub},
url={https://dealhub.io/glossary/revenue-concentration},
year={2025}
}

@article{de2001competition,
  title={Competition, innovation, and product exit},
  author={de Figueiredo, John and Kyle, Margaret},
  url={https://home.business.utah.edu/mgtwh/WSC/Figueiredo.pdf},
  year={2001}
}

@article{byrne2015the,
  title={The True Cost of a Long Tail},
  author={Byrne, Robert F.},
  url={https://www.sdcexec.com/safety-security/risk-compliance/article/12126243/terra-technology-division-of-e2open-the-true-cost-of-a-long-tail},
  year={2015}
}

@article{ECommerceData,
    author={Michael Kechinov},
	title={eCommerce purchase history from electronics store},
	journal = {Kaggle},
	year={2020},
	url={https://www.kaggle.com/datasets/mkechinov/ecommerce-purchase-history-from-electronics-store/data},
}

@article{feng2024near,
  title={Near-Optimal Bayesian Online Assortment of Reusable Resources},
  author={Feng, Yiding and Niazadeh, Rad and Saberi, Amin},
  journal={Operations Research},
  volume={72},
  number={5},
  pages={1861--1873},
  year={2024},
  publisher={INFORMS}
}

@article{gallego2015general,
  title={A general attraction model and sales-based linear program for network revenue management under customer choice},
  author={Gallego, Guillermo and Ratliff, Richard and Shebalov, Sergey},
  journal={Operations Research},
  volume={63},
  number={1},
  pages={212--232},
  year={2015},
  publisher={INFORMS}
}
}

\newpage

\begin{appendices}{\Large\noindent{\bf Appendix}}

\section{Proofs for the Static Problem}

\subsection{Proof of Proposition~\ref{thm:value_randomization}}\label{sec:thm:value_randomization}
To prove the proposition, we first prove that Inequality \eqref{eq:gap_upper_bound}~holds for all instances. Then we construct an instance for which~Inequality \eqref{eq:gap_lower_bound}~holds.

\noindent\underline{\bf Proof of Inequality \eqref{eq:gap_upper_bound}:} We first prove that $R^*_{det}\geq R^*/n$. We define $\bm{x}^*$ as an optimal solution to~\ref{prob:randomized}. Let $k=\arg\max_{i\in\mathcal{N}}r_ix_i^*$. Since $\bm{x}^*$ is feasible to~\ref{prob:randomized}, we get $x_{k}^*+x_0^*\leq 1$ and $x_{k}^*\leq v_{k}x_0^*$. Combining the two inequalities we get $x_{k}^*\leq v_{k}/(1+v_{k})$. It is easy to verify that the assortment $\{k\}$ is feasible to~\texttt{BMS-Deterministic}, therefore
\begin{equation}
R^*_{det}\geq R(\{k\})=\dfrac{r_{k}v_{k}}{1+v_{k}}\geq r_{k}x_{k}^*\geq \dfrac{1}{n}\sum_{i\in\mathcal{N}}r_ix_i^*=\dfrac{1}{n}R^*.\label{eq:gap_upper_bound_1}
\end{equation}

Then we prove that $R^*_{det}\geq (1-\alpha)R^*/2$. By Corollary~\ref{cor:optimal_x}, for some $v\in\{v_i:i\in\mathcal{N}\}$, an optimal solution $\bm{x}^*$ to~\ref{prob:randomized}~is given by $x_i^*=w_i^*/(1+\sum_{j\in\mathcal{N}}w_i^*)$, where $\bm{w}^*$ is given by
\begin{equation*}
w_i^*=\begin{cases}
\min\{v_i,v/\alpha\},\ &\text{if}\ v_i\geq v\ \text{and}\ r_i\geq R^*,\\
0,\ &\text{otherwise}.
\end{cases}
\end{equation*}
We have $R^*=\sum_{i\in\mathcal{N}}r_ix_i^*=\sum_{i\in\mathcal{N}}r_iw_i^*/(1+\sum_{i\in\mathcal{N}}w_i^*)$, which is equivalent to $\sum_{i\in\mathcal{N}}(r_i-R^*)w_i^*=R^*$. We define $\tilde{r}_i=r_i-R^*$, then we get $\sum_{i\in\mathcal{N}}\tilde{r}_iw_i^*=R^*$.

We further define
\begin{align*}
&A=\{i\in\mathcal{N}:v\leq v_i\leq v/\alpha,\ r_i\geq R^*\},\ B=\{i\in \mathcal{N}:v/\alpha<v_i\leq v/\alpha^2,\ r_i\geq R^*\},\\
&C=\{i\in\mathcal{N}:v_i>v/\alpha^2,\ r_i\geq R^*\}.
\end{align*}
Then we get
\begin{equation}
R^*=\sum_{i\in\mathcal{N}}\tilde{r}_iw_i^*=\sum_{i\in A}\tilde{r}_iv_i+\sum_{i\in B\cup C}\tilde{r}_i\dfrac{v}{\alpha}.\label{eq:obj_w_star}
\end{equation}
We further define $\bm{w}'\in\mathbb{R}^n$ as $w_i'=v_i$ for all $i\in B$, $w_i'=v/\alpha^2$ for all $i\in C$ and $w_i'=0$ for all $i\notin B\cup C$. By definition of $B$ and $C$ we have that $\bm{w}'$ is feasible to~Problem~\eqref{prob:randomized_linearized}. Since by Lemma~\ref{lemma:randomized_linearized}, the optimal value of Problem~\eqref{prob:randomized_linearized}~is $R^*$, we get
\begin{equation}
\sum_{i\in B}\tilde{r}_iv_i+\sum_{i\in C}\tilde{r}_i\dfrac{v}{\alpha^2}=\sum_{i\in\mathcal{N}}\tilde{r}_iw_i'\leq R^*.\label{eq:obj_w_prime}
\end{equation}
We prove that either $\sum_{i\in A}\tilde{r}_iv_i\geq (1-\alpha)R^*/2$ or $\sum_{i\in B}\tilde{r}_iv_i\geq (1-\alpha)R^*/2$. Suppose $\sum_{i\in A}\tilde{r}_iv_i<(1-\alpha)R^*/2$, then by~\eqref{eq:obj_w_star}~we get
$\sum_{i\in B\cup C}\tilde{r}_iv/\alpha\geq (1+\alpha)R^*/2$, which further implies
\begin{equation}\label{eq:sum_b_c}
\sum_{i\in B\cup C}\tilde{r}_i\dfrac{v}{\alpha^2}\geq \Big(\dfrac{1}{2\alpha}+\dfrac{1}{2}\Big)R^*.
\end{equation}
Subtracting~\eqref{eq:sum_b_c}~with~\eqref{eq:obj_w_prime} we get
\begin{equation*}
\sum_{i\in B}\tilde{r}_i\Big(\dfrac{v}{\alpha^2}-v_i\Big)\geq \Big(\dfrac{1}{2\alpha}-\dfrac{1}{2}\Big)R^*.
\end{equation*}
Since by definition $v_i>v/\alpha$ for all $i\in B$, we have that for any $i\in B$,
\begin{equation*}
\dfrac{1-\alpha}{\alpha}v_i\geq \dfrac{(1-\alpha)v}{\alpha^2}=\dfrac{v}{\alpha^2}-\dfrac{v}{\alpha}\geq \dfrac{v}{\alpha^2}-v_i.
\end{equation*}
Therefore we have
\begin{equation*}
\sum_{i\in B}\tilde{r}_iv_i\geq \dfrac{\alpha}{1-\alpha}\sum_{i\in B}\tilde{r}_i\Big(\dfrac{v}{\alpha^2}-v_i\Big)\geq \dfrac{\alpha}{1-\alpha}\Big(\dfrac{1}{2\alpha}-\dfrac{1}{2}\Big)R^*=\dfrac{1}{2}R^*\geq \dfrac{1-\alpha}{2}R^*.
\end{equation*}
This shows that $\sum_{i\in A}\tilde{r}_iv_i\geq (1-\alpha)R^*/2$ or $\sum_{i\in B}\tilde{r}_iv_i\geq (1-\alpha)R^*/2$ (i.e., at least one of the two inequalities is true). Suppose the former inequality holds, and recall that $\tilde{r}_i=r_i-R^*$, then we get
\begin{equation*}
\sum_{i\in A}r_iv_i\geq \Big(\dfrac{1-\alpha}{2}+\sum_{i\in A}v_i\Big)R^*\geq \dfrac{1-\alpha}{2}\Big(1+\sum_{i\in A}v_i\Big)R^*,
\end{equation*}
which implies
\begin{equation*}
R(A)=\dfrac{\sum_{i\in A}r_iv_i}{1+\sum_{i\in A}v_i}\geq \dfrac{1-\alpha}{2}R^*.
\end{equation*}
Similarly, if $\sum_{i\in B}\tilde{r}_iv_i\geq (1-\alpha)R^*/2$, then we have $R(B)\geq (1-\alpha)R^*/2$. Therefore we always have $\max\{R(A),R(B)\}\geq (1-\alpha)R^*/2$. 

Finally we verify that both assortments $A$ and $B$ are feasible to~\texttt{BMS-Deterministic}. Consider any $i,j\in A$, we get $\phi(i,A)/\phi(j,A)=v_i/v_j{ \geq} v/(v/\alpha)=\alpha$, thus $A$ is feasible to~\texttt{BMS-Deterministic}. Similarly $B$ is also feasible to~\texttt{BMS-Deterministic}. Therefore we get
\begin{equation}
R^*_{det}\geq \max\{R(A),R(B)\}\geq \dfrac{1-\alpha}{2}R^*.\label{eq:gap_upper_bound_2}
\end{equation}
Combining~\eqref{eq:gap_upper_bound_1} and \eqref{eq:gap_upper_bound_2}~we get the desired Inequality~\eqref{eq:gap_upper_bound}.

\noindent\underline{\bf Proof of Inequality \eqref{eq:gap_lower_bound}:} We distinguish two different cases: $\alpha>1-1/n$ and $\alpha\leq1-1/n$.

Suppose $\alpha>1-1/n$. In this case, we have $1/(1-\alpha)>n$, thus it suffices to construct an example such that $R^*_{det}\leq 2R^*/n$. We define $\epsilon=(1-2/e)n^{-1}$. The revenue and preference weights are defined as
\begin{equation*}
v_i=\Big(1-\dfrac{1}{n}\Big)^i\epsilon,\ r_i=1/v_i,\ \forall i\in\{1,2,\dots,n\}.
\end{equation*}
In this setting, if $v_i<v_j$, then $v_i/v_j\leq 1-1/n<\alpha$. Therefore any feasible solution to~\texttt{BMS-Deterministic}~contains at most one product. Then we get
\begin{equation*}
R^*_{det}=\max_{i\in\mathcal{N}}\dfrac{r_iv_i}{1+v_i}\leq \max_{i\in\mathcal{N}}r_iv_i=1.
\end{equation*}
We further define $\tilde{\bm{x}}$ as
\begin{equation*}
\tilde{x}_i=\dfrac{v_n}{1+nv_n},\ \forall i\in\mathcal{N}\ \text{and}\ \tilde{x}_0=\dfrac{1}{1+nv_n}.
\end{equation*}
It is easy to verify that $\tilde{x}_0+\sum_{i\in\mathcal{N}}\tilde{x}_i=1$. Since $\tilde{x}_i=\tilde{x}_j$ for all $i,j\in\mathcal{N}$, $\tilde{\bm{x}}$ satisfies the third set of constraints in~\ref{prob:randomized}. Since $v_i\geq v_n$ for all $i\in\mathcal{N}$, we get $\tilde{x}_i=v_n\tilde{x}_0\leq v_i\tilde{x}_0$ for all $i\in\mathcal{N}$. Therefore $\bm{x}$ is feasible to~\ref{prob:randomized}. Therefore
\begin{equation*}
R^*\geq\sum_{i\in\mathcal{N}}r_i\tilde{x}_i=\dfrac{v_n}{1+nv_n}\sum_{i=1}^n \dfrac{1}{v_i}\geq \dfrac{1}{2-2/e}\sum_{i=1}^n\Big(1-\dfrac{1}{n}\Big)^{n-i}=\dfrac{1}{2-2/e}\dfrac{1-(1-1/n)^n}{1/n}\geq \dfrac{n}{2}.
\end{equation*}
Here, the second inequality is due to $nv_n\leq 1-2/e$, the third inequality is due to $(1-1/n)^n<1/e$ for all $n\in\mathbb{Z}^+$. Therefore in this setting
\begin{equation*}
\dfrac{2}{n}R^*\geq 1\geq R^*_{det}.
\end{equation*}

Suppose $\alpha\leq 1-1/n$, then we get $1/(1-\alpha)\leq n$, thus it suffices to construct an example such that $R^*_{det}\leq 2(1-\alpha)R^*$. We assume $\epsilon$ satisfies $(1+n\epsilon)^2=2-2/e$. The revenue and preference weights are defined as
\begin{equation*}
v_i=\alpha^{i}(1-\epsilon)^i\epsilon,\ r_i=1/v_i,\ \forall i\in\{1,2,\dots,n\}.
\end{equation*}
In this setting, if $v_i<v_j$ then $v_i/v_j\leq(1-\epsilon)\alpha<\alpha$. Therefore any feasible solution to~\texttt{BMS-Deterministic}~contains at most one product. Then we get
\begin{equation*}
R^*_{det}=\max_{i\in\mathcal{N}}\dfrac{r_iv_i}{1+v_i}\leq \max_{i\in\mathcal{N}}r_iv_i=1.
\end{equation*}
We further define $\tilde{\bm{x}}$ as
\begin{equation*}
\tilde{x}_i=\dfrac{v_n}{1+nv_n},\ \forall i\in\mathcal{N},\ \tilde{x}_0=\dfrac{1}{1+nv_n}.
\end{equation*}
Similar to the previous case, $\tilde{\bm{x}}$ is feasible to~\ref{prob:randomized}. We further have
\begin{align*}
R^*\geq&\sum_{i\in\mathcal{N}}r_i\tilde{x}_i=\dfrac{v_n}{1+nv_n}\sum_{i=1}^n \dfrac{1}{v_i}\geq \dfrac{1}{1+n\epsilon}\sum_{i=1}^n (1-\epsilon)^{n-i}\alpha^{n-i}=\dfrac{1}{1+n\epsilon}\dfrac{1-(1-\epsilon)^n\alpha^n}{1-(1-\epsilon)\alpha}\\
\geq&\dfrac{1}{1+n\epsilon}\dfrac{1-(1-1/n)^n}{1-\alpha+\epsilon}\geq\dfrac{1}{1+n\epsilon}\dfrac{1-1/e}{1-\alpha+\epsilon}=\dfrac{1}{1+n\epsilon}\dfrac{1-1/e}{1-\alpha}\dfrac{1-\alpha}{1-\alpha+\epsilon}\\
\geq&\dfrac{1}{1+n\epsilon}\dfrac{1-1/e}{1-\alpha}\dfrac{1/n}{1/n+\epsilon}=\dfrac{1}{(1+n\epsilon)^2}\dfrac{1-1/e}{1-\alpha}=\dfrac{1}{2-2/e}\dfrac{1-1/e}{1-\alpha}=\dfrac{1}{2(1-\alpha)}.
\end{align*}
Here, the second inequality is due to $v_n\leq \epsilon$. The third inequality is due to $\alpha\leq 1-1/n$. The fourth inequality is due to $(1-1/n)^n<1/e$ for all $n\in\mathbb{Z}^+$. The fifth inequality is due to $\alpha\leq 1-1/n$. Therefore we have that
\begin{equation}
2(1-\alpha)R^*\geq 1\geq R^*_{det}.
\end{equation}
Therefore we conclude that for any $0<\alpha\leq 1$ and any $n\in\mathbb{Z}^+$, there exists an instance with $n$ products such that~\eqref{eq:gap_lower_bound}~is satisfied.

\subsection{Example for the Number of Assortments Randomized Over}\label{sec:example:value_randomization}

In this section, we provide an example to show that for any $\alpha>0$, there exists an instance such that the optimal solution to~\ref{prob:randomized_0} randomizes over $n$ assortments.

\begin{example}\label{example:value_randomization}
Consider any $n\geq 2$ and $\alpha>0$. We set $\mathcal{N}=\{1,2,\dots,n\}$, $r_i=1$ for all $i\in\mathcal{N}$, and the preference weights are defined as
\begin{equation*}
v_1=\alpha,\ v_i=1+i\alpha/n^2,\ \forall i\in\{2,3,\dots,n\}.
\end{equation*}

Suppose product $1$ is offered with positive probability. By the second set of constraints in \ref{prob:randomized} we get $x_1\leq v_1x_0=\alpha x_0$. Furthermore, since we have assumed $x_1>0$, by the market share balancing constraint we get $x_i\leq x_1/\alpha\leq \alpha x_0/\alpha=x_0$ for all $i\in\{2,3,\dots,n\}$. Therefore we have \mbox{$(n+\alpha)x_0\geq x_0+\sum_{i\in\mathcal{N}}x_i=1$}, which implies $x_0\geq 1/(n+\alpha)$. Since $r_i=1$ for all $i\in\mathcal{N}$, the expected revenue is upper bounded by $\sum_{i\in\mathcal{N}}x_i=1-x_0\leq (n+\alpha-1)/(n+\alpha)$, and equality holds if and only if $x_1=\alpha/(n+\alpha)$ and $x_0=x_i=1/(n+\alpha)$ for all $i\neq 1$. Furthermore, if we set $\tilde{x}_1=\alpha/(n+\alpha)$ and $\tilde{x}_0=\tilde{x}_i=1/(n+\alpha)$ for all $i\in\{2,3,\dots,n\}$, then $\tilde{\bm{x}}$ is feasible to \ref{prob:randomized} and $\sum_{i\in\mathcal{N}}\tilde{x}_i=(n+\alpha-1)/(n+\alpha)$. Therefore the maximum expected revenue when product $1$ is offered is $(n+\alpha-1)/(n+\alpha)$.
 
If product $1$ is not offered, we can consider $\{2,3,\dots,n\}$ as the universe of products, and the corresponding optimal unconstrained assortment is $\{2,3,\dots,n\}$, which generates an expected revenue of
\begin{equation*}
\dfrac{\sum_{i=2}^n v_i}{1+\sum_{i=2}^n v_i}=\dfrac{n-1+(n+2)(n-1)\alpha/2n^2}{n+(n+2)(n-1)\alpha/2n^2}<\dfrac{n+\alpha-1}{n+\alpha},
\end{equation*}
where the inequality is due to $(n+2)(n-1)/2<n^2$ for all $n\geq 2$. In the presence of market share balancing constraint, the maximal expected revenue when product $1$ is not offered can be even smaller. Therefore the optimal value of \ref{prob:randomized} is $(n+1-\alpha)/(n+\alpha)$, and the unique optimal solution to \ref{prob:randomized} is $\tilde{x}_1=\alpha/(n+\alpha)$ and $\tilde{x}_i=1/(n+\alpha)$ for all $i\in\{2,3,\dots,n\}$. 

Given an optimal solution $\tilde{\bm{x}}$ of \ref{prob:randomized}, if we construct a distribution over assortments following Equation~\eqref{eq:sales_to_distribution}, then the number of assortments the distribution randomizes over is equal to the number of distinct values of $\tilde{x}_i/v_i$ for all $i\in\mathcal{N}$. Since in this setting there are $n$ distinct values of $\tilde{x}_i/v_i$ for the optimal solution $\tilde{\bm{x}}$, the optimal solution to~\ref{prob:randomized_0}~randomizes over $n$ assortments.
\end{example}

\subsection{Proof of Lemma~\ref{lemma:randomized_linearized}}\label{sec:lemma:randomized_linearized}

We first prove that the optimal value of~Problem~\eqref{prob:randomized_linearized}~is greater than or equal to $R^*$. Consider any optimal solution $\hat{\bm{x}}$ to~\ref{prob:randomized}. If $\hat{x}_0=0$, then by the second set of constraints in~\ref{prob:randomized}~we have $\hat{x}_i=0$ for all $i\in\mathcal{N}$, which contradicts with the first constraint of \ref{prob:randomized}, i.e., \mbox{$\hat{x}_0+\sum_{i\in\mathcal{N}}\hat{x}_i=1$.} Therefore we must have $\hat{x}_0>0$. We define $\hat{\bm{w}}$ as $\hat{w}_i$ as $\hat{w}_i=\hat{x}_i/\hat{x}_0$ for all $i\in\mathcal{N}$. Since $\hat{\bm{x}}$ is a feasible solution of \ref{prob:randomized}, by the second set of constraints in~\ref{prob:randomized}~we have $\hat{w}_i\leq v_i$ for all $i\in\mathcal{N}$, and by the third set of constraints in~\ref{prob:randomized}~we have that for all $i\in\mathcal{N}$,
\begin{equation*}
\hat{w}_i\in\{0\}\cup\Big[\alpha\cdot\max_{j\in\mathcal{N}}\hat{w}_j,\infty\Big).
\end{equation*}
Therefore $\hat{\bm{w}}$ is a feasible solution to~Problem~\eqref{prob:randomized_linearized}. By the first constraint of~\ref{prob:randomized}~we get
\begin{equation*}
1=\hat{x}_0+\sum_{i\in\mathcal{N}}\hat{x}_i=\hat{x}_0+\sum_{i\in\mathcal{N}}\hat{w}_i\hat{x}_0=\Big(1+\sum_{i\in\mathcal{N}}\hat{w}_i\Big)\hat{x}_0.
\end{equation*}
Thus $\hat{x}_0=1/(1+\sum_{i\in\mathcal{N}}\hat{w}_i)$, and for any $i\in\mathcal{N}$, $\hat{x}_i=\hat{w}_i/(1+\sum_{i\in\mathcal{N}}\hat{w}_i)$. Since $\hat{\bm{x}}$ is an optimal solution to~\ref{prob:randomized}, we get
\begin{equation*}
R^*=\sum_{i\in\mathcal{N}}r_i\hat{x}_i=\dfrac{\sum_{i\in\mathcal{N}}r_i\hat{w}_i}{1+\sum_{i\in\mathcal{N}}\hat{w}_i},
\end{equation*}
which implies
\begin{equation*}
\sum_{i\in\mathcal{N}}(r_i-R^*)\hat{w}_i=R^*.
\end{equation*}
Therefore we conclude that $\hat{\bm{w}}$ is feasible to~Problem~\eqref{prob:randomized_linearized}~and the objective value of $\hat{\bm{w}}$ is equal to $R^*$, thus the objective value of~Problem~\eqref{prob:randomized_linearized}~is greater than or equal to that of~\ref{prob:randomized}.

Next, we prove that the optimal value of~Problem~\eqref{prob:randomized_linearized}~is equal to $R^*$. Consider any optimal solution $\bm{w}^*$ to~Problem~\eqref{prob:randomized_linearized}. Since we have proven that the optimal value of~Problem~\eqref{prob:randomized_linearized}~is greater than or equal to that of~\ref{prob:randomized}, we have
$\sum_{i\in\mathcal{N}}(r_i-R^*)w_i^*\geq R^*$, which implies
\begin{equation}
\dfrac{\sum_{i\in\mathcal{N}}r_iw_i^*}{1+\sum_{i\in\mathcal{N}}w_i^*}\geq R^*.\label{eq:value_x}
\end{equation}
We define $\bm{x}^*$ using~\eqref{eq:reconstruct_x}. It is easy to verify that $x_0^*+\sum_{i\in\mathcal{N}}x_i^*=1$. Since $\bm{w}^*$ is feasible to~Problem~\eqref{prob:randomized_linearized}, by the first set of constraints in~Problem~\eqref{prob:randomized_linearized}~we get $x_i^*\leq v_ix_0^*$ for any $i\in\mathcal{N}$, and by the second set of constraints in~Problem~\eqref{prob:randomized_linearized}~we have that for any $i\in\mathcal{N}$,
\begin{equation*}
x_i^*\in\{0\}\cup\Big[\alpha\cdot\max_{j\in\mathcal{N}}x_j^*,\infty\Big).
\end{equation*}
Therefore $\bm{x}^*$ is feasible to~\ref{prob:randomized}. Furthermore, by~\eqref{eq:value_x}~we get $\sum_{i\in\mathcal{N}}r_ix_i^*\geq R^*$, thus $\bm{x}^*$ is an optimal solution to~\ref{prob:randomized}. This further implies that $\sum_{i\in\mathcal{N}}r_ix_i^*=R^*$, thus
\begin{equation*}
\dfrac{\sum_{i\in\mathcal{N}}r_iw_i^*}{1+\sum_{i\in\mathcal{N}}w_i^*}=R^*.
\end{equation*}
Then we get $\sum_{i\in\mathcal{N}}(r_i-R^*)w_i^*=R^*$. Therefore the optimal value of~Problem~\eqref{prob:randomized_linearized}~is equal to $R^*$, and given any optimal solution $\bm{w}^*$ to~Problem~\eqref{prob:randomized_linearized}, we obtain an optimal solution $\bm{x}^*$ to~\ref{prob:randomized}~using~\eqref{eq:reconstruct_x}.

\subsection{Proof of Lemma~\ref{lemma:randomized_constrained_linearized}}\label{sec:lemma:randomized_constrained_linearized}

We first prove that the optimal value of~Problem~\eqref{prob:randomized_constrained_linearized}~is greater than or equal to $R_{\mathcal{X}}^*$. Consider any optimal solution $\hat{\bm{x}}$ to~\ref{prob:randomized_constrained}. We define $\hat{\bm{w}}$ as $\hat{w}_i=\hat{x}_i/\hat{x}_0$ for each $i\in\mathcal{N}$. Following similar arguments in the proof of Lemma~\ref{lemma:randomized_linearized}~we have that $\hat{\bm{w}}$ satisfies the first and second sets of constraints in Problem~\eqref{prob:randomized_constrained_linearized}. We further have $\{i\in\mathcal{N}:\hat{w}_i>0\}=\{i\in\mathcal{N}:\hat{x}_i>0\}\in\mathcal{X}$. Therefore $\hat{\bm{w}}$ is a feasible solution to~Problem~\eqref{prob:randomized_constrained_linearized}. Using similar arguments in the proof of Lemma~\ref{lemma:randomized_linearized} we get
\begin{equation*}
\sum_{i\in\mathcal{N}}(r_i-R_{\mathcal{X}}^*)\hat{w}_i=R_{\mathcal{X}}^*.
\end{equation*}
Therefore we conclude that $\hat{\bm{w}}$ is feasible to~Problem~\eqref{prob:randomized_constrained_linearized}~and the objective value of $\hat{\bm{w}}$ is equal to $R_{\mathcal{X}}^*$, thus the optimal value of~Problem~\eqref{prob:randomized_constrained_linearized}~is greater than or equal to that of~\ref{prob:randomized_constrained}.

Next, we prove that the optimal value of~Problem~\eqref{prob:randomized_constrained_linearized}~is equal to $R^*_{\mathcal{X}}$. Consider any optimal solution $\bm{w}^*$ to~Problem~\eqref{prob:randomized_constrained_linearized}. Since we have proven that the optimal value of~Problem~\eqref{prob:randomized_constrained_linearized}~is greater than or equal to that of~\ref{prob:randomized_constrained}, we have
$\sum_{i\in\mathcal{N}}(r_i-R_{\mathcal{X}}^*)w_i^*\geq R_{\mathcal{X}}^*$, which implies
\begin{equation}
\dfrac{\sum_{i\in\mathcal{N}}r_iw_i^*}{1+\sum_{i\in\mathcal{N}}w_i^*}\geq R_{\mathcal{X}}^*.\label{eq:value_x_2}
\end{equation}
We define $\bm{x}^*$ using~\eqref{eq:reconstruct_x}. Using the same steps in the proof of Lemma~\ref{lemma:randomized_linearized}~we have that $\bm{x}^*$ satisfies the first three sets of constraints in~\ref{prob:randomized_constrained}. We further have $\{i\in\mathcal{N}:x_i^*>0\}=\{i\in\mathcal{N}:w_i^*>0\}\in\mathcal{X}$. Therefore $\bm{x}^*$ is feasible to~\ref{prob:randomized_constrained}. Furthermore, by~\eqref{eq:value_x_2}~we get $\sum_{i\in\mathcal{N}}r_ix_i^*\geq R_{\mathcal{X}}^*$, thus $\bm{x}^*$ is an optimal solution to~\ref{prob:randomized_constrained}. This further implies that $\sum_{i\in\mathcal{N}}r_ix_i^*=R^*_{\mathcal{X}}$, thus
\begin{equation*}
\dfrac{\sum_{i\in\mathcal{N}}r_iw_i^*}{1+\sum_{i\in\mathcal{N}}w_i^*}=R_{\mathcal{X}}^*.
\end{equation*}
Then we get $\sum_{i\in\mathcal{N}}(r_i-R_{\mathcal{X}}^*)w_i^*=R_{\mathcal{X}}^*$. Therefore the optimal value of~Problem~\eqref{prob:randomized_constrained_linearized}~is $R^*_{\mathcal{X}}$, and given an optimal solution $\bm{w}^*$ to~Problem~\eqref{prob:randomized_constrained_linearized}, we obtain an optimal solution $\bm{x}^*$ to~\ref{prob:randomized_constrained}~using~\eqref{eq:reconstruct_x}.

\subsection{Proof of Lemma~\ref{lemma:w_structure}}\label{sec:lemma:w_structure}

If $\mathcal{X}=\{\varnothing\}$, then $\bm{w}^*=0$ is optimal to Problem~\eqref{prob:randomized_constrained_linearized}, and the conditions in the lemma is satisfied by setting $S^*=\varnothing$ and $\underline{v}=v_1$. Thus for the rest of the proof, we assume $\mathcal{X}$ contains at least one nonempty subset of $\mathcal{N}$.

We use the variable $\bar{S}$ to capture $\{i\in\mathcal{N}:w_i>0\}$ and  the variable $v$ to capture $\alpha$ times the maximum element of $\bm{w}$. In this way we rewrite~Problem~\eqref{prob:randomized_constrained_linearized}~in the following equivalent form
\begin{equation}\label{prob:randomized_constrained_linearized_2}
\begin{aligned}
&\max_{\bar{S}\subseteq\mathcal{N},\bm{w},v}&&\sum_{i\in\bar{S}}(r_i-R_{\mathcal{X}}^*)w_i\\
&\text{s.t.}&&0\leq w_i\leq v_i,\ \forall i\in\bar{S},\\
&&&v\leq w_i\leq v/\alpha,\ \forall i\in\bar{S},\\
&&&\bar{S}\in\mathcal{X}.
\end{aligned}
\end{equation}
The first set of constraints in~Problem~\eqref{prob:randomized_constrained_linearized_2}~follows directly from the first set of constraints in~Problem~\eqref{prob:randomized_constrained_linearized}. We further show that the second set of constraints in~Problem~\eqref{prob:randomized_constrained_linearized_2}~is equivalent to the second set of constraints in~Problem~\eqref{prob:randomized_constrained_linearized}. Suppose $\bm{w}$ is feasible to~Problem~\eqref{prob:randomized_constrained_linearized}, then by setting $\bar{S}=\{i\in\mathcal{N}:w_i>0\}$ and $v=\alpha\cdot\max_{j\in\mathcal{N}}w_j$, by the second set of constraints in~Problem~\eqref{prob:randomized_constrained_linearized}~we have that for all $i\in\bar{S}$,
\begin{equation*}
v=\alpha\cdot\max_{j\in\mathcal{N}}w_j\leq w_i\leq \max_{j\in\mathcal{N}}w_j=v/\alpha.
\end{equation*}
On the other hand, if $(\bar{S},\bm{w},v)$ is feasible to~Problem~\eqref{prob:randomized_constrained_linearized_2}, we set $w_i=0$ for all $i\notin\bar{S}$, then by the last set of constraints in~Problem~\eqref{prob:randomized_constrained_linearized_2}~we get $v/\alpha\geq \max_{j\in\mathcal{N}}w_j$, thus for all $i\in\bar{S}$, $w_i\geq v\geq \alpha\cdot\max_{j\in\mathcal{N}}w_j$. This implies that $\bm{w}$ satisfies the second set of constraints in~Problem~\eqref{prob:randomized_constrained_linearized}. Finally, if $w_i=0$ for some $i\in\bar{S}$, then by the second set of constraints in~Problem~\eqref{prob:randomized_constrained_linearized_2}~we get $v=0$, thus $w_i=0$ for all $i\in\bar{S}$, which is obviously suboptimal. Therefore we can focus on the case where $w_i>0$ for all $i\in\bar{S}$, and in this case by setting $w_i=0$ for all $i\notin\bar{S}$, we get $\bar{S}=\{i\in\mathcal{N}:w_i>0\}$, thus the last constraint in~Problem~\eqref{prob:randomized_constrained_linearized_2}~is equivalent to the last constraint in~Problem~\eqref{prob:randomized_constrained_linearized}. Therefore we conclude that~Problem~\eqref{prob:randomized_constrained_linearized_2}~is equivalent to~Problem~\eqref{prob:randomized_constrained_linearized}.

Consider any optimal solution $(S^*,\tilde{\bm{w}},\tilde{v})$ to~Problem~\eqref{prob:randomized_constrained_linearized_2}. We fix $\bar{S}=S^*$ and only optimize over $\bm{w}$ and $v$, then the following problem is equivalent to~Problem~\eqref{prob:randomized_constrained_linearized_2},
\begin{equation}\label{prob:randomized_constrained_linearized_3}
\begin{aligned}
&\max_{\bm{w},v}&&\sum_{i\in S^*}(r_i-R^*_{\mathcal{X}})w_i\\
&\text{s.t.}&&0\leq w_i\leq v_i,\ \forall i\in S^*,\\
&&& v\leq w_i\leq v/\alpha,\ \forall i\in S^*.
\end{aligned}
\end{equation}
Suppose $(\hat{\bm{w}},\underline{v})$ is a basic feasible solution that is optimal to~Problem~\eqref{prob:randomized_constrained_linearized_3}. If $\hat{w}_i=0$ for some $i\in S^*$, then $\underline{v}=0$, thus $\hat{w}_j=0$ for all $j\in S^*$, and the corresponding objective value is zero. To prove that $\hat{\bm{w}}=0$ is not optimal, we claim that $R^*_{\mathcal{X}}>0$. Since we have assumed that $\mathcal{X}$ contains at least one nonempty subset of $\mathcal{N}$, there exists $S_0\in\mathcal{X}$ such that $S_0\neq \varnothing$. We define $K_0=|S_0|$, and define $\tilde{\bm{x}}$ as $\tilde{x}_i=v_{\min}/(1+K_0v_{\min})$ for all $i\in S_0$, $\tilde{x}_i=0$ for all $i\notin S_0$ and $\tilde{x}_0=1/(1+K_0v_{\min})$. Then we have $\tilde{x}_0+\sum_{i\in\mathcal{N}}\tilde{x}_i=1$, and $\tilde{x}_i\leq v_{\min}\tilde{x}_0\leq v_i\tilde{x}_0$ for all $i\in\mathcal{N}$. Since $\tilde{x}_i\in\{0,v_{\min}/(1+K_0v_{\min})\}$ for all $i\in\mathcal{N}$, $\tilde{\bm{x}}$ satisfies the third set of constraints of \ref{prob:randomized_constrained}. Finally $S_0=\{i\in\mathcal{N}:\tilde{x}_i>0\}\in\mathcal{X}$, thus $\tilde{\bm{x}}$ is feasible to \ref{prob:randomized_constrained}. Then we have $R_{\mathcal{X}}^*\geq \sum_{i\in\mathcal{N}}r_i\tilde{x}_i>0$. Therefore the optimal value of Problem~\eqref{prob:randomized_constrained_linearized} is not zero, and we conclude that the optimal solution $(\hat{\bm{w}},\underline{v})$ satisfies $\hat{w}_i>0$ for all $i\in S^*$.

Letting $K=|S^*|$, since $(\hat{\bm{w}},\underline{v})$ is a basic feasible solution of Problem~\eqref{prob:randomized_constrained_linearized_3}, and the dimension of the variables is $K+1$, equality should hold in at least $K+1$ inequality constraints of Problem~\eqref{prob:randomized_constrained_linearized_3}. We define \mbox{$A_1=\{i\in S^*:\hat{w}_i=v_i\}$,} \mbox{$A_2=\{i\in S^*:\hat{w}_i=\underline{v}/\alpha\}\cup\{i\in S^*:\hat{w}_i=\underline{v}\}$,} then we get \mbox{$|A_1|+|A_2|\geq K+1$.} Suppose $A_1\cap A_2=\varnothing$, then $|A_1|+|A_2|=|A_1\cup A_2|\leq |S^*|=K$, which contradicts with $|A_1|+|A_2|\geq K+1$. Therefore we get $A_1\cap A_2\neq\varnothing$, thus there exists $k\in A_1\cap A_2$. Since $k\in A_1$, we have $v_k=\hat{w}_k$. Since $k\in A_2$, we have that either $\hat{w}_k=\underline{v}$ or $\hat{w}_k=\underline{v}/\alpha$. Therefore either $v_k=\hat{w}_k=\underline{v}$, or $v_k=\hat{w}_k=\underline{v}/\alpha$. In both cases we have $\underline{v}\in\{v_i:i\in\mathcal{N}\}\cup\{\alpha v_i:i\in\mathcal{N}\}$. Furthermore, for all $i\in S^*$, we get $v_i\geq \hat{w}_i\geq \underline{v}$, thus $v_i\geq \underline{v}$ for all $i\in S^*$.

We further fix $\bar{S}=S^*$ and $v=\underline{v}$ and optimize only over $\bm{w}$, then we obtain the following problem equivalent to~Problem~\eqref{prob:randomized_constrained_linearized},
\begin{equation}\label{prob:randomized_constrained_linearized_4}
\begin{aligned}
&\max_{\bm{w}}&&\sum_{i\in S^*}(r_i-R_{\mathcal{X}}^*)w_i\\
&\text{s.t.}&&\underline{v}\leq w_i\leq \min\{v_i,\underline{v}/\alpha\},\ \forall i\in S^*.
\end{aligned}
\end{equation}
In~Problem~\eqref{prob:randomized_constrained_linearized_4}, the variable $\bm{w}$ is separable for each $i\in S^*$. It is easy to see that for each $i\in S^*$, setting $w_i^*=\min\{v_i,\underline{v}/\alpha\}$ if $r_i\geq R^*_{\mathcal{X}}$ and $w_i^*=\underline{v}$ if $r_i<R^*_{\mathcal{X}}$ is optimal to~Problem~\eqref{prob:randomized_constrained_linearized_4}. Thus the $\bm{w}^*$ defined in~\eqref{eq:w_structure}~is optimal to~Problem~\eqref{prob:randomized_constrained_linearized_4}. Then we conclude that there exists an optimal solution $\bm{w}^*$ to~Problem~\eqref{prob:randomized_constrained_linearized}~and $\underline{v}>0$, $S^*\in\mathcal{X}$ that satisfies the conditions stated in Lemma~\ref{lemma:w_structure}.

\subsection{Proof of Lemma~\ref{lemma:transformed_mnl}}\label{sec:lemma:transformed_mnl}

Suppose $\underline{v}$ and $S^*$ satisfies the conditions stated in Lemma~\ref{lemma:w_structure}, and we define $\bm{w}^*$ using~\eqref{eq:w_structure}. Since by the conditions in Lemma~\ref{lemma:w_structure}~we have $v_i\geq \underline{v}$ for all $i\in S^*$, we get $S^*\subseteq S^+(\underline{v})$. The conditions in Lemma~\ref{lemma:w_structure} also state that $S^*\in\mathcal{X}$. Therefore $S^*$ is feasible to~\ref{prob:transformed_mnl}. In what follows, we prove that the objective value of $S^*$ in~\ref{prob:transformed_mnl}~is $R^*_{\mathcal{X}}$. By definition of $\bm{v}'$ we have $w_i^*=v'_i$ for all $i\in S^*$ and $w_i^*=0$ for all $i\notin S^*$. Since $\bm{w}^*$ is optimal to~Problem~\eqref{prob:randomized_constrained_linearized}, and by Lemma~\ref{lemma:randomized_constrained_linearized}~the optimal value of~Problem~\eqref{prob:randomized_constrained_linearized}~is $R^*_{\mathcal{X}}$, we get
\begin{equation*}
\sum_{i\in S^*}(r_i-R_{\mathcal{X}}^*)v_i'=\sum_{i\in\mathcal{N}}(r_i-R_{\mathcal{X}}^*)w_i^*=R_{\mathcal{X}}^*,
\end{equation*}
which is equivalent to
\begin{equation*}
\dfrac{\sum_{i\in S^*}r_iv_i'}{1+\sum_{i\in S^*}v_i'}=R_{\mathcal{X}}^*.
\end{equation*}

Then we prove that $R_{\mathcal{X}}^*$ is the optimal value of~\ref{prob:transformed_mnl}. Consider any feasible solution $S$ of~\ref{prob:transformed_mnl}. We define $\tilde{\bm{w}}$ as $\tilde{w}_i=v_i'$ for all $i\in S$ and $\tilde{w}_i=0$ for all $i\notin S$. We prove that $\tilde{\bm{w}}$ is feasible to~Problem~\eqref{prob:randomized_constrained_linearized}. We first have $\{i\in\mathcal{N}:\tilde{w}_i>0\}=S\in\mathcal{X}$. Consider any $i\in S$, if $r_i<R_{\mathcal{X}}^*$, since $S\subseteq S^+(\underline{v})$, we have $v_i\geq \underline{v}=\tilde{w}_i$; if $r_i\geq R_{\mathcal{X}}^*$ then we have $v_i\geq \min\{v_i,\underline{v}/\alpha\}=\tilde{w}_i$. Therefore $\tilde{\bm{w}}$ satisfies $\tilde{w}_i\leq v_i$ for all $i\in\mathcal{N}$. Since $S\subseteq S^+(\underline{v})$, we have $v_i\geq \underline{v}$ for all $i\in S$, thus $\underline{v}\leq \tilde{w}_i\leq \underline{v}/\alpha$. This further implies $\tilde{w}_i\geq \underline{v}\geq \alpha\cdot\max_{j\in\mathcal{N}}\tilde{w}_j$ for all $i\in S$, thus $\tilde{\bm{w}}$ satisfies the second set of constraints in~Problem~\eqref{prob:randomized_constrained_linearized}. Therefore we conclude that $\tilde{\bm{w}}$ is feasible to~Problem~\eqref{prob:randomized_constrained_linearized}. Since by Lemma~\ref{lemma:randomized_constrained_linearized} the optimal value of~Problem~\eqref{prob:randomized_constrained_linearized}~is $R^*_{\mathcal{X}}$, we get
\begin{equation*}
\sum_{i\in S}(r_i-R_{\mathcal{X}}^*)v'_i=\sum_{i\in \mathcal{N}}(r_i-R_{\mathcal{X}}^*)\tilde{w}_i\leq R_{\mathcal{X}}^*,
\end{equation*}
which is equivalent to
\begin{equation*}
\dfrac{\sum_{i\in S}r_iv_i'}{1+\sum_{i\in S}v_i'}\leq R_{\mathcal{X}}^*.
\end{equation*}
Therefore, the optimal value of~\ref{prob:transformed_mnl}~is $R_{\mathcal{X}}^*$, and $S^*$ is an optimal solution to~\ref{prob:transformed_mnl}.

\section{Proofs for the Dynamic Problem}

\subsection{Proof of Lemma~\ref{lemma:inventory_upper_bound}}\label{sec:lemma:inventory_upper_bound}

Consider any dynamic control policy $\pi$. We use $\pi_t(\bm{H}_{t-1},\cdot)$ to denote the distribution over assortments returned by policy $\pi$ at time period $t$ given history of sales $\bm{H}_{t-1}$. In other words, given history of sales $\bm{H}_{t-1}$, the probability of offering assortment $S$ at time period $t$ is $\pi_t(\bm{H}_{t-1},S)$. For any $S\subseteq\mathcal{N}$, we define $h^\pi(S)$ as the total expected number of times in which assortment $S$ is offered, given by
\begin{equation*}
h^\pi(S)=\sum_{t=1}^T\mathbb{E}[\pi_t(\bm{H}_{t-1},S)].
\end{equation*}
Here the expectation is taken over the randomness in $\bm{H}_{t-1}$. We claim that for all $i\in\mathcal{N}$,
\begin{equation}\label{eq:expected_sales}
\mathbb{E}[X_{iT}^\pi]=\sum_{S\subseteq\mathcal{N}}\phi(i,S)h^\pi(S).
\end{equation}
By definition of $\pi_t(\bm{H}_{t-1},\cdot)$, we first have that for all $i\in\mathcal{N}$ and $t\in\{1,2,\dots,T\}$,
\begin{equation*}
\mathbb{E}[X^\pi_{it}-X^\pi_{i,t-1}|\bm{H}_{t-1}]=\sum_{S\subseteq\mathcal{N}}\phi(i,S)\pi_t(\bm{H}_{t-1},S).
\end{equation*}
Taking the expectation over $\bm{H}_{t-1}$ we get
\begin{equation*}
\mathbb{E}[X^\pi_{it}-X_{i,t-1}^\pi]=\sum_{S\subseteq\mathcal{N}}\phi(i,S)\mathbb{E}[\pi_t(\bm{H}_{t-1},S)].
\end{equation*}
Taking the sum over $t\in\{1,2,\dots,T\}$ we get
\begin{equation*}
\mathbb{E}[X_{iT}^\pi]=\sum_{t=1}^T\mathbb{E}[X_{it}^\pi-X_{i,t-1}^\pi]=\sum_{S\subseteq\mathcal{N}}\sum_{t=1}^T\phi(i,S)\mathbb{E}[\pi_t(\bm{H}_{t-1},S)]=\sum_{S\subseteq\mathcal{N}}\phi(i,S)h^\pi(S).
\end{equation*}

We define $\tilde{\bm{x}}$ as $\tilde{x}_i=\mathbb{E}[X_{iT}^\pi]/T$ for all $i\in\mathcal{N}$ and $\tilde{x}_0=\sum_{S\subseteq\mathcal{N}}\phi(0,S)h^\pi(S)/T$. Then by~\eqref{eq:expected_sales}~we have
\begin{equation*}
\tilde{x}_0+\sum_{i\in\mathcal{N}}\tilde{x}_i=\sum_{S\subseteq\mathcal{N}}\Big(\phi(0,S)+\sum_{i\in\mathcal{N}}\phi(i,S)\Big)h^\pi(S)/T=\sum_{S\subseteq\mathcal{N}}h^\pi(S)/T=\sum_{t=1}^T\sum_{S\subseteq\mathcal{N}}\mathbb{E}[\pi_t(\bm{H}_{t-1},S)]/T=1.
\end{equation*}
Therefore $\tilde{\bm{x}}$ satisfies the first constraint of~the \ref{prob:inventory_upper_bound} problem. Furthermore, we have that for all $i\in\mathcal{N}$,
\begin{equation*}
\tilde{x}_i=\sum_{S\subseteq\mathcal{N}}\phi(i,S)h^\pi(S)/T\leq \sum_{S\subseteq\mathcal{N}}v_i\phi(0,S)h^\pi(S)/T=v_i\tilde{x}_0.
\end{equation*}
Here the inequality holds because $\phi(i,S)\leq v_i\phi(0,S)$ for all $S\subseteq\mathcal{N}$. This is because $\phi(i,S)=0$ if $i\notin S$, and $\phi(i,S)=v_i\phi(0,S)$ if $i\in S$. Therefore $\bm{x}$ satisfies the second set of constraints of~the \ref{prob:inventory_upper_bound} problem. Since $\mathbb{P}(X_{iT}^\pi\leq c_i)=1$ for all $i\in\mathcal{N}$, we have $\mathbb{E}[X_{iT}^\pi]\leq c_i$, thus $\tilde{x}_i=\mathbb{E}[X_{iT}^\pi]/T\leq c_i/T$ for all $i\in\mathcal{N}$. Therefore $\tilde{\bm{x}}$ satisfies the third set of constraints in~the \ref{prob:inventory_upper_bound} problem. By definition of $\tilde{\bm{x}}$, the second set of constraints in~\ref{prob:inventory}~immediately implies the last set of constraints in~the \ref{prob:inventory_upper_bound} problem. Therefore $\tilde{\bm{x}}$ is feasible to~the \ref{prob:inventory_upper_bound} problem. We further have $\sum_{i\in\mathcal{N}}r_i\mathbb{E}[X_{iT}^\pi]=T\sum_{i\in\mathcal{N}}r_i\tilde{x}_i$. Therefore the optimal value of~the \ref{prob:inventory_upper_bound} problem~is an upper bound of that of~\ref{prob:inventory}.

\subsection{Proof of Proposition~\ref{thm:inventory_upper_bound}}\label{sec:thm:inventory_upper_bound}

We prove the NP-hardness of~the \ref{prob:inventory_upper_bound} problem~by using a reduction from the cardinality constrained subset sum problem. We start by presenting the subset sum problem.

\noindent\textbf{Cardinality Constrained Subset Sum Problem:} Given elements $\{1,2,\dots,2m\}$, each element $i$ is associated with a weight $a_i\in\mathbb{Z}^+$ such that $\sum_{i=1}^{2m}a_i=2Q$ for some $Q\in\mathbb{Z}^+$. The goal of the problem is to determine whether there exists $A\subseteq\{1,2,\dots,2m\}$ such that $|A|=m$ and $\sum_{i\in A}a_i=Q$. The cardinality constrained subset sum problem is NP-hard (see Appendix 3.2 of \citealt{garey1979computers}).

Consider any instance of the cardinality constrained subset sum problem. We assume that $m\geq 2$. Based on the cardinality constrained subset sum problem, we construct an instance of the \ref{prob:inventory_upper_bound} problem as follows. We set $\mathcal{N}=\{1,2,\dots,2m,2m+1\}$, and define $\epsilon=1/4mQ$, $M=10m^2/\epsilon^2$ and $B=2M-m-\epsilon Q$. The revenue and preference weights of products are defined as
\begin{equation*}
r_i=B-\dfrac{1+\epsilon a_i}{1+\epsilon^2 a_i},\ v_i=1+\epsilon^2 a_i,\ \forall i\in\{1,2,\dots,2m\},\ r_{2m+1}=M+B,\ v_{2m+1}=2.
\end{equation*}
We set $\alpha=1/2$. We define $c_{2m+1}'$ as
\begin{equation*}
c_{2m+1}'=\dfrac{2}{3+m+\epsilon^2Q}.
\end{equation*}
We further set $T$ as an positive integer such that $Tc_{2m+1}'$ is an integer. Such a $T$ exists because $c_{2m+1}'$ is rational. We set $c_{2m+1}=Tc_{2m+1}'$ and $c_i=T$ for all $i\in\{1,2,\dots,2m\}$.

Next, we prove that the optimal value of~the \ref{prob:inventory_upper_bound} problem~is greater than or equal to $TB$ if and only if there exists $A\subseteq\{1,2,\dots,2m\}$ such that $|A|=m$ and $\sum_{i\in A}a_i=Q$.

Suppose there exists $A\subseteq\{1,2,\dots,2m\}$ such that $|A|=m$ and $\sum_{i\in A}a_i=Q$. We define $\bm{x}$ such that
\begin{equation*}
x_i=\phi(i,A\cup\{2m+1\}),\ \forall i\in\mathcal{N}\cup\{0\}.
\end{equation*}
It is easy to verify that $\bm{x}$ satisfies the first and second set of constraints in~the \ref{prob:inventory_upper_bound} problem. Since $c_i=T$ for all $i\in\{1,2,\dots,2m\}$, we have $x_i\leq 1=c_i/T$ for all $i\in\{1,2,\dots,2m\}$. Since by definition $\sum_{i\in A}a_i=Q$ and $|A|=m$, we have
\begin{equation*}
1+v_{2m+1}+\sum_{j\in A}v_j=1+2+m+\epsilon^2\sum_{i\in A}a_i=3+m+\epsilon^2Q.
\end{equation*}
Then we get
\begin{equation*}
x_{2m+1}=\dfrac{v_{2m+1}}{1+v_{2m+1}+\sum_{j\in A}v_j}=\dfrac{2}{3+m+\epsilon^2Q}=c_{2m+1}'=c_{2m+1}/T.
\end{equation*}
Therefore $\bm{x}$ satisfies the third set of constraints in~the \ref{prob:inventory_upper_bound} problem. By definition we have $v_{2m+1}=2$, and for all $i\in\{1,2,\dots,2m\}$, we have
\begin{equation*}
1\leq v_{i}=1+\epsilon^2 a_i\leq 1+2\epsilon^2Q=1+\epsilon/2m\leq 2.
\end{equation*}
Here the second inequality is because $a_i\leq\sum_{j=1}^m a_j=2Q$, and the third inequality is because $\epsilon\leq 1\leq m$. Therefore we conclude that $1\leq v_i\leq 2$ for all $i\in\mathcal{N}$. Recall that $\alpha=1/2$, we have that for all $i\in S$,
\begin{equation*}
x_i=\phi(i,S)=\dfrac{v_i}{1+\sum_{\ell\in S}v_\ell}\geq \dfrac{1}{1+\sum_{\ell\in S}v_\ell}=\alpha\cdot \dfrac{2}{1+\sum_{\ell\in S}v_\ell}\geq \alpha\cdot \max_{j\in S}\dfrac{v_j}{1+\sum_{\ell\in S}v_\ell}=\alpha\cdot \max_{j\in S}x_j.
\end{equation*}
Furthermore, we have $x_i=0$ for all $i\notin S$. Therefore $\bm{x}$ satisfies the last set of constraints in~the \ref{prob:inventory_upper_bound} problem. Therefore $\bm{x}$ is feasible to~the \ref{prob:inventory_upper_bound} problem. We further have
\begin{align*}
\sum_{i\in\mathcal{N}}r_ix_i=&\dfrac{r_{2m+1}v_{2m+1}+\sum_{i\in A}r_iv_i}{1+v_{2m+1}+\sum_{i\in A}v_i}
=\dfrac{2M+2B+B\sum_{i\in A}(1+\epsilon^2 a_i)-\sum_{i\in A}(1+\epsilon a_i)}{3+m+\epsilon^2Q}\\
=&\dfrac{2M+(m+2)B+\epsilon^2QB-m-\epsilon Q}{3+m+\epsilon^2Q}=\dfrac{(3+m+\epsilon^2Q)B}{3+m+\epsilon^2Q}=B,
\end{align*}
Here, the third equality uses the assumption that $|A|=m$ and $\sum_{i\in A}a_i=Q$, and the fourth equation uses the definition of $B=2M-m-\epsilon Q$. Therefore $\bm{x}$ is feasible to~the \ref{prob:inventory_upper_bound} problem~and its objective value is equal to $T\sum_{i\in\mathcal{N}}r_ix_i=TB$, thus the optimal value of~the \ref{prob:inventory_upper_bound} problem~is greater than or equal to $TB$.

Conversely, suppose the optimal value of~the \ref{prob:inventory_upper_bound} problem~is greater than or equal to $TB$. In this case, there exists a feasible solution $\bm{x}$ of~the \ref{prob:inventory_upper_bound} problem~such that $\sum_{i\in\mathcal{N}}r_ix_i\geq B$. Define $w_i=x_i/x_0$ for all $i\in\mathcal{N}$. We have $x_i=w_i/(1+\sum_{j\in\mathcal{N}}w_j)$ for all $i\in\mathcal{N}$, where we use the fact that $x_0+\sum_{i\in\mathcal{N}}x_i=1$. Since $\sum_{i\in\mathcal{N}}r_ix_i\geq B$, we get
\begin{equation*}
    \sum_{i\in\mathcal{N}}r_iw_i\geq B+\sum_{i\in\mathcal{N}}Bw_i,
\end{equation*}
which is equivalent to $\sum_{i\in\mathcal{N}}(r_i-B)w_i\geq B$. We define $\tilde{r}_i=r_i-B$ for all $i\in\mathcal{N}$, then
\begin{equation*}
\tilde{r}_{i}=-\dfrac{1+\epsilon a_i}{1+\epsilon^2 a_i},\ \forall i\in\{1,2,\dots,2m\},\ \tilde{r}_{2m+1}=M.
\end{equation*}
We also have $\sum_{i\in\mathcal{N}}\tilde{r}_iw_i\geq B$. By the second set of constraints of~the \ref{prob:inventory_upper_bound} problem we get \mbox{$w_{2m+1}\leq v_{2m+1}=2$}, therefore
\begin{equation}\label{eq:upper_bound_rw}
-\sum_{i=1}^{2m}\tilde{r}_iw_i\leq \tilde{r}_{2m+1}w_{2m+1}-B\leq 2M-B=m+\epsilon Q.
\end{equation}

Since $\tilde{r}_i<0$ for all $i\in\{1,2,\dots,2m\}$, we have $\tilde{r}_{2m+1}w_{2m+1}\geq B$, thus
\begin{equation*}
w_{2m+1}\geq \dfrac{B}{\tilde{r}_{2m+1}}=2-\dfrac{m+\epsilon Q}{M}\geq 2-\dfrac{2m}{M}=2-\dfrac{\epsilon^2}{5m},
\end{equation*}
where the second inequality is due to $\epsilon Q\leq 1\leq m$. By the third set of constraints in~the \ref{prob:inventory_upper_bound} problem~we get $x_{2m+1}\leq c_{2m+1}/T=c_{2m+1}'$. Since $x_{2m+1}=w_{2m+1}/(1+\sum_{i\in\mathcal{N}}w_i)$, we get
\begin{equation*}
2-\dfrac{\epsilon^2}{5m}\leq w_{2m+1}\leq c_{2m+1}'\Big(1+\sum_{i\in\mathcal{N}}w_i\Big)=\dfrac{2}{3+m+\epsilon^2Q}\Big(1+\sum_{i\in\mathcal{N}}w_i\Big).
\end{equation*}
Then we get
\begin{equation*}
\sum_{i\in\mathcal{N}}w_i\geq \Big(2-\dfrac{\epsilon^2}{5m}\Big)\dfrac{3+m+\epsilon^2Q}{2}-1\geq m+\epsilon^2 Q+2-\dfrac{\epsilon^2}{2}=m+\Big(Q-\dfrac{1}{2}\Big)\epsilon^2+2,
\end{equation*}
where the last inequality is due to $\epsilon^2Q\leq 1\leq m$, thus $3+m+\epsilon^2Q\leq 5m$. By the second set of constraints in the \ref{prob:inventory_upper_bound} we get $w_{2m+1}\leq v_{2m+1}=2$, therefore
\begin{equation}\label{eq:lower_bound_sum_w}
\sum_{i=1}^{2m} w_i\geq m+\Big(Q-\dfrac{1}{2}\Big)\epsilon^2.
\end{equation}

We define $A=\{i\in\{1,\dots,2m\}:w_i>0\}$. We first prove that $|A|=m$. Since for all $i\in\{1,2,\dots,2m\}$ we have $\epsilon a_i\leq \epsilon\sum_{i=1}^{2m}a_i=2\epsilon Q=1/2m$, we get $w_i\leq v_i=1+\epsilon^2 a_i\leq 1+\epsilon/2m$ for all $i\in\{1,2,\dots,2m\}$. Therefore
\begin{equation*}
\Big(1+\dfrac{\epsilon}{2m}\Big)|A|\geq\sum_{i\in A}w_i=\sum_{i=1}^{2m}w_i\geq m+\Big(Q-\dfrac{1}{2}\Big)\epsilon^2\geq m,
\end{equation*}
and we get
\begin{equation*}
|A|\geq \dfrac{m}{1+\epsilon/2m}\geq m\Big(1-\dfrac{\epsilon}{2m}\Big)=m-\dfrac{\epsilon}{2}.
\end{equation*}
Since $\epsilon\leq 1$, the above inequality implies $|A|\geq m$. Furthermore, by the last set of constraints in~the \ref{prob:inventory_upper_bound} problem~we get $w_{i}\geq w_{2m+1}/2$ for all $i\in A$, and we have proven that $w_{2m+1}\geq 2-\epsilon^2/5m$. Therefore for any $i\in A$,
\begin{equation}\label{eq:lb_nonzero_w}
w_i\geq \dfrac{w_{2m+1}}{2}\geq 1-\dfrac{\epsilon^2}{10m}.
\end{equation}
Furthermore, since we have proven that $-\sum_{i=1}^{2m}\tilde{r}_iw_i\leq m+\epsilon Q$, we have
\begin{equation*}
\Big(1-\dfrac{\epsilon^2}{10m}\Big)|A|\leq -\Big(1-\dfrac{\epsilon^2}{10m}\Big)\sum_{i\in A}\tilde{r}_i\leq -\sum_{i\in A}\tilde{r}_iw_i\leq m+\epsilon Q.
\end{equation*}
Here, the first inequality is due to $\tilde{r}_i\leq -1$ for all $i\in\{1,2,\dots,2m\}$, the second inequality is due to $w_i\geq 1-\epsilon^2/10m$ for all $i\in A$ by~\eqref{eq:lb_nonzero_w}. Therefore
\begin{equation*}
|A|\leq\dfrac{m+\epsilon Q}{1-\epsilon^2/10m}\leq (m+\epsilon Q)\Big(1+\dfrac{\epsilon^2}{5m}\Big)=m+\epsilon Q+\dfrac{\epsilon^2}{5}+\dfrac{\epsilon^3Q}{5m}\leq m+3\epsilon Q=m+\dfrac{3}{4m}.
\end{equation*}
Here the second inequality is due to $0<\epsilon^2/10m\leq 1/10$, thus $(1-\epsilon^2/10m)(1+\epsilon^2/5m)\geq 1$. The third inequality is due to $\epsilon\leq 1\leq Q$, thus $\epsilon^3 Q\leq \epsilon Q$ and $\epsilon^2\leq \epsilon Q$. Therefore we get $|A|\leq m$. Combining this with our previous result that $|A|\geq m$ we conclude that $|A|=m$.

Next, we prove that $\sum_{i\in A}a_i=Q$. Since $w_i\leq v_i$ for all $i\in\mathcal{N}$, we get
\begin{equation*}
\sum_{i\in A}v_i\geq \sum_{i\in A}w_i\geq m+\Big(Q-\dfrac{1}{2}\Big)\epsilon^2,
\end{equation*}
Here the last inequality is due to~\eqref{eq:lower_bound_sum_w}. By definition of $v_i$ we have that $v_i/\epsilon^2\in\mathbb{Z}^+$ for all $i\in\mathcal{N}$, thus $\sum_{i\in A}v_i\geq m+Q\epsilon^2$. Since we have proven that $|A|=m$, we get $\sum_{i\in A}v_i=m+\epsilon^2\sum_{i\in A}a_i\geq m+Q\epsilon^2$, therefore $\sum_{i\in A}a_i\geq Q$. Furthermore, by~\eqref{eq:lb_nonzero_w}~we have $w_i\geq 1-\epsilon^2/10m$ for all $i\in A$, therefore for any $i\in A$,
\begin{equation*}
0\leq v_i-w_i\leq \epsilon^2a_i+\dfrac{\epsilon^2}{10m}.
\end{equation*}
Furthermore, we have $0\leq-\tilde{r}_i\leq 1+\epsilon a_i\leq 2$, thus
\begin{equation*}
-\sum_{i\in A}\tilde{r}_i(v_i-w_i)\leq 2\epsilon^2\sum_{i\in A}a_i+\dfrac{\epsilon^2}{5}\leq 4\epsilon^2Q+\dfrac{\epsilon^2}{5}=\dfrac{\epsilon}{m}+\dfrac{\epsilon^2}{5}\leq \dfrac{7\epsilon}{10}.
\end{equation*}
Here the first inequality is due to $|A|=m$ and $-\tilde{r}_i\leq 2$ for all $i\in\{1,2,\dots,2m\}$, the second inequality is due to $\sum_{i\in A}a_i\leq \sum_{i=1}^{2m} a_i=2Q$, and the third inequality is due to $m\geq 2$ and $\epsilon\leq 1$. Since by~\eqref{eq:upper_bound_rw}~we have $-\sum_{i\in A}\tilde{r}_iw_i\leq m+Q\epsilon$, we get
\begin{equation*}
-\sum_{i\in A}\tilde{r}_iv_i=-\sum_{i\in A}\tilde{r}_iw_i-\sum_{i\in A}\tilde{r}_i(v_i-w_i)\leq m+\Big(Q+\dfrac{7}{10}\Big)\epsilon.
\end{equation*}
Since by definition $-\tilde{r}_iv_i=1+\epsilon a_i$ for all \mbox{$i\in\{1,2,\dots,2m\}$}, we have $-\tilde{r}_iv_i/\epsilon\in\mathbb{Z}^+$ for all \mbox{$i\in\{1,2,\dots,2m\}$}. Therefore, we get
$-\sum_{i\in A}\tilde{r}_iv_i\leq m+Q\epsilon$.
Since we have proven that $|A|=m$, we have $-\sum_{i\in A}\tilde{r}_iv_i=m+\epsilon\sum_{i\in A}a_i\leq m+Q\epsilon$, which implies $\sum_{i\in A}a_i\leq Q$. Combining this with our previous result that $\sum_{i\in A}a_i\geq Q$ we conclude that $\sum_{i\in A}a_i=Q$. 

Combining all discussions above we conclude that the optimal value of~the \ref{prob:inventory_upper_bound} problem~is greater than or equal to $TB$ if and only if there exists $A\subseteq\{1,2,\dots,2m\}$ such that $|A|=m$ and $\sum_{i\in A}a_i=Q$. Since the cardinality constrained subset sum problem is NP-hard (see Appendix 3.2 of \citealt{garey1979computers}), the \ref{prob:inventory_upper_bound} problem~is also NP-hard.

\subsection{Proof of Proposition~\ref{prop:upper_bound_fptas}}\label{sec:prop:upper_bound_fptas}

We give an FPTAS for~the \ref{prob:inventory_upper_bound} problem. We first rewrite~the \ref{prob:inventory_upper_bound} problem as
\begin{equation}\label{prob:inventory_upper_bound_2}
\begin{aligned}
&\max_{\bm{x},y}&&{T\sum_{i\in\mathcal{N}}r_ix_i}\\
&\text{s.t.}&&{x_0+\sum_{i\in\mathcal{N}}x_i\leq1}\\
&&&{0\leq x_i\leq v_ix_0,\ \forall i\in\mathcal{N}}\\
&&&{x_i\leq c_i/T,\ \forall i\in\mathcal{N}}\\
&&&{x_i\in\{0\}\cup[\alpha y,y],\ \forall i\in\mathcal{N}.}
\end{aligned}
\end{equation}
Compared to~the original version of the \ref{prob:inventory_upper_bound} problem provided in Section \ref{sec:upper_bound}, the first and the last set of constraints are rewritten in~Problem~\eqref{prob:inventory_upper_bound_2}. Changes in the first constraint in~the \ref{prob:inventory_upper_bound} problem~does not change the optimal value because if $x_0+\sum_{i\in\mathcal{N}}x_i\leq 1$, setting a new $x_0$ as $1-\sum_{i\in\mathcal{N}}x_i$ still verifies all constraints in~the \ref{prob:inventory_upper_bound} problem. In the last set of constraints in~Problem~\eqref{prob:inventory_upper_bound_2}, we use a decision variable $y$ to capture the maximum of $x_i$.

The key idea for the FPTAS is as follows: Suppose $(\bm{x}^*,y^*)$ is an optimal solution to~Problem~\eqref{prob:inventory_upper_bound_2}. We construct two collections of grid points $\texttt{Grid}_0$ and $\texttt{Grid}_y$ for $x_0^*$ and $y^*$ respectively. We define $y_{\min}=\min\{c_{\min}/T,v_{\min}/(1+nv_{\max})\}$. We will show that $y_{\min}$ is a lower bound for $y^*$. Letting $\delta>0$ be the precision of the grid, we further define $L_0=\lceil\log_{1+\delta}(1+nv_{\max})\rceil$ and \mbox{$L_y=\lceil\log_{1+\delta}(1/y_{\min})\rceil$}. The grids $\texttt{Grid}_0$ and $\texttt{Grid}_y$ are defined as
\begin{equation*}
\texttt{Grid}_0=\Big\{\dfrac{1}{1+nv_{\max}}\cdot(1+\delta)^{\ell}:\ell\in\{0,1,\dots,L_0\}\Big\},\ \texttt{Grid}_y=\Big\{y_{\min}\cdot(1+\delta)^{\ell-1}:\ell\in\{0,\dots,L_y\}\Big\}.
\end{equation*}
The two collections of grid points serve as the guesses for the no-purchase probability and the maximum purchase probability.

Consider any pair of $(\bar{x}_0,\bar{y})\in\texttt{Grid}_0\times\texttt{Grid}_y$. Suppose $x_i>0$, then by the last three sets of constraints in Problem~\eqref{prob:inventory_upper_bound_2} we have $\alpha\bar{y}\leq x_i\leq\bar{y}$, $x_i\leq v_i\bar{x}_0$ and $x_i\leq c_i/T$. Therefore $x_i>0$ implies $\alpha\bar{y}\leq \min\{c_i/T,v_i\bar{x}_0\}$ and $\alpha\bar{y}\leq x_i\leq \min\{c_i/T,\bar{y},v_i\bar{x}_0\}$. We define \mbox{$A(\bar{x}_0,\bar{y})=\{i\in\mathcal{N}:\alpha\bar{y}\leq \min\{v_i\bar{x}_0,c_i/T\}\}$}. Then~we have that $x_i=0$ for all \mbox{$i\notin A(\bar{x}_0,\bar{y})$,} and \mbox{$x_i\in\{0\}\cup[\alpha\bar{y},\min\{\bar{y},c_i/T,v_i\bar{x}_0\}]$} for all \mbox{$i\in A(\bar{x}_0,\bar{y})$}. For any $i\in A(\bar{x}_0,\bar{y})$, we use a set of grid points to discretize the range $[\alpha\bar{y},\min\{\bar{y},c_i/T,v_i\bar{x}_0\}]$. We define \mbox{$L_i=\lceil\log_{1+\delta}(\min\{c_i/T,\bar{y},v_i\bar{x}_0\}/\alpha\bar{y})\rceil$,} and the grid points are defined as
\begin{equation*}
\bar{x}_i^\ell(\bar{x}_0,\bar{y})=\alpha\bar{y}\cdot(1+\delta)^\ell,\ \forall\ell\in\{0,1,\dots,L_i\}.
\end{equation*}
Finally, for each pair of $(\bar{x}_0,\bar{y})$, we generate a candidate solution by approximately solving the following multiple-choice knapsack problem

\begin{equation}\label{prob:knapsack_fptas}
\begin{aligned}
&\max_{\bm{u}}&&\sum_{i\in A(\bar{x}_0,\bar{y})}\sum_{\ell=0}^{L_i}r_i\bar{x}_i^{\ell}(\bar{x}_0,\bar{y})u_i^\ell\\
&\text{s.t.}&&\sum_{i\in A(\bar{x}_0,\bar{y})}\sum_{\ell=0}^{L_i}\bar{x}_i^{\ell}(\bar{x}_0,\bar{y})u_i^\ell\leq 1-\bar{x}_0,\\
&&&\sum_{\ell=0}^{L_i} u_i^\ell\leq 1,\ \forall i\in A(\bar{x}_0,\bar{y}),\\
&&&u_i^\ell\in\{0,1\},\ \forall i\in A(\bar{x}_0,\bar{y}),\ \ell\in\{0,1,\dots,L_i\}.
\end{aligned}
\end{equation}
In Problem~\eqref{prob:knapsack_fptas}, we optimize $\sum_{i\in\mathcal{N}}r_ix_i$ when each $x_i$ is restricted to $\{0\}\cup\{\bar{x}_i^\ell(\bar{x}_0,\bar{y}):\ell\in\{0,\dots,L_i\}\}$. We use the variable $u_i^\ell$ to capture whether $x_i$ is equal to $\bar{x}_i^\ell(\bar{x}_0,\bar{y})$, where $u_i^\ell=1$ if $x_i=\bar{x}_i^\ell(\bar{x}_0,\bar{y})$, and $u_i^\ell=0$ if $x_i\neq \bar{x}_i^\ell(\bar{x}_0,\bar{y})$. If $x_i=0$ then $u_i^\ell=0$ for all $\ell\in\{0,\dots,L_i\}$. Then we have $x_i=\sum_{\ell=0}^{L_i}\bar{x}_i^{\ell}(\bar{x}_0,\bar{y})u_i^\ell$ for all $i\in A(\bar{x}_0,\bar{y})$. Thus the objective value of Problem~\eqref{prob:knapsack_fptas}~is equal to $\sum_{i\in\mathcal{N}}r_ix_i$, and the first constraint in Problem~\eqref{prob:knapsack_fptas}~is equivalent to $\bar{x}_0+\sum_{i\in\mathcal{N}}x_i\leq 1$, which is the first constraint in Problem~\eqref{prob:inventory_upper_bound_2}. Since for each $i\in A(\bar{x}_0,\bar{y})$, $x_i$ can be equal to at most one of $\bar{x}_i^\ell(\bar{x}_0,\bar{y})$, we enforce $\sum_{\ell=0}^{L_i}u_i^\ell\leq 1$ for all $i\in A(\bar{x}_0,\bar{y})$. The complete algorithm is provided in Algorithm~\ref{alg:fptas_inventory_upper_bound}.

\begin{algorithm}[t]
\SingleSpacedXI
\caption{FPTAS for~the \ref{prob:inventory_upper_bound} problem}
\label{alg:fptas_inventory_upper_bound}
\begin{algorithmic}
\For{$(\bar{x}_0,\bar{y})\in \texttt{Grid}_0\times \texttt{Grid}_y$}
\State Obtain an $1/(1+\delta)$-approximation $\bm{u}(\bar{x}_0,\bar{y})$ of~Problem~\eqref{prob:knapsack_fptas}.
\State Set $\bm{x}(\bar{x}_0,\bar{y})$ such that 
\begin{equation}\label{eq:candidate_fptas}
x_i(\bar{x}_0,\bar{y})=\begin{cases}\sum_{\ell=0}^{L_i}\bar{x}_i^{\ell}(\bar{x}_0,\bar{y})u_i^\ell(\bar{x}_0,\bar{y}),\ &\text{if}\ i\in A(\bar{x}_0,\bar{y}),\\
0,\ &\text{if}\ i\notin A(\bar{x}_0,\bar{y}).
\end{cases}
\end{equation}
\State Set $R(\bar{x}_0,\bar{y})=\sum_{i\in\mathcal{N}}r_ix_i(\bar{x}_0,\bar{y})$.
\EndFor
\State Set $(\bar{x}_0^*,\bar{y}^*)=\argmax_{\bar{x}_0,\bar{y}}R(\bar{x}_0,\bar{y})$, and return $\bm{x}(\bar{x}_0^*,\bar{y}^*)$ as defined in \eqref{eq:candidate_fptas}.
\end{algorithmic}
\end{algorithm}

Next, we prove that for any $\epsilon\in(0,1/2)$, by setting $\delta=(1-\epsilon)^{-1/4}-1$, Algorithm~\ref{alg:fptas_inventory_upper_bound}~returns a $(1-\epsilon)$-approximation to~the \ref{prob:inventory_upper_bound} problem.

Let $(\bm{x}^*,y^*)$ be an optimal solution to~Problem~\eqref{prob:inventory_upper_bound_2}. We first prove that $y^*\geq y_{\min}$. We claim that there exists $i\in\mathcal{N}$ such that either $x_i^*=c_i/T$ or $x_i^*=v_ix_0^*$. Suppose $x_i^*<c_i/T$ and $x_i^*<v_ix_0^*$ for all $i\in\mathcal{N}$. Since $x_0^*+\sum_{i\in\mathcal{N}}x_i^*\leq 1$, we have $x_i^*<v_ix_0^*\leq v_i(1-\sum_{j\in\mathcal{N}}x_j^*)$. Then there exists $\delta'>0$ such that $(1+\delta')x_i^*<c_i/T$ and \mbox{$(1+\delta')x_i^*<v_i(1-(1+\delta')\sum_{i\in\mathcal{N}}x_i^*)$} for all $i\in\mathcal{N}$. We define $\hat{\bm{x}}$ as $\hat{x}_i=(1+\delta')x_i^*$ for all $i\in\mathcal{N}$, $\hat{x}_0=1-(1+\delta')\sum_{i\in\mathcal{N}}x_i^*$, and $\hat{y}=(1+\delta')y^*$. Since $\hat{x}_i$ are scaled up by the same ratio from $x_i^*$ for all $i\in\mathcal{N}$, we have that $\hat{\bm{x}}$ satisfies the market share balancing constraint, thus $(\hat{\bm{x}},\hat{y})$ is feasible to~Problem~\eqref{prob:inventory_upper_bound_2}. It is easy to verify that $\sum_{i\in\mathcal{N}}r_ix_i^*>0$, thus $\sum_{i\in\mathcal{N}}r_i\hat{x}_i=(1+\delta')\sum_{i\in\mathcal{N}}r_ix_i^*>\sum_{i\in\mathcal{N}}r_ix_i^*$, which contradicts the assumption that $\bm{x}^*$ is optimal to~Problem~\eqref{prob:inventory_upper_bound_2}. Therefore there exists $i\in\mathcal{N}$ such that either $x^*_i=c_i/T$ or $x^*_i=v_ix_0^*$. If $x_{k}^*=c_{k}/T$ for some $k\in\mathcal{N}$, then $y^*\geq\max_{j\in\mathcal{N}}x_j^*\geq x_{k}^*\geq c_{\min}/T\geq y_{\min}$. Furthermore, it is easy to verify that $x_0^*\geq 1/(1+nv_{\max})$. If $x_{k}^*=v_{k}x_0^*$ for some $k\in\mathcal{N}$, then \mbox{$y^*\geq \max_{j\in\mathcal{N}}x_j^*\geq x_{i_0}^*\geq v_{\min}/(1+nv_{\max})\geq y_{\min}$.} Thus in both cases we get $y^*\geq y_{\min}$.

We set $\hat{x}_0=\max\{\bar{x}_0\in\texttt{Grid}_0:\bar{x}_0\leq x_0^*\}$ and $\hat{y}=\max\{\bar{y}\in\texttt{Grid}_y:(1+\delta)\bar{y}\leq y^*\}$. Then we have $x_0^*/(1+\delta)\leq\hat{x}_0\leq x_0^*$ and $y^*/(1+\delta)^2\leq \hat{y}\leq y^*/(1+\delta)$. We further set $\tilde{x}_i=x_i^*\hat{y}/y^*$ for all $i\in\mathcal{N}$, then we get $x_i^*/(1+\delta)^2\leq\tilde{x}_i\leq x_i^*/(1+\delta)$ for all $i\in\mathcal{N}$. Since $x_i^*\in\{0\}\cup[\alpha y^*,y^*]$ for all $i\in\mathcal{N}$, we get $\tilde{x}_i\in\{0\}\cup[\alpha\hat{y},\hat{y}]$. Since $\tilde{x}_i\leq x_i^*/(1+\delta)$ for all $i\in\mathcal{N}$ and $\hat{x}_0\geq x_0^*/(1+\delta)$, we get $$\tilde{x}_i\leq \dfrac{x_i^*}{1+\delta}\leq \dfrac{v_ix_0^*}{1+\delta}\leq v_i\hat{x}_0.$$ 
Furthermore, we also have $\tilde{x}_i\leq x_i^*\leq c_i/T$ for all $i\in\mathcal{N}$.

Based on $\bm{x}^*$, we construct a feasible solution of Problem~\eqref{prob:knapsack_fptas} with $\bar{x}_0=\hat{x}$ and $\bar{y}=\hat{y}$. For any $i\in\mathcal{N}$, if $\tilde{x}_i>0$ then we set $\hat{x}_i=\max\{\bar{x}_i^\ell(\hat{x}_0,\hat{y}):\bar{x}_i^\ell(\hat{x}_0,\hat{y})\leq \tilde{x}_i,\ell\in\{0,\dots,L_i\}\}$, if $\tilde{x}_i=0$ then we set $\hat{x}_i=0$. For all $i\in\mathcal{N}$ and $\ell\in\{0,1,\dots,L_i\}$, we further define $\hat{\bm{u}}$ as $\hat{u}_i^\ell=1$ if \mbox{$\hat{x}_i=\bar{x}_i^\ell(\bar{x}_0,\bar{y})$,} otherwise $\hat{u}_i^\ell=0$. Since we have proven that $\tilde{x}_i\leq v_i\hat{x}_0$ and $\tilde{x}_i\leq c_i/T$, we have $\tilde{x}_i\in\{0\}\cup[\alpha\hat{y},\min\{\hat{y},v_i\hat{x}_0,c_i/T\}]$ for all $i\in\mathcal{N}$. We further have \mbox{$\hat{x}_i\leq x_i^*$} for all $i\in\mathcal{N}$ and $\hat{x}_0\leq x_0^*$, thus $\hat{x}_0+\sum_{i\in\mathcal{N}}\hat{x}_i\leq x_0^*+\sum_{i\in\mathcal{N}}x_i^*\leq 1$. Therefore $\hat{\bm{u}}$ is feasible to~Problem~\eqref{prob:knapsack_fptas} with $\bar{x}_0=\hat{x}_0$ and $\bar{y}=\hat{y}$. Since for every pair of $(\bar{x}_0,\bar{y})$ we obtain a $1/(1+\delta)$-approximation to~Problem~\eqref{prob:knapsack_fptas}, we have that the objective of the output $\bm{x}$ of Algorithm~\ref{alg:fptas_inventory_upper_bound}~is at least \mbox{$\sum_{i\in\mathcal{N}}r_i\hat{x}_i/(1+\delta)$.} By definition of $\hat{\bm{x}}$ we have \mbox{$\hat{x}_i\geq \tilde{x}_i/(1+\delta)\geq x_i^*/(1+\delta)^3$} for all $i\in\mathcal{N}$. Therefore we have
\begin{equation*}
\sum_{i\in\mathcal{N}}r_ix_i\geq \dfrac{1}{1+\delta}\sum_{i\in\mathcal{N}}r_i\hat{x}_i\geq \dfrac{1}{(1+\delta)^4}\sum_{i\in\mathcal{N}}r_ix_i^*.
\end{equation*}
Since by definition $\delta=(1-\epsilon)^{-1/4}-1$, we have that $1/(1+\delta)^4=1-\epsilon$, therefore we have that
\begin{equation*}
\sum_{i\in\mathcal{N}}r_ix_i\geq (1-\epsilon) \sum_{i\in\mathcal{N}}r_ix_i^*.
\end{equation*}
Thus Algorithm~\ref{alg:fptas_inventory_upper_bound}~returns a $(1-\epsilon)$-approximation algorithm for~the \ref{prob:inventory_upper_bound} problem.

Finally we analyze the runtime of Algorithm~\ref{alg:fptas_inventory_upper_bound}. \cite{bansal2004improved} design an FPTAS for the multiple-choice knapsack problem. Given a multiple-choice knapsack problem with $n$ items and $m$ multiple-choice classes, their algorithm returns a $(1-\epsilon)$-approximation to the knapsack problem in $O(mn/\epsilon)$ time. For any $\bar{x}_0\in\texttt{Grid}_0$ and $\bar{y}\in\texttt{Grid}_y$, we use the FPTAS in~\cite{bansal2004improved}~to obtain a $1/(1+\delta)$-approximation of~Problem~\eqref{prob:knapsack_fptas}, which is a multiple-choice knapsack problem. In Problem~\eqref{prob:knapsack_fptas}, the number of items is $\sum_{i\in A(\bar{x}_0,\bar{y})} L_i$, and the number of multiple-choice class is $|A(\bar{x}_0,\bar{y})|\leq n$. Since $L_i\leq \log_{1+\delta}(\bar{y}/\alpha\bar{y})=\log_{1+\delta}(1/\alpha)=O(\log(1/\alpha)/\epsilon)$ for all $i\in A(\bar{x}_0,\bar{y})$, the runtime of the FPTAS for~Problem~\eqref{prob:knapsack_fptas}~is
\begin{equation*}
O\Big(\Big(1-\dfrac{1}{1+\delta}\Big)^{  -1}n\sum_{i\in A(\bar{x}_0,\bar{y})}L_i\Big)=O\Big(\dfrac{n^2}{\epsilon^2}\log\Big(\dfrac{1}{\alpha}\Big)\Big).
\end{equation*}
We further have
\begin{equation*}
|\texttt{Grid}_0|=L_0+1=O\Big(\dfrac{\log(nv_{\max}+1)}{\epsilon}\Big),\ |\texttt{Grid}_y|=L_y+1=O\Big(\dfrac{\log(1/y_{\min})}{\epsilon}\Big).
\end{equation*}
By definition of $y_{\min}$ we further have
\begin{equation*}
\dfrac{1}{y_{\min}}=\max\Big\{\dfrac{T}{c_{\min}},\dfrac{nv_{\max}+1}{v_{\min}}\Big\}\leq T+\dfrac{nv_{\max}+1}{v_{\min}}.
\end{equation*}
Here the inequality is due to $c_{\min}\in\mathbb{Z}^+$. Therefore we have
\begin{equation*}
|\texttt{Grid}_0\times \texttt{Grid}_y|=O(L_0L_y)=O\Big(\dfrac{1}{\epsilon^2}\log(nv_{\max}+1)\log\Big(T+\dfrac{nv_{\max}+1}{v_{\min}}\Big)\Big).
\end{equation*}
The runtime of identifying $(\bar{x}_0^*,\bar{y}^*)$ is bounded by $O(L_0L_y)$, which is dominated by the runtime of the for-loop in Algorithm \ref{alg:fptas_inventory_upper_bound}. Therefore the runtime of Algorithm~\ref{alg:fptas_inventory_upper_bound} is bounded by
\begin{equation*}
O\Big(\dfrac{n^2}{\epsilon^4}\log\Big(\dfrac{1}{\alpha}\Big)\log(nv_{\max}+1)\log\Big(T+\dfrac{nv_{\max}+1}{v_{\min}}\Big)\Big).
\end{equation*}

\subsection{Proof of Lemma~\ref{lemma:binary_search}}\label{sec:lemma:binary_search}

Let $\tilde{\bm{x}}$ be the approximate solution of the \ref{prob:inventory_upper_bound} returned by Algorithm~\ref{alg:fptas_inventory_upper_bound}. Then there exists $\bar{y}\in\texttt{Grid}_y$ such that $\tilde{x}_i\in\{0\}\cup[\alpha\bar{y},\bar{y}]$ for all $i\in\mathcal{N}$. By definition of $\texttt{Grid}_y$ we have \mbox{$\bar{y}\geq \min \texttt{Grid}_y=y_{\min}/(1+\delta)$.} By definition $\epsilon_1=\epsilon/2\leq 1/2$, thus we have \mbox{$\delta=(1-\epsilon_1)^{-1/4}-1\leq (1-1/2)^{-1/4}-1\leq 1/2$.} Therefore for any $\tilde{x}_i>0$, we get $$\tilde{x}_i\geq \dfrac{\alpha y_{\min}}{1+\delta}\geq \dfrac{\alpha y_{\min}}{2}=\dfrac{\alpha}{2}\min\Big\{\dfrac{c_{\min}}{T},\dfrac{v_{\min}}{1+nv_{\max}}\Big\}\geq \dfrac{\alpha v_{\min}}{2T(1+nv_{\max})}.$$
Here the last inequality is due to $c_{\min}\in\mathbb{Z}^+$. Lemma E.1 in \cite{bai2022coordinated} states that if $Tp\leq c$ for some $c\in\mathbb{Z}^+$, then $\mathbb{E}[(Z(p,T)-c)^+]\leq Tp/2$, where $Z(p,T)$ is a binomial random variable with $T$ trials and success probability being $p$. This result implies that if $Tp\geq c$, then $\mathbb{E}[\max\{Z(p,T),c\}]=Tp-\mathbb{E}[(Z(p,T)-c)^+]\geq Tp/2$. Therefore if $Tp\leq c$ then $G(p,c)\geq Tp/2$. Letting $i_{\min}=\arg\min_{i\in\mathcal{N}}\{G(\tilde{x}_i,c_i):\tilde{x}_i>0\}$, since $\tilde{x}_i\leq c_i/T$ for all $i\in\mathcal{N}$, by Lemma E.1 in \cite{bai2022coordinated}~we get
\begin{equation*}
G_{\min}=G(\tilde{x}_{i_{\min}},c_{i_{\min}})\geq \dfrac{T\tilde{x}_{i_{\min}}}{2}\geq \dfrac{\alpha v_{\min}}{4(1+nv_{\max})}.
\end{equation*}

We claim that for any $c\in\mathbb{Z}^+$, $G(p,c)$ is monotonically increasing and $T$-Lipschitz in $p$. Since for all $c>T$, we have $G(p,c)=Tp=G(p,T)$, it suffices to consider the cases of $c\leq T$. Consider any $0<p_1,p_2<1$, assume without loss of generality that $p_1>p_2$. Suppose $Y_1$ and $Y_2$ are binomial random variables with $T$ trials and success probability being $p_1$ and $p_2$ respectively. For any $i\in\{1,2\}$, we have
\begin{equation*}
G(p_i,c)=\mathbb{E}[\min\{Y_i,c\}]=\sum_{k=1}^c \mathbb{P}(\min\{Y_i,c\}\geq k)=\sum_{k=1}^c \mathbb{P}(Y_i\geq k).
\end{equation*}
Similarly we also have $\mathbb{E}[Y_i]=\sum_{k=1}^T\mathbb{P}(Y_i\geq k)$ for any $i\in\{1,2\}$. Furthermore, by a coupling argument it is easy to verify that for any $k\in\mathbb{Z}^+$, $\mathbb{P}(Y_1\geq k)\geq \mathbb{P}(Y_2\geq k)$. Therefore we have
\begin{align*}
0\leq&G(p_1,c)-G(p_2,c)=\sum_{k=1}^c (\mathbb{P}(Y_1\geq k)-\mathbb{P}(Y_2\geq k))\\
\leq&\sum_{k=1}^T (\mathbb{P}(Y_1\geq k)-\mathbb{P}(Y_2\geq k))=\sum_{k=1}^T\mathbb{P}(Y_1\geq k)-\sum_{k=1}^T\mathbb{P}(Y_2\geq k)=\mathbb{E}[Y_1]-\mathbb{E}[Y_2]=T(p_1-p_2).
\end{align*}
Therefore $G(p,c)$ is monotonically increasing and $T$-Lipschitz in $p$.

If $G(x_i,c_i)>G_{\min}/\alpha$, since for any $c\in\mathbb{Z}^+$ $G(x,c)$ is monotonically increasing and continuous in $x$, there exists $0<y_i<x_i$ such that $G(y_i,c_i)=G_{\min}/\alpha$. We define $M(\epsilon_2)$ as
\begin{equation*}
M(\epsilon_2)=\Big\lceil\log_2 T+\log_2\Big(\dfrac{1}{\epsilon_2}\Big)+\log_2\Big(\dfrac{1+nv_{\max}}{\alpha v_{\min}}\Big)+2\Big\rceil.
\end{equation*}
After $M(\epsilon_2)$ iterations of bisection search, we obtain a interval containing $y_i$ whose length is upper bounded by
\begin{equation*}
2^{-M(\epsilon_2)}\leq \dfrac{\epsilon_2}{T}\cdot\dfrac{\alpha v_{\min}}{4(1+nv_{\max})}\leq \dfrac{\epsilon_2}{T}G_{\min}.
\end{equation*}
If we set $\hat{x}_i$ as left endpoint of the aforementioned interval, we have $0\leq y_i-\hat{x}_i\leq \epsilon_2G_{\min}/T$. Thus we get
\begin{equation*}
0\leq \dfrac{G_{\min}}{\alpha}-G(\hat{x}_i,c_i)=G(y_i,c_i)-G(\hat{x}_i,c_i)\leq T(y_i-\hat{x}_i)\leq \epsilon_2 G_{\min},
\end{equation*}
which implies $G_{\min}/\alpha\geq G(\hat{x}_i,c_i)\geq (1/\alpha-\epsilon_2)G_{\min}\geq (1-\epsilon_2)G_{\min}/\alpha$. Therefore the bisection search is able to return a $\hat{x}_i$ that satisfies $(1-\epsilon_2)G_{\min}/\alpha\leq G(\hat{x}_i,c_i)\leq G_{\min}/\alpha$ with at most $M(\epsilon_2)$ iterations.

\section{Algorithm for $\alpha=1$}\label{sec:policy_alpha=1}

In Section~\ref{sec:dynamic_inventory}, we have established Algorithm \ref{alg:policy_1}, which returns an asymptotically optimal policy for~\ref{prob:inventory}~in the case of $\alpha<1$. However, the algorithm does not work in the case of $\alpha=1$. This is because when $\alpha=1$, we would need the purchase probabilities $\hat{\bm{x}}$ to satisfy $G(\hat{x}_i,c_i)=G_{\min}$ exactly for all $i\in\mathcal{N}$ such that $\tilde{x}_i>0$, where $\tilde{\bm{x}}$ is an approximate solution of the \ref{prob:inventory_upper_bound} problem. Since there is no closed-form expression for the inverse function of $G(x,c)$ in terms of $x$, it is computationally challenging to obtain the precise value of $\hat{x}_i$. In this section, we first provide an algorithm that solves the \ref{prob:inventory_upper_bound} problem exactly in polynomial time when $\alpha=1$. Based on an optimal solution of the \ref{prob:inventory_upper_bound} problem, we introduce an algorithm that returns an asymptotically optimal policy when $\alpha=1$.

\subsection{Algorithm for the \ref{prob:inventory_upper_bound} Problem when $\alpha=1$}

Although we have proven in Proposition~\ref{thm:inventory_upper_bound}~that~the \ref{prob:inventory_upper_bound} problem~is NP-hard, in the special case of $\alpha=1$ we are able to establish an algorithm that solves \ref{prob:inventory_upper_bound} in polynomial time.

We use the following theorem to reveal the structure of the optimal solution to the~\ref{prob:inventory_upper_bound}~problem when $\alpha=1$.

\begin{theorem}\label{thm:inventory_alpha=1}
If $\alpha=1$, then for some $\bar{x}\in\{c_i/T:i\in\mathcal{N}\}\cup\{v_i/(1+\ell v_i):i\in\mathcal{N},\ell\in\{1,2,\dots,n\}\}$ such that $k\bar{x}<1$, there exists an optimal solution $\bm{x}^*$ of the \ref{prob:inventory_upper_bound} problem such that $x^*_i\in\{0,\bar{x}\}$ for all $i\in\mathcal{N}$. Furthermore, there exists $k\in\{1,2,\dots,n\}$ such that $x_i^*=\bar{x}$ if product $i$ are among the top $k$ products in $S(\bar{x},k)$ with the highest revenue, where \mbox{$S(\bar{x},k)=\{i\in\mathcal{N}:c_i\geq T\bar{x},v_i\geq \bar{x}/(1-k\bar{x})\}$}.
\end{theorem}

\begin{proof}
We first prove that any optimal solution $\bm{x}^*$ to~the \ref{prob:inventory_upper_bound} problem~satisfies either $x_i^*=c_i/T$ for some $i\in\mathcal{N}$, or $x_i^*=v_ix_0^*$ for some $i\in\mathcal{N}$. Suppose $x_i^*<c_i/T$ and $x_i^*<v_ix_0^*$ for all $i\in\mathcal{N}$. By the first constraint of the \ref{prob:inventory_upper_bound} problem we have $x_0^*=1-\sum_{i\in\mathcal{N}}x_i^*$, thus we get $x_i^*<v_i(1-\sum_{j\in\mathcal{N}}x_j^*)$ for all $i\in\mathcal{N}$. Then there exists $\epsilon>0$ such that $(1+\epsilon)x_i^*\leq c_i/T$ and $(1+\epsilon)x_i^*\leq v_i(1-(1+\epsilon)\sum_{i\in\mathcal{N}}x_i^*)$ for all $i\in\mathcal{N}$. We define $\hat{x}_i=(1+\epsilon)x_i^*$ and $\hat{x}_0=1-(1+\epsilon)\sum_{i\in\mathcal{N}}x_i^*$, then $\hat{\bm{x}}$ is feasible to~the \ref{prob:inventory_upper_bound} problem. Furthermore, it is easy to verify that $\sum_{i\in\mathcal{N}}r_ix_i^*>0$, thus
\begin{equation*}
\sum_{i\in\mathcal{N}}r_i\hat{x}_i=(1+\epsilon)\sum_{i\in\mathcal{N}}r_ix_i^*>\sum_{i\in\mathcal{N}}r_ix_i^*,
\end{equation*}
which contradicts the assumption that $\bm{x}^*$ is an optimal solution to the \ref{prob:inventory_upper_bound} problem. Therefore either $x_i^*=c_i/T$ for some $i\in\mathcal{N}$, or $x_i^*=v_ix_0^*$ for some $i\in\mathcal{N}$.

Since $\alpha=1$, by the last set of constraints in~the \ref{prob:inventory_upper_bound} problem~there exists $\bar{x}>0$ such that $x_i^*\in\{0,\bar{x}\}$ for all $i\in\mathcal{N}$. If $x_i^*=c_i/T$ for some $i\in\mathcal{N}$, then $\bar{x}\in\{c_i/T:i\in\mathcal{N}\}$. If $x_i^*=v_ix_0^*$ for some $i\in\mathcal{N}$, then $\bar{x}=v_ix_0^*$. We set $k=|\{i\in\mathcal{N}:x_i^*>0\}|$, then by the first constraint we get $k\bar{x}+x_0^*=1$, which implies $\bar{x}=v_i/(1+kv_i)$ for some $i\in\mathcal{N}$. Therefore $\bar{x}\in \{c_i/T:i\in\mathcal{N}\}\cup\{v_i/(1+kv_i):i\in\mathcal{N},\,k\in\{1,2,\dots,n\}\}$.

We set $A=\{i\in\mathcal{N}:x_i^*>0\}$, then by the second and third sets of constraints of~the \ref{prob:inventory_upper_bound} problem~we get $\bar{x}\leq c_i/T$ and $\bar{x}/x_0^*\leq v_i$ for all $i\in A$. Recall that \mbox{$S(\bar{x},k)=\{i\in\mathcal{N}:c_i\geq T\bar{x},v_i\geq \bar{x}/(1-k\bar{x})\}$,} thus $A\subseteq S(\bar{x},k)$. Since by definition $|A|=k$, we get $|S(\bar{x},k)|\geq k$. Furthermore, since the optimal value is $\bar{x}\sum_{i\in A}r_i$, $A$ should consist of the top $k$ products in $S(\bar{x},k)$ with the highest revenue. Therefore $\bm{x}^*$ satisfies the structure identified by the algorithm, thus the algorithm returns an optimal solution to~the \ref{prob:inventory_upper_bound} problem~in polynomial time.
\end{proof}

Based on Theorem~\ref{thm:inventory_alpha=1}, we establish a polynomial time algorithm for the \ref{prob:inventory_upper_bound}~problem when $\alpha=1$. We enumerate all possible values of \mbox{$\bar{x}\in\{c_i/T:i\in\mathcal{N}\}\cup\{v_i/(1+\ell v_i):i\in\mathcal{N},\ell\in\{1,2,\dots,n\}\}$} and $k\in\{1,2,\dots,n\}$. For each combination of $(\bar{x},k)$, we first check whether $|S(\bar{x},k)|$ is greater than or equal to $k$. If $|S(\bar{x},k)|\geq k$, we construct a candidate solution $\bm{x}(\bar{x},k)$ by setting $x_i(\bar{x},k)=\bar{x}$ if product $i$ are among the top $k$ products in $S(\bar{x},k)$ with the highest revenue, and otherwise setting $x_i(\bar{x},k)=0$. Finally we return the candidate solution with the maximum expected revenue. The complete algorithm is provided in Algorithm~\ref{alg:inventory_alpha=1}.

\begin{algorithm}[t]
\SingleSpacedXI
\caption{Algorithm for~the \ref{prob:inventory_upper_bound} problem~when $\alpha=1$}
\label{alg:inventory_alpha=1}
\begin{algorithmic}
\For{$\bar{x}\in\{c_i/T:i\in\mathcal{N}\}\cup\{v_i/(1+\ell v_i):i\in\mathcal{N},\,\ell\in\{1,2,\dots,n\}\}$}
\For{$k\in\{1,2,\dots,n\}$}
\If{$k\bar{x}\geq 1$}
\State Set $R(\bar{x},k)=0$.
\Else
\State Set $x_0(\bar{x},k)=1-k\bar{x}$.
\State Set $S(\bar{x},k)=\{i\in\mathcal{N}:c_i\geq T\bar{x},\,v_i\geq \bar{x}/x_0(\bar{x},k)\}$.
\If{$|S(\bar{x},k)|<k$}
\State Set $R(\bar{x},k)=0$.
\Else
\State Set $(i_1,i_2,\dots,i_k)$ as the top $k$ products with the highest revenue in $S(\bar{x},k)$.
\State Define $\bm{x}(\bar{x},k)$ as $x_{i_\ell}(\bar{x},k)=\bar{x}$ for all $\ell\in\{1,2,\dots,k\}$ and $x_i=0$ otherwise.
\State Set $R(\bar{x},k)=\sum_{i\in\mathcal{N}}r_ix_i(\bar{x},k)$.
\EndIf
\EndIf
\EndFor
\EndFor
\State Set $(\bar{x}^*,k^*)=\argmax_{\bar{x},k}R(\bar{x},k)$, and return $\bm{x}(\bar{x}^*,k^*)$.
\end{algorithmic}
\end{algorithm}

\subsection{Algorithm for \ref{prob:inventory} when $\alpha=1$}

We construct an asymptotically optimal policy based on an optimal solution $\bm{x}^*$ of~the~\ref{prob:inventory_upper_bound}~problem. Similar to the case of $\alpha<1$, if we maintain a purchase probability $x_i^*$ for product $i$ as long as the product has remaining inventory, we may not get a policy that satisfies the market share balancing constraint. When $\alpha<1$, we bypass this challenge by searching for an appropriate $\hat{x}_i$ via a bisection search. However, the same approach cannot be applied in the case of $\alpha=1$. When $\alpha=1$, the market share balancing constraint requires all nonzero expected sales to be exactly equal to each other, which is hard to satisfy using the same approach as the bisection search is unable to return $y_i$ that satisfies $G(y_i,c_i)=G_{\min}$ exactly in finitely many iterations.

For the case of $\alpha=1$, we follow an alternative approach. Instead of reducing the purchase probabilities, we construct a feasible policy by proactively capping the realized sales. Specifically, given an optimal solution $\bm{x}^*$ of~the \ref{prob:inventory_upper_bound} problem, we define $\bar{c}=\min\{c_i:x^*_i>0\}$. Then in each time period, the seller randomly draws an assortment such that the purchase probability for any product $i\in\mathcal{N}$ is $x^*_i$ as long as the realized sales of the product is strictly smaller than $\bar{c}$. Otherwise the purchase probability for product $i$ is set to be zero. In other words, regardless of the actual initial inventory for each product, we always cap the realized sales of each product to be at most $\bar{c}$. Since there exists $\bar{x}>0$ such that $x_i^*\in\{0,\bar{x}\}$ for all $i\in\mathcal{N}$, by applying a uniform cap of $\bar{c}$ on the realized sales, the expected sales of all products are either zero or $G(\bar{x},\bar{c})$, which satisfies the market share balancing constraint. The complete algorithm is provided in Algorithm~\ref{alg:policy_2}. In the last step of Algorithm~\ref{alg:policy_2}, we can construct a distribution over assortments from purchase probabilities following Equation~\eqref{eq:sales_to_distribution}.

\begin{algorithm}[t]
\SingleSpacedXI
\caption{Asymptotically Optimal Policy for \ref{prob:inventory} when $\alpha=1$}
\label{alg:policy_2}
\begin{algorithmic}
\State Obtain the optimal solution $\bm{x}^*$ to~the \ref{prob:inventory_upper_bound} problem using Algorithm \ref{alg:inventory_alpha=1}.
\State Define $\bar{c}=\min\{c_i:x^*_i>0,\ i\in\mathcal{N}\}$.
\State \textbf{Policy:} For each time period $t$, given realized sales $\bm{X}_{t-1}$ by the end of time period $t-1$, set the purchase probabilities $\bm{p}$ as $p_{it}=x^*_i$ if $X_{i,t-1}<\bar{c}$, and $p_{it}=0$ if $X_{i,t-1}=\bar{c}$. Compute a distribution over assortments $\pi_t$ such that the purchase probabilities are $\bm{p}_t$, and offer an assortment randomly drawn from $\pi_t$ at time period $t$.
\end{algorithmic}
\end{algorithm}

We use the following theorem to prove that the policy returned by Algorithm~\ref{alg:policy_2}~is asymptotically optimal.

\begin{theorem}\label{thm:inventory_constant_2}
If $\alpha=1$, Algorithm~\ref{alg:policy_2}~returns a $\max\{{1}/{2},1-{1}/{\sqrt{c_{\min}}}\}$-approximation to~\ref{prob:inventory}.
\end{theorem}

\begin{proof}
We define $\bar{S}=\{i\in\mathcal{N}:x^*_i>0\}$. Since $\alpha=1$, there exists $\bar{x}>0$ such that $x^*_i=\bar{x}$ for all $i\in\bar{S}$. By definition of the policy, we have that for any $i\in\bar{S}$, the purchase probability of product $i$ is $\bar{x}$ until the realized sales of the product reaches $\bar{c}$. Therefore for any $i\in\bar{S}$, we have $\mathbb{E}[X_{iT}^\pi]=G(\bar{x},\bar{c})$ for all $i\in\bar{S}$ and $\mathbb{E}[X_{iT}^\pi]=0$ for all $i\notin\bar{S}$, thus $\pi$ is feasible to~\ref{prob:inventory}.

Since $\bm{x}^*$ is feasible to~the \ref{prob:inventory_upper_bound} problem, we get $T\bar{x}\leq c_i$ for all $i\in\bar{S}$, thus $T\bar{x}\leq \bar{c}$. Therefore by Lemma E.1 in \cite{bai2022coordinated} the expected revenue of the proposed policy is at least 
\begin{equation*}
\sum_{i\in \mathcal{N}}\mathbb{E}[X_{iT}^\pi]=\sum_{i\in\bar{S}}r_iG(\bar{x},\bar{c})\geq \max\Big\{\dfrac{1}{2},1-\dfrac{1}{\sqrt{\bar{c}}}\Big\}\sum_{i\in\bar{S}}r_ix_i^*\geq \max\Big\{\dfrac{1}{2},1-\dfrac{1}{\sqrt{c_{\min}}}\Big\}\sum_{i\in\bar{S}}r_ix_i^*.
\end{equation*}
Since $\bm{x}^*$ is an optimal solution to~the \ref{prob:inventory_upper_bound} problem, the optimal value of~\ref{prob:inventory} is upper bounded by $\sum_{i\in\mathcal{N}}r_ix_i^*$. Therefore the proposed policy is a $\max\{{1}/{2},1-{1}/{\sqrt{c_{\min}}}\}$-approximation to~\ref{prob:inventory}.
\end{proof}

Theorem~\ref{thm:inventory_constant_2}~states that the policy obtained by Algorithm~\ref{alg:policy_2}~is always a $1/2$-approximation to~\ref{prob:inventory}. Furthermore, as $c_{\min}$ increases, the approximation ratio converges to one, which implies that the policy is also asymptotically optimal.

{ 

\section{Joint Inventory Stocking and Dynamic Assortment Optimization with Balanced Market Share}\label{sec:joint_inventory_assortment}

In this section, we consider the joint inventory stocking and dynamic assortment problem with balanced market share. In this problem, the seller first chooses the initial inventories $c_i\in\mathbb{Z}_+$ for all $i\in\mathcal{N}$ such that the total initial inventory does not exceed a maximum total capacity $K$. Then the seller sequentially offers assortments over a time horizon of $T$ time periods, where in each time period $t$, the seller offers an assortment $S_t$, and a customer makes a purchase based on MNL. Similar to the dynamic problem introduced in Section~\ref{sec:dynamic}, the seller's dynamic policy should satisfy two constraints: first, the realized sales of any product $i$ should not exceed its initial inventory $c_i$; second, the expected sales of any two product offered with positive probability should differ by a factor of at most $\alpha$. In this section, we focus on the case of $\alpha<1$.

We use $X_{it}^\pi$ to denote the realized sales of product $i$ by the end of time period $t$ under policy $\pi$. The problem is formally as
\begin{equation}\label{prob:joint_inventory_assortment}
\begin{aligned}
&\max_{\bm{c},\pi}&&\sum_{i\in\mathcal{N}}r_i\mathbb{E}[X_{iT}^\pi]\\
&\text{s.t.}&&\mathbb{P}(X_{iT}^\pi\leq c_i)=1,\ \forall i\in\mathcal{N},\\
&&&\mathbb{E}[X_{iT}^\pi]\in\{0\}\cup\Big[\alpha\cdot \max_{j\in\mathcal{N}}\mathbb{E}[X_{jT}^\pi],\infty\Big),\\
&&&\sum_{i\in\mathcal{N}}c_i\leq K,\\
&&&c_i\in\mathbb{Z}_+,\ \forall i\in\mathcal{N}.
\end{aligned}
\end{equation}

Since the dynamic control policy depends on the history of sales which is high-dimensional, it is computationally intractable to solve for the exact optimal policy to Problem~\eqref{prob:joint_inventory_assortment}. Therefore to solve Problem~\eqref{prob:joint_inventory_assortment}, we establish an optimization problem that provides an upper bound for the optimal value of Problem~\eqref{prob:joint_inventory_assortment}. The upper bound problem is defined as
\begin{equation}\label{prob:joint_inventory_assortment_ub}
\begin{aligned}
&\max_{\bm{c},\bm{x}}&&T\sum_{i\in\mathcal{N}}r_ix_i\\
&\text{s.t.}&&0\leq x_i\leq v_ix_0,\ \forall i\in\mathcal{N},\\
&&&x_0+\sum_{i\in\mathcal{N}}x_i=1,\\
&&&x_i\in\{0\}\cup\Big[\alpha\cdot\max_{j\in\mathcal{N}}x_j,\infty\Big),\ \forall i\in\mathcal{N}\\
&&&x_i\leq c_i/T,\ \forall i\in\mathcal{N},\\
&&&\sum_{i\in\mathcal{N}}c_i\leq K,\\
&&&c_i\in\mathbb{Z}_+,\ \forall i\in\mathcal{N}.
\end{aligned}
\end{equation}
By Proposition~\ref{thm:inventory_upper_bound} we have that the optimal value of Problem~\eqref{prob:joint_inventory_assortment_ub} is greater than or equal to that of Problem~\eqref{prob:joint_inventory_assortment}.

We use the following proposition to show that Problem~\eqref{prob:joint_inventory_assortment_ub} admits an FPTAS.

\begin{proposition}\label{prop:joint_ub_fptas}
Problem~\eqref{prob:joint_inventory_assortment_ub}~admits an FPTAS, and its runtime of obtaining a $(1-\epsilon)$-approximation of Problem~\eqref{prob:joint_inventory_assortment} is
\begin{equation*}
O\Big(\dfrac{n^3}{\epsilon^6}\cdot\log\Big(\dfrac{1}{\alpha}\Big)\log(1+nv_{\max})\Big(\log T+\log\Big(\dfrac{1+nv_{\max}}{v_{\min}}\Big)\Big)\Big(\log T+\log\Big(\dfrac{r_{\max}}{r_{\min}}\Big)+\log\Big(\dfrac{1+nv_{\max}}{v_{\min}}\Big)\Big)\Big).
\end{equation*}
\end{proposition}

Proof of Proposition~\ref{prop:joint_ub_fptas}~is provided in Appendix~\ref{sec:prop:joint_ub_fptas}.

Based on an approximate solution $(\tilde{\bm{c}},\tilde{\bm{x}})$, we can construct dynamic policies for Problem~\eqref{prob:joint_inventory_assortment}~with performance guarantee. To obtain a dynamic policy with constant factor approximation guarantee, we first obtain a $(1-\epsilon/2)$-approximate solution $(\tilde{\bm{c}},\tilde{\bm{x}})$ of Problem~\eqref{prob:joint_inventory_assortment_ub}. Then, we set the initial inventories as $\tilde{\bm{c}}$, and construct a dynamic policy following Algorithm~\ref{alg:policy_1}, where the precision of the bisection search $\epsilon_2$ in the algorithm is set to be $\min\{\epsilon/2,1-\alpha\}$. By Theorem~\ref{thm:inventory_constant_1}, the policy above is a $(1/2-\epsilon)$-approximation to the joint inventory stock and assortment optimization problem.

In what follows, we establish an asymptotically optimal policy for Problem~\eqref{prob:joint_inventory_assortment}, whose approximation ratio converges to $1$ as the maximum total capacity $K$ gets large. In this case, we use a rounding based algorithm similar to the one in \cite{bai2022coordinated}. First, we obtain a $(1-\epsilon_1)$-approximate solution $(\tilde{\bm{c}},\tilde{\bm{x}})$ of Problem~\eqref{prob:joint_inventory_assortment_ub}. Then, we define $\gamma=\lfloor(K/n)^{2/3}\rfloor$, and set $\hat{\bm{c}}$ as
\begin{equation}\label{eq:bar_c}
\bar{c}_i=\lfloor(1-\gamma n/K)\tilde{c}_i\rfloor+\gamma,\ \forall i\in\mathcal{N}.
\end{equation}
Finally, we set the initial inventory as $\bar{\bm{c}}$, and construct a dynamic policy following Algorithm~\ref{alg:policy_1}. The complete algorithm is provided in Algorithm~\ref{alg:joint_inventory_assortment}.

\begin{algorithm}[t]
\SingleSpacedXI
\caption{Asymptotically Optimal Algorithm for Joint Inventory Stocking and Dynamic Assortment Optimization}
\label{alg:joint_inventory_assortment}
\begin{algorithmic}
\State Obtain a $(1-\epsilon_1)$-approximation $(\tilde{\bm{c}},\tilde{\bm{x}})$ of Problem~\eqref{prob:joint_inventory_assortment_ub}.
\State Set the initial inventories as $\bar{\bm{c}}$ defined \eqref{eq:bar_c}.
\State Set $\bar{\bm{x}}$ as $\bar{x}_i=(1-\gamma n/K)\tilde{x}_i$ for all $i\in\mathcal{N}$.
\State Set $\bar{S}=\{i\in\mathcal{N}:\bar{x}_i>0\}$, $G_{\min}=\min_{i\in\bar{S}} G(\bar{x}_i,\bar{c}_i)$.
\For{$i\in\bar{S}$}
\If{$G(\bar{x}_i,\bar{c}_i)\leq G_{\min}/\alpha$}
\State Set $\hat{x}_i=\bar{x}_i$.
\Else
\State Search for $\hat{x}_i$ that satisfies $(1-\epsilon_2)G_{\min}/\alpha\leq G(\hat{x}_i,\bar{c}_i)\leq G_{\min}/\alpha$ using bisection search.
\EndIf
\EndFor
\State \textbf{Policy:} For each time period $t$, given realized sales $\bm{X}_{t-1}$ by the end of time period $t-1$, set the purchase probabilities $\bm{p}_t$ for time period $t$ as $p_{it}=\hat{x}_i$ if $X_{i,t-1}<c_i$, and $p_{it}=0$ if $X_{i,t-1}=c_i$. Compute a distribution over assortments $\pi_t$ such that the purchase probabilities are $\bm{p}_t$, and offer an assortment randomly drawn from $\pi_t$ at time period $t$.
\end{algorithmic}
\end{algorithm}

We show that the policy returned by the Algorithm~\ref{alg:joint_inventory_assortment} is asymptotically optimal, i.e., its approximation ratio converges to one as $K$ gets large.

\begin{theorem}
By setting $\epsilon_1=\epsilon/2$ and $\epsilon_2=\min\{\epsilon/2,1-\alpha\}$, the policy returned by Algorithm~\ref{alg:joint_inventory_assortment} attains an expected revenue at least $(1-(\sqrt{2}+1)(n/K)^{1/3}-\epsilon)$ fraction of the optimal expected revenue.
\end{theorem}

\begin{proof}
Since $1-(\sqrt{2}+1)(n/K)^{1/3}<0$ when $K<n$, it suffices to only consider the case where $K\geq n$. We define $\bar{x}_0=1-\sum_{i\in\mathcal{N}}\bar{x}_i$. We first prove that $(\bar{\bm{c}},\bar{\bm{x}})$ is feasible to Problem~\eqref{prob:joint_inventory_assortment_ub}. Since $\bar{\bm{x}}$ is scaled by a same factor from $\tilde{\bm{x}}$, $\bar{\bm{x}}$ satisfies the the third set of constraints of Problem~\eqref{prob:joint_inventory_assortment_ub}. By definition of $\bar{\bm{x}}$ we have that $\bar{\bm{x}}$ satisfies the second constraint of Problem~\eqref{prob:joint_inventory_assortment_ub}. Since $\bar{x}_i\leq \tilde{x}_i$ for all $i\in\mathcal{N}$ and $\bar{x}_0\geq \tilde{x}_0$, we have that $\bar{\bm{x}}$ satisfies the first set of constraints in Problem~\eqref{prob:joint_inventory_assortment_ub}. Since $\tilde{x}_i\leq \tilde{c}_i/T$ for all $i\in\mathcal{N}$, we have
\begin{equation*}
T\bar{x}_i=(1-\gamma n/K)T\tilde{x}_i\leq (1-\gamma n/K)\tilde{c}_i\leq \lfloor(1-\gamma n/K)\tilde{c}_i\rfloor+1\leq \lfloor(1-\gamma n/K)\tilde{c}_i\rfloor+\gamma=\bar{c}_i.
\end{equation*}
Therefore $(\bar{\bm{c}},\bar{\bm{x}})$ satisfies the fourth set of constraints in Problem~\eqref{prob:joint_inventory_assortment_ub}. We also have
\begin{equation*}
\sum_{i\in\mathcal{N}}\bar{c}_i\leq (1-\gamma n/K)\sum_{i\in\mathcal{N}}\tilde{c}_i+n\gamma\leq K(1-\gamma n/K)+n\gamma=K-n\gamma+n\gamma=K.
\end{equation*}
Therefore $\bar{\bm{c}}$ satisfies the fifth constraint in Problem~\eqref{prob:joint_inventory_assortment_ub}.

Let $(\bm{c}^*,\bm{x}^*)$ be an optimal solution to Problem~\eqref{prob:joint_inventory_assortment_ub}. Following similar arguments in the proof of Theorem 5.1 we have that the expected revenue of the policy is at least
\begin{align*}
&T\Big(1-\dfrac{1}{\sqrt{\min_{j\in\mathcal{N}}\bar{c}_j}}-\epsilon_2\Big)\cdot \sum_{i\in\mathcal{N}}r_i\bar{x}_i\\
\geq&T\left(1-\dfrac{1}{\sqrt{\gamma}}-\epsilon_2\right)\Big(1-\dfrac{\gamma n}{K}\Big)\cdot\sum_{i\in\mathcal{N}}r_i\tilde{x}_i\\
\geq&T\Big(1-\dfrac{1}{\sqrt{\gamma}}-\dfrac{\gamma n}{K}-\epsilon_2\Big)\cdot \sum_{i\in\mathcal{N}}r_i\tilde{x}_i\\
\geq & T\Big(1-\sqrt{2}\Big(\dfrac{n}{K}\Big)^{1/3}-\dfrac{\gamma n}{K}-\epsilon_2\Big)(1-\epsilon_1)\cdot\sum_{i\in\mathcal{N}}r_ix_i^*\\
\geq&T\Big(1-(\sqrt{2}+1)\Big(\dfrac{n}{K}\Big)^{1/3}-\epsilon_2\Big)(1-\epsilon_1)\cdot \sum_{i\in\mathcal{N}}r_i\tilde{x}_i\\
\geq &T\Big(1-(\sqrt{2}+1)\Big(\dfrac{n}{K}\Big)^{1/3}-\epsilon_1-\epsilon_2\Big)\cdot\sum_{i\in\mathcal{N}}r_ix_i^*\\
\geq &T\Big(1-(\sqrt{2}+1)\Big(\dfrac{n}{K}\Big)^{1/3}-\epsilon\Big)\cdot\sum_{i\in\mathcal{N}}r_ix_i^*
\end{align*}
Here the third inequality is because we have assumed $K\geq n$, thus $\gamma=\lfloor(K/n)^{2/3}\rfloor\geq (K/n)^{2/3}/2$, and we get $1/\sqrt{\gamma}\leq \sqrt{2}(K/n)^{1/3}$. The fourth inequality is because
\begin{equation*}
\dfrac{\gamma n}{K}=\Big\lfloor\Big(\dfrac{K}{n}\Big)^{2/3}\Big\rfloor \dfrac{n}{K}\leq \Big(\dfrac{K}{n}\Big)^{2/3}\dfrac{n}{K}=\Big(\dfrac{n}{K}\Big)^{1/3}.
\end{equation*}
Therefore we have that the policy returned by Algorithm~\ref{alg:joint_inventory_assortment} attains an expected revenue at least $(1-(\sqrt{2}+1)(n/K)^{1/3}-\epsilon)$ fraction of the optimal expected revenue.
\end{proof}

\subsection{Proof of Proposition~\ref{prop:joint_ub_fptas}}\label{sec:prop:joint_ub_fptas}

We establish an FPTAS for Problem~\eqref{prob:joint_inventory_assortment_ub}. We first rewrite Problem~\eqref{prob:joint_inventory_assortment_ub}~as follows,

\begin{equation}
\label{prob:joint_ub_2}
\begin{aligned}
&\max_{\bm{x},y}&&T\sum_{i\in\mathcal{N}}r_ix_i\\
&\text{s.t.}&&0\leq x_i\leq v_ix_0,\ \forall i\in\mathcal{N},\\
&&&\sum_{i\in\mathcal{N}}x_i\leq 1-x_0,\\
&&&x_i\in \{0\}\cup [\alpha y,y],\ \forall i\in\mathcal{N},\\
&&&\sum_{i\in\mathcal{N}}\lceil Tx_i\rceil\leq K.
\end{aligned}
\end{equation}

Specifically, in Problem~\eqref{prob:joint_ub_2}, we rewrote the second, the third and the last three constraints in the original version of Problem~\eqref{prob:joint_inventory_assortment_ub}. Changes in the second constraint does not change the optimal value because if $\sum_{i\in\mathcal{N}}x_i\leq 1-x_0$, then setting a new $x_0$ as $1-\sum_{i\in\mathcal{N}}x_i$ still verifies all constraints in Problem~\eqref{prob:joint_inventory_assortment_ub}. In the third constraint of Problem~\eqref{prob:joint_ub_2}, we use a variable $y$ to capture the maximum of $x_i$. Finally, for the last three constraints, if suffices to define $c_i=\lceil Tx_i\rceil$ for all $i\in\mathcal{N}$.

The key idea of the FPTAS is as follows: Suppose $(\bm{x}^*,y^*)$ is an optimal solution of Problem~\eqref{prob:joint_ub_2}. We construct two collections of grid points $\texttt{Grid}_0$ and $\texttt{Grid}_y$ for $x_0^*$ and $y^*$ respectively. We define $y_{\min}=\min\{1/T,v_{\min}/(1+nv_{\max})\}$, and we will show that $y_{\min}$ is a lower bound for $y^*$. Letting $\delta$ be the precision of the grid, we further define $L_0=\lfloor\log_{1+\delta}(1+nv_{\max})\rfloor$, $L_y=\lfloor\log_{1+\delta}(1/y_{\min})\rfloor$. The sets of grid points $\texttt{Grid}_0$ and $\texttt{Grid}_y$ are defined as

\begin{equation*}
\texttt{Grid}_0=\Big\{\dfrac{1}{1+nv_{\max}}\cdot (1+\delta)^\ell:\ell\in\{0,1,\dots,L_0\}\Big\},\ \texttt{Grid}_y=\Big\{y_{\min}\cdot (1+\delta)^{\ell-1}:\ell\in\{0,1,\dots,L_y\}\Big\}.
\end{equation*}

The two collections of grid points serve as guesses for the no-purchase probability and the maximum purchase probability.

Consider any pair of $(\bar{x}_0,\bar{y})\in \texttt{Grid}_0\times \texttt{Grid}_y$. Suppose $x_i>0$, then by the first three constraints in Problem~\eqref{prob:joint_ub_2}~we have that $\alpha\bar{y}\leq x_i\leq \bar{y}$, $x_i\leq 1-\bar{x}_0$ and $x_i\leq v_i\bar{x}_0$, thus $x_i>0$ implies $\alpha\bar{y}\leq x_i\leq \min\{v_i\bar{x}_0,\bar{y},1-\bar{x}_0\}$ and $\alpha\bar{y}\leq v_i\bar{x}_0$. We define $A(\bar{x}_0,\bar{y})=\{i\in\mathcal{N}:v_i\bar{x}_0\geq \alpha\bar{y},\,1-\bar{x}_0\geq \alpha\bar{y}\}$. Then we have that $x_i=0$ for all $i\notin A(\bar{x}_0,\bar{y})$. Then for fixed $\bar{x}_0$ and $\bar{y}$, the problem can be rewritten as
\begin{equation}\label{prob:joint_ub_4}
\begin{aligned}
&\max_{\bm{x}}&&T\sum_{i\in A(\bar{x}_0,\bar{y})}r_ix_i\\
&\text{s.t.}&&\sum_{i\in A(\bar{x}_0,\bar{y})} x_i\leq 1-\bar{x}_0,\\
&&&\sum_{i\in A(\bar{x}_0,\bar{y})}\lceil Tx_i\rceil\leq K,\\
&&&x_i\in\{0\}\cup[\alpha \bar{y},\min\{\bar{y},v_i\bar{x}_0,1-\bar{x}_0\}],\ \forall i\in A(\bar{x}_0,\bar{y}).
\end{aligned}
\end{equation}

To solve Problem~\eqref{prob:joint_ub_4}, we use a collection of grid points $\texttt{Grid}_r$ to serve as guesses of $\sum_{i\in A(\bar{x}_0,\bar{y})}r_ix_i$ in Problem~\eqref{prob:joint_ub_4}. We define $L_r=\lceil\log_{1+\delta}(r_{\max}/y_{\min}r_{\min})\rceil$, and define $\texttt{Grid}_r$ as
\begin{equation*}
\texttt{Grid}_r=\{r_{\min}y_{\min}\cdot(1+\delta)^{\ell-5}:\ell\in\{0,1,\dots,L_r+5\}\}.
\end{equation*}
For every $\bar{r}\in \texttt{Grid}_r$, we study whether there exists $\bm{x}$ that satisfies the following conditions
\begin{equation}\label{prob:joint_ub_5}
\begin{aligned}
&\sum_{i\in A(\bar{x}_0,\bar{y})}r_ix_i\geq \bar{r},\\
&\sum_{i\in A(\bar{x}_0,\bar{y})} x_i\leq 1-\bar{x}_0,\\
&\sum_{i\in A(\bar{x}_0,\bar{y})}\lceil Tx_i\rceil\leq K,\\
&x_i\in\{0\}\cup[\alpha \bar{y},\min\{\bar{y},v_i\bar{x}_0,1-\bar{x}_0\}],\ \forall i\in A(\bar{x}_0,\bar{y}).
\end{aligned}
\end{equation}
We further discretize the range $[\alpha\bar{y},\min\{\bar{y},v_i\bar{x}_0,1-\bar{x}_0\}]$ for each $i\in A(\bar{x}_0,\bar{y})$ using a set of grid points. Specifically, we define $L_i=\lfloor\log_{1+\delta}(\min\{v_i\bar{x}_0,\bar{y},1-\bar{x}_0)\}/\alpha\bar{y}\rfloor+1$, and for each $\ell\in\{1,2,\dots,L_i\}$, we define $\bar{x}_i^\ell$ as
\begin{equation*}
\bar{x}_i^\ell=\alpha\bar{y}\cdot (1+\delta)^{\ell_i-1}.
\end{equation*}
We define $\bar{x}_i^0=0$ for all $i\in A(\bar{x}_0,\bar{y})$. We further define $M=\lceil1/\delta\rceil$, and define $X_i^\ell$, $R_i^\ell$ and $C_i^\ell$ as
\begin{equation*}
X_i^\ell=\Big\lceil\dfrac{Mn\bar{x}_i^\ell}{1-\bar{x}_0}\Big\rceil,\ R_i^\ell=\Big\lfloor \dfrac{Mnr_i\bar{x}_i^\ell}{\bar{r}}\Big\rfloor,\ C_i^\ell=\lceil T\bar{x}_i^\ell\rceil
\end{equation*}
Then we rewrite Problem~\eqref{prob:joint_ub_5}~into the following problem,
\begin{equation}
\begin{aligned}
&\sum_{i\in A(\bar{x}_0,\bar{y})} R_i^{\ell_i}\geq Mn,\\
&\sum_{i\in A(\bar{x}_0,\bar{y})} X_i^{\ell_i}\leq (M+1)n,\\
&\sum_{i\in A(\bar{x}_0,\bar{y})} C_i^{\ell_i}\leq K,\\
&\ell_i\in \{0,1,\dots,L_i\},\ \forall i\in A(\bar{x}_0,\bar{y}).
\end{aligned}
\end{equation}
To solve Problem~\eqref{prob:joint_ub_5}, it suffices to etermine whether the optimal value of the following problem is smaller than or equal to $K$,
\begin{equation}\label{prob:joint_ub_6}
\begin{aligned}
&\min_{\bm{\ell}}&&\sum_{i\in A(\bar{x}_0,\bar{y})} C_i^{\ell_i}\\
&\text{s.t.}&&\sum_{i\in A(\bar{x}_0,\bar{y})}X_i^{\ell_i}\leq (M+1)n,\\
&&&\sum_{i\in A(\bar{x}_0,\bar{y})} R_i^{\ell_i}\geq Mn,\\
&&&\ell_i\in \{0,1,\dots,L_i\},\ \forall i\in A(\bar{x}_0,\bar{y}).
\end{aligned}
\end{equation}
Assume without loss of generality that products in $A(\bar{x}_0,\bar{y})$ are indexed by $1,2,\dots,m$. We define $F_{j}(\tilde{L},\tilde{R})$ as
\begin{equation*}
\begin{aligned}
F_j(\tilde{L},\tilde{R})=&\min_{\bm{\ell}}&&\sum_{i=j}^m C_i^{\ell_i}\\
&\text{s.t.}&&\sum_{i=j}^m X_{i}^{\ell_i}\leq \tilde{L},\\
&&&\sum_{i=j}^m R_i^{\ell_i}\geq \tilde{R},\\
&&&\ell_i\in\{0,1,2,\dots,L_i\},\ \forall i\in\{j,\dots,m\}.
\end{aligned}
\end{equation*}
Then $F_j(\tilde{L},\tilde{R})$ satisfies the following equations
\begin{equation}\label{eq:dp_formulation_new}
\begin{aligned}
&F_j(\tilde{L},\tilde{R})=\min_{\ell}\Big\{C_j^{\ell}+F_{j+1}(\tilde{L}-X_i^\ell,\max\{\tilde{R}-R_i^\ell,0\})\Big\},\\
&F_{j}(\tilde{L},\tilde{R})=\infty,\ \text{if}\ \tilde{L}<0,\\
&F_{m+1}(\tilde{L},0)=0,\ \text{if}\ \tilde{L}\geq 0,\\
&F_{m+1}(\tilde{L},\tilde{R})=\infty,\ \text{if}\ \tilde{R}>0.
\end{aligned}
\end{equation}
Then Problem~\eqref{prob:joint_ub_6} can be solved exactly by solving a dynamic program defined in \eqref{eq:dp_formulation_new}. The complete algorithm is provided in Algorithm~\ref{alg:fptas_joint_ub_new}.

\begin{algorithm}
\SingleSpacedXI
\caption{FPTAS for Problem~\eqref{prob:joint_inventory_assortment_ub}}
\label{alg:fptas_joint_ub_new}
\begin{algorithmic}
\For{$(\bar{x}_0,\bar{y},\bar{r})\in\texttt{Grid}_0\times \texttt{Grid}_y\times \texttt{Grid}_r$}
\State Solve Problem~\eqref{prob:joint_ub_6} using the dynamic program in \eqref{eq:dp_formulation_new}, let $\text{OPT}(\bar{x}_0,\bar{y},\bar{r})$ be the optimal value and $\hat{\ell}$ be an optimal solution.
\If{$\text{OPT}(\bar{x}_0,\bar{y},\bar{r})\leq K$}
\State Set $\bm{x}(\bar{x}_0,\bar{y},\bar{r})$ as $x_i(\bar{x}_0,\bar{y},\bar{r})=\bar{x}_i^{\hat{\ell}_i}/(1+\delta)$ for all $i\in A(\bar{x}_0,\bar{y})$, $x_i(\bar{x}_0,\bar{y},\bar{r})=0$ for all $i\notin A(\bar{x}_0,\bar{y})$ and $x_0(\bar{x}_0,\bar{y},\bar{r})=\bar{x}_0$.
\Else 
\State Set $\bm{x}(\bar{x}_0,\bar{y},\bar{r})$ as $x_i(\bar{x}_0,\bar{y},\bar{r})=0$ and $x_0(\bar{x}_0,\bar{y},\bar{r})=\bar{x}_0$.
\EndIf
\State Set $\bm{c}(\bar{x}_0,\bar{y},\bar{r})$ as $c_i(\bar{x}_0,\bar{y},\bar{r})=\lceil Tx_i(\bar{x}_0,\bar{y},\bar{r})\rceil$ for all $i\in\mathcal{N}$.
\EndFor
\State Return $(\bm{c}(\bar{x}_0,\bar{y},\bar{r}),\bm{x}(\bar{x}_0,\bar{y},\bar{r}))$ that maximizes the objective over all $(\bar{x}_0,\bar{y},\bar{r})\in\texttt{Grid}_0\times \texttt{Grid}_y\times \texttt{Grid}_r$.
\end{algorithmic}
\end{algorithm}

We first prove that for all $(\bar{x}_0,\bar{y},\bar{r})\in \texttt{Grid}_0\times \texttt{Grid}_y\times \texttt{Grid}_r$, $(\bm{c}(\bar{x}_0,\bar{y},\bar{r}),\bm{x}(\bar{x}_0,\bar{y},\bar{r}))$ is feasible to Problem~\eqref{prob:joint_inventory_assortment_ub}. It suffices to consider the case where the optimal value of Problem~\eqref{prob:joint_ub_6} is smaller than or equal to $K$. By definition we have that $\bar{x}_i^\ell\leq v_i\bar{x}_0$ and $\bar{x}_i^\ell\in\{0\}\cup[\alpha\bar{y},\bar{y}]$ for all $i\in A(\bar{x}_0,\bar{y})$. Therefore by definition of $\bm{x}(\bar{x}_0,\bar{y},\bar{r})$ we have that $x_i(\bar{x}_0,\bar{y},\bar{r})\leq v_i\bar{x}_0/(1+\delta)\leq v_i\bar{x}_0$ for all $i\in\mathcal{N}$, and $x_i(\bar{x}_0,\bar{y},\bar{r})\in\{0\}\cup[\alpha\bar{y}/(1+\delta),\bar{y}/(1+\delta)]$ for all $i\in\mathcal{N}$. Therefore $\bm{x}(\bar{x}_0,\bar{y},\bar{r})$ satisfies the first and the third constraints in Problem~\eqref{prob:joint_inventory_assortment_ub}. To verify the second constraint in Problem~\eqref{prob:joint_inventory_assortment_ub}, we have
\begin{align*}
&\sum_{i\in\mathcal{N}}x_i(\bar{x}_0,\bar{y},\bar{r})=\dfrac{1}{1+\delta}\sum_{i\in A(\bar{x}_0,\bar{y})}\bar{x}_i^{\hat{\ell}_i}=\dfrac{1-\bar{x}_0}{Mn(1+\delta)}\sum_{i\in A(\bar{x}_0,\bar{y})}\dfrac{Mn\bar{x}_i^{\hat{\ell}_i}}{1-\bar{x}_0}\\
\leq&\dfrac{1-\bar{x}_0}{Mn(1+\delta)}\sum_{i\in A(\bar{x}_0,\bar{y})}X_i^{\hat{\ell}_i}\leq \dfrac{(1-\bar{x}_0)(M+1)n}{Mn(1+\delta)}\leq 1-\bar{x}_0.
\end{align*}
Here the first inequality is due to the definition of $X_i^\ell$, the second inequality is due to the first constraint in Problem~\eqref{prob:joint_ub_6}, and the third inequality is because $M=\lceil 1/\delta\rceil\geq 1/\delta$, thus $M+1\leq M(1+\delta)$. Therefore we have that $\bm{x}(\bar{x}_0,\bar{y},\bar{r})$ satisfies the second constraint in Problem~\eqref{prob:joint_inventory_assortment_ub}. By definition of $\bm{c}(\bar{x}_0,\bar{y},\bar{r})$ we get $c_i(\bar{x}_0,\bar{y},\bar{r})\geq Tx_i(\bar{x}_0,\bar{y},\bar{r})$ for all $i\in\mathcal{N}$, thus $(\bm{c}(\bar{x}_0,\bar{y},\bar{r}),\bm{x}(\bar{x}_0,\bar{y},\bar{r}))$ satisfies the fourth constraint in Problem~\eqref{prob:joint_inventory_assortment_ub}. To verify the fifth constraint in Problem~\eqref{prob:joint_inventory_assortment_ub}, we have
\begin{equation*}
\sum_{i\in\mathcal{N}}c_i(\bar{x}_0,\bar{y},\bar{r})=\sum_{i\in A(\bar{x}_0,\bar{y})} \Big\lceil\dfrac{T\bar{x}_i^{\hat{\ell}_i}}{1+\delta}\Big\rceil\leq \sum_{i\in A(\bar{x}_0,\bar{y})}\lceil T\bar{x}_i^{\hat{\ell}_i}\rceil=\sum_{i\in A(\bar{x}_0,\bar{y})} C_i^{\hat{\ell}_i}\leq K.
\end{equation*}
Here the second inequality is because we have assumed that the optimal value of Problem~\eqref{prob:joint_ub_6} is smaller than or equal to $K$. Therefore we conclude that $(\bm{c}(\bar{x}_0,\bar{y},\bar{r}),\bm{x}(\bar{x}_0,\bar{y},\bar{r}))$ is feasible to Problem~\eqref{prob:joint_inventory_assortment_ub} for all $(\bar{x}_0,\bar{y},\bar{r})\in\texttt{Grid}_0\times \texttt{Grid}_y\times \texttt{Grid}_r$.

Next, we prove an approximation ratio of Algorithm~\ref{alg:fptas_joint_ub_new}. Let $(\bm{x}^*,y^*)$ be an optimal solution to Problem~\eqref{prob:joint_ub_2}, and define $R^*=\sum_{i\in\mathcal{N}}r_ix_i^*$. Since the no-purchase probability under any assortment is greater than or equal to $1/(1+nv_{\max})$, we have $x_0^*\geq 1/(1+nv_{\max})$. Then we show that $y^*\geq y_{\min}$. Suppose $y^*<y_{\min}$, by the second constraint of Problem~\eqref{prob:joint_ub_2} we get $x_i^*<v_{\min}/(1+nv_{\max})$ and $x_i^*<1/T$ for all $i\in\mathcal{N}$. The former inequality implies that
\begin{equation*}
1-\sum_{i\in\mathcal{N}}x_i^*>1-\dfrac{nv_{\min}}{1+nv_{\max}}\geq 1-\dfrac{nv_{\max}}{1+nv_{\max}}\geq \dfrac{1}{1+nv_{\max}}.
\end{equation*}
Then for any $i\in\mathcal{N}$, we get
\begin{equation*}
x_i^*<\dfrac{v_{\min}}{1+nv_{\max}}\leq \dfrac{v_i}{1+nv_{\max}}\leq v_i\Big(1-\sum_{j\in\mathcal{N}}x_j^*\Big).
\end{equation*}
Therefore there exists $\epsilon'>0$ such that $(1+\epsilon')x_i^*\leq v_i(1-\sum_{j\in\mathcal{N}}(1+\epsilon')x_j^*)$ and $(1+\epsilon')x_i^*\leq 1/T$ for all $i\in\mathcal{N}$. We define $\underline{\bm{x}}$ as $\underline{x}_i=(1+\epsilon')x_i^*$ for all $i\in\mathcal{N}$ and $\underline{x}_0=1-(1+\epsilon')\sum_{i\in\mathcal{N}}x_i^*$. We further define $\underline{y}=(1+\epsilon')y^*$. By definition of $\underline{\bm{x}}$ we have that $\underline{\bm{x}}$ satisfies the first two constraints in Problem~\eqref{prob:joint_ub_2}. Since $(\bm{x}^*,y^*)$ is feasible to Problem~\eqref{prob:joint_ub_2}, we have $x_i^*\in \{0\}\cup[\alpha y^*,y^*]$ for all $i\in\mathcal{N}$. Since $\underline{x}_i=(1+\epsilon')x_i^*$ and $y^*=(1+\epsilon')y^*$, we also have $\underline{x}_i\in\{0\}\cup[\alpha\underline{y},\underline{y}]$. Therefore $\underline{\bm{x}}$ satisfies the third constraint in Problem~\eqref{prob:joint_ub_2}. Finally, we verify that $\underline{\bm{x}}$ satisfies the last constraint in Problem~\eqref{prob:joint_ub_2}. We claim that $\lceil Tx_i^*\rceil=\lceil T\underline{x}_i\rceil$ for all $i\in\mathcal{N}$. If $x_i^*=0$, then $\underline{x}_i=0$ and the equation holds. If $x_i^*>0$, then we have $0<Tx_i^*<T\underline{x}_i\leq 1$, thus in this case $\lceil Tx_i^*\rceil=\lceil T\underline{x}_i\rceil=1$. Therefore in both cases we get $\lceil Tx_i^*\rceil=\lceil T\underline{x}_i\rceil$. Therefore $\sum_{i\in\mathcal{N}}\lceil T\underline{x}_i\rceil=\sum_{i\in\mathcal{N}}\lceil Tx_i^*\rceil\leq K$. Therefore $\underline{\bm{x}}$ is a feasible solution to Problem~\eqref{prob:joint_ub_2}. Furthermore, it is easy to verify that $y=x_1=\min\{1/T,v_{\min}/(1+v_{\max})\}$, $x_i=0$ for all $i\neq 1$ and $x_0=1-x_1$ is a feasible solution, therefore the optimal value of Problem~\eqref{prob:joint_ub_2} is strictly postive. Therefore we have
\begin{equation*}
T\sum_{i\in\mathcal{N}}r_i\underline{x}_i=T(1+\epsilon')\sum_{i\in\mathcal{N}}r_ix_i^*>T\sum_{i\in\mathcal{N}}r_ix_i^*,
\end{equation*}
which contradicts with the assumption that $(\bm{x}^*,y^*)$ is optimal to Problem~\eqref{prob:joint_ub_2}. Therefore we have $y^*\geq y_{\min}$. Finally, since we have proven that there exists $i\in\mathcal{N}$ such that $x_i^*\geq y_{\min}$, we have $r_{\min}y_{\min}\leq R^*\leq r_{\max}$.

We define
\begin{equation*}
\hat{x}_0=\max\{x_0\in\texttt{Grid}_0:x_0\leq x_0^*\},\ \hat{y}=\max\{y\in\texttt{Grid}_y:(1+\delta)y\leq y^*\},\ \hat{r}=\max\{r\in\texttt{Grid}_r:(1+\delta)^4r\leq R^*\}.
\end{equation*}
Then we get $x_0^*/(1+\delta)\leq \hat{x}_0\leq x_0^*$, $y^*/(1+\delta)^2\leq y^*\leq y^*/(1+\delta)$, $R^*/(1+\delta)^5\leq \hat{r}\leq R^*/(1+\delta)^4$. We prove that under $(\hat{x}_0,\hat{y},\hat{r})$, the optimal value of Problem~\eqref{prob:joint_ub_6} is smaller than or equal to $K$. We define $\hat{x}_i=x_i^*\hat{y}/y^*$ for all $i\in\mathcal{N}$. Since $y^*/(1+\delta)^2\leq \hat{y}\leq y^*/(1+\delta)$, we have $x_i^*/(1+\delta)^2\leq\hat{x}_i\leq x_i^*/(1+\delta)$ for all $i\in\mathcal{N}$. Consider any $i\in\mathcal{N}$ such that $x_i^*>0$, we have $\alpha y^*\leq x_i^*\leq y^*$, thus $\alpha\hat{y}\leq \hat{x}_i\leq\hat{y}$. Since $\hat{x}_0/x_0^*\geq 1/(1+\delta)\geq \hat{y}/y^*$ and $x_i^*\leq v_ix_0^*$, we have $\hat{x}_i=x_i^*\hat{y}/y^*\leq v_ix_0^*\cdot \hat{x}_0/x_0^*=v_i\hat{x}_0$. Finally, we have $\hat{x}_i\leq x_i^*\leq 1-x_0^*\leq 1-\hat{x}_0$. Therefore we have $\alpha\hat{x}_0\leq\hat{x}_i\leq \min\{v_i\hat{x}_0,\hat{y},1-\hat{x}_0\}$. We define $\tilde{\bm{\ell}}$ as
\begin{equation*}
\tilde{\ell}_i=\max\{\ell\in\{0,1,\dots,L_i\}:\bar{x}_i^\ell\leq \hat{x}_i\},\ \forall i\in A(\hat{x}_0,\hat{y}).
\end{equation*}
Then we get $x_i^*/(1+\delta)^3\leq \bar{x}_i^{\tilde{\ell}_i}\leq x_i^*/(1+\delta)^2$. Therefore
\begin{equation*}
\sum_{i\in A(\hat{x}_0,\hat{y})}\dfrac{Mnr_i\bar{x}_i^{\tilde{\ell}_i}}{\hat{r}}\geq \dfrac{Mn}{(1+\delta)^3\hat{r}}\sum_{i\in\mathcal{N}}r_ix_i^*=\dfrac{MnR^*}{(1+\delta)^3\hat{r}}\geq \dfrac{MnR^*}{R^*/(1+\delta)}=(1+\delta)Mn\geq (M+1)n.
\end{equation*}
Here the first inequality is due to $x_i^*/(1+\delta)^3\leq \bar{x}_i^{\tilde{\ell}_i}\leq x_i^*/(1+\delta)^2$ for all $i\in A(\hat{x}_0,\hat{y})$, the second inequality is due to $R^*/(1+\delta)^5\leq \hat{r}\leq R^*/(1+\delta)^4$, and the third inequality is due to $M=\lceil1/\delta\rceil\geq 1/\delta$, thus $(1+\delta)M\geq M+1$. Therefore we have
\begin{equation*}
\sum_{i\in A(\hat{x}_0,\hat{y})}R_i^{\tilde{\ell}_i}=\sum_{i\in A(\hat{x}_0,\hat{y})}\Big\lfloor\dfrac{Mnr_i\bar{x}_i^{\tilde{\ell}_i}}{\hat{r}}\Big\rfloor\geq \sum_{i\in A(\hat{x}_0,\hat{y})}\Big(\dfrac{Mnr_i\bar{x}_i^{\tilde{\ell}_i}}{\hat{r}}-1\Big)\geq \sum_{i\in A(\hat{x}_0,\hat{y})}\dfrac{Mnr_i\bar{x}_i^{\tilde{\ell}_i}}{\hat{r}}-n\geq Mn.
\end{equation*}
Therefore $\tilde{\bm{\ell}}$ satisfies the second constraint in Problem~\eqref{prob:joint_ub_6}. To verify the first constraint in Problem~\eqref{prob:joint_ub_6}, we have
\begin{equation*}
\sum_{i\in A(\hat{x}_0,\hat{y})}\dfrac{Mn\bar{x}_i^{\tilde{\ell}_i}}{1-\hat{x}_0}=\dfrac{Mn}{1-\hat{x}_0}\sum_{i\in \mathcal{N}}x_i^{\tilde{\ell}_i}\leq \dfrac{Mn}{1-x_0^*}\sum_{i\in\mathcal{N}}x_i^*\leq \dfrac{Mn(1-x_0^*)}{1-x_0^*}=Mn.
\end{equation*}
Here the first constraint is because $\hat{x}_0\leq x_0^*$, and the second constraint is because by the second constraint in Problem~\eqref{prob:joint_ub_2}~we have $\sum_{i\in\mathcal{N}}x_i^*\leq 1-x_0^*$. Therefore
\begin{equation*}
\sum_{i\in A(\hat{x}_0,\hat{y})}\Big\lceil\dfrac{Mn\bar{x}_i^{\tilde{\ell}_i}}{1-\hat{x}_0}\Big\rceil\leq \sum_{i\in A(\hat{x}_0,\hat{y})}\Big(\dfrac{Mn\bar{x}_i^{\tilde{\ell}_i}}{1-\hat{x}_0}+1\Big)\leq n+\sum_{i\in A(\hat{x}_0,\hat{y})}\dfrac{Mn\bar{x}_i^{\tilde{\ell}_i}}{1-\hat{x}_0}\leq n+Mn=(M+1)n.
\end{equation*}
Therefore $\tilde{\bm{\ell}}$ satisfies the first constraint in Problem~\eqref{prob:joint_ub_6}. Finally, to verify that the objective value of $\tilde{\bm{\ell}}$ in Problem~\eqref{prob:joint_ub_6} is smaller than or equal to $K$, we have that
\begin{equation*}
\sum_{i\in A(\hat{x}_0,\hat{y})} C_i^{\tilde{\ell}_i}=\sum_{i\in A(\hat{x}_0,\hat{y})}\lceil T\bar{x}_i^{\tilde{\ell}_i}\rceil\leq \sum_{i\in \mathcal{N}} \lceil Tx_i^*\rceil\leq K.
\end{equation*}
Here the first inequality is because $\bar{x}_i^{\tilde{\ell}_i}\leq x_i^*$ for all $i\in\mathcal{N}$, and the second inequality is due to the last constraint in Problem~\eqref{prob:joint_ub_2}. Therefore we conclude that the optimal value of Problem~\eqref{prob:joint_ub_6} under $(\hat{x}_0,\hat{y},\hat{r})$ is smaller than or equal to $K$.

Next, we show that the solution returned by Algorithm~\ref{alg:fptas_joint_ub_new} is a $1/(1+\delta)^6$-approximation of Problem~\eqref{prob:joint_inventory_assortment_ub}. Since the objective value of the solution returned by Algorithm~\ref{alg:fptas_joint_ub_new} is greater than or equal to that of $\bm{x}(\hat{x}_0,\hat{y},\hat{r})$, it suffices to prove that $\sum_{i\in\mathcal{N}}r_ix_i(\hat{x}_0,\hat{y},\hat{r})\geq R^*/(1+\delta)^6$. Let $\hat{\bm{\ell}}$ be an optimal solution to Problem~\eqref{prob:joint_ub_6}. Since we have proven that the optimal value of Problem~\eqref{prob:joint_ub_6} under $(\hat{x}_0,\hat{y},\hat{r})$ is smaller than or equal to $K$, we have $x_i(\hat{x}_0,\hat{y},\hat{r})=\bar{x}_i^{\hat{\ell}_i}/(1+\delta)$ for all $i\in A(\hat{x}_0,\hat{y})$. Furthermore, by the second constraint in Problem~\eqref{prob:joint_ub_6}, we have
\begin{equation*}
\sum_{i\in A(\hat{x}_0,\hat{y})} \dfrac{Mnr_i\bar{x}_i^{\hat{\ell}_i}}{\hat{r}}\geq \sum_{i\in A(\hat{x}_0,\hat{y})}R_i^{\hat{\ell}_i}\geq Mn.
\end{equation*}
This implies $\sum_{i\in A(\hat{x}_0,\hat{y})}r_i\bar{x}_i^{\hat{\ell}_i}\geq \hat{r}\geq R^*/(1+\delta)^5$. Finally, by definition of $\bm{x}(\hat{x}_0,\hat{y},\hat{r})$ we have
\begin{equation*}
\sum_{i\in\mathcal{N}}r_ix_i(\hat{x}_0,\hat{y},\hat{r})=\dfrac{1}{1+\delta}\sum_{i\in A(\hat{x}_0,\hat{y})}r_i\bar{x}_i^{\hat{\ell}_i}\geq \dfrac{1}{(1+\delta)^6}\sum_{i\in A(\hat{x}_0,\hat{y})}r_ix_i^*=\dfrac{R^*}{(1+\delta)^6}.
\end{equation*}
Therefore Algorithm~\ref{alg:fptas_joint_ub_new}~returns a $1/(1+\delta)^6$-approximation to Problem~\eqref{prob:joint_inventory_assortment_ub}. By setting $\delta=(1/(1-\epsilon))^{1/6}-1$, the algorithm returns a $(1-\epsilon)$-approximation to Problem~\eqref{prob:joint_inventory_assortment_ub}.

Finally, we analyze the runtime of Algorithm~\ref{alg:fptas_joint_ub_new}. To solve Problem~\eqref{prob:joint_ub_6} via a dynamic program, we need to enumerate $O(n\cdot (M+1)n\cdot Mn)=O(n^3/\delta^2)=O(n^3/\epsilon^2)$ states, and for each state we need to enumerate $L_i$ actions. Furthermore, we have $L_i\leq O(\log_{1+\delta}(1/\alpha))=O(\log(1/\alpha)/\delta)=O(\log(1/\alpha)/\epsilon)$. Therefore the runtime of solving Problem~\eqref{prob:joint_ub_6} is upper bounded by $O(n^3\log(1/\alpha)/\epsilon^3)$. We also need to provide upper bounds for $|\texttt{Grid}_0|$, $|\texttt{Grid}_y|$, $|\texttt{Grid}_r|$. We have
\begin{equation*}
|\texttt{Grid}_0|=O(\log_{1+\delta}(1+nv_{\max}))=O\Big(\dfrac{\log(1+nv_{\max})}{\delta}\Big)=O\Big(\dfrac{\log(1+nv_{\max})}{\epsilon}\Big).
\end{equation*}
We also have
\begin{equation*}
y_{\min}=\min\Big\{\dfrac{1}{T},\dfrac{v_{\min}}{1+nv_{\max}}\Big\}\geq \dfrac{v_{\min}}{T(1+nv_{\max})}.
\end{equation*}
Therefore
\begin{align*}
&|\texttt{Grid}_y|=O(\log_{1+\delta}(1/y_{\min}))=O\Big(\dfrac{1}{\epsilon}\Big(\log T+\log\Big(\dfrac{1+nv_{\max}}{v_{\min}}\Big)\Big)\Big),\\
&|\texttt{Grid}_r|=O(\log_{1+\delta}(r_{\max}/r_{\min}y_{\min}))=O\Big(\dfrac{1}{\epsilon}\Big(\log T+\log\Big(\dfrac{r_{\max}}{r_{\min}}\Big)+\log\Big(\dfrac{1+nv_{\max}}{v_{\min}}\Big)\Big)\Big).
\end{align*}
Then the total runtime of the algorithm is upper bounded by
\begin{equation*}
O\Big(\dfrac{n^3}{\epsilon^6}\cdot\log\Big(\dfrac{1}{\alpha}\Big)\log(1+nv_{\max})\Big(\log T+\log\Big(\dfrac{1+nv_{\max}}{v_{\min}}\Big)\Big)\Big(\log T+\log\Big(\dfrac{r_{\max}}{r_{\min}}\Big)+\log\Big(\dfrac{1+nv_{\max}}{v_{\min}}\Big)\Big)\Big).
\end{equation*}

}

{ 
\section{Numerical Experiments on Synthetic Data}\label{sec:numerical}

In this section, we conduct numerical experiments to test the performance of our policy for the dynamic problem. Since we have established an algorithm that solves the static problem optimally in polynomial time, there is no need to test the performance of our algorithm for the static problem. Thus, we focus on the dynamic problem in the numerical experiments. We first describe the experimental setup of our instances. Then we present the settings of our policy as well as benchmark policies. Finally, we present the results of our policy and benchmark policies on the instances.

\noindent\underline{\bf Experimental Setup:} We use the following steps to generate our test problems. In all of our test problems, the number of products is set to be $n=40$. There are $T$ time periods in the selling horizon, where we vary $T$ in our computational experiments. The revenue $r_i$ for each product $i\in\mathcal{N}$ is sampled from a uniform distribution over $[0,10]$. To generate the preference weight of each product, we first sample $\tilde{v}_i$ from a uniform distribution over $[1,10]$ for each product $i\in\mathcal{N}$. Then we set $v_i=({(1-P_0)}/{P_0})\cdot\tilde{v}_i/\sum_{j\in\mathcal{N}}\tilde{v}_j$ for each $i\in\mathcal{N}$ for a given parameter $P_0$. In this case, we get $\sum_{i\in\mathcal{N}}v_i={(1-P_0)}/{P_0}$, and if we offer all $n$ products, the probability of no-purchase is equal to $1/(1+\sum_{i\in\mathcal{N}}v_i)=1/({(1-P_0)}/{P_0}+1)=P_0$. Thus the parameter $P_0$ controls the likelihood of customers leaving without a purchase, and we also vary the parameter $P_0$ in our experiments.

To generate the initial inventory for each product, we first compute the unconstrained optimal assortment $S^*=\arg\max_{S\subseteq\mathcal{N}}\sum_{i\in S}r_i\phi(i,S)$. If we offer the unconstrained optimal assortment $S^*$ in $3T/4$ time periods and offer all products in $T/4$ time periods, the expected demand for product $i$ is $\texttt{Demand}_i=3T\phi(i,S^*)/4+T\phi(i,\mathcal{N})/4$. We set the initial inventory of product $i$ as $c_i=\lceil\gamma\texttt{Demand}_i\rceil$. In the definition of $\texttt{Demand}_i$, we assume that the seller offers $S^*$ for $3T/4$ time periods because we would like $\texttt{Demand}_i$ to capture the expected demand of product $i$ when the unconstrained optimal assortment $S^*$ is offered for a majority of time periods. However, instead of assuming that the seller offers $S^*$ throughout the time horizon, we assume that the seller offers the entire universe of product in $T/4$ time periods because we would like to generate instances where every product has a nonzero initial inventory. The parameter $\gamma$ captures the ratio of the initial inventory of any product $i$ over its expected demand $\texttt{Demand}_i$, allowing us to vary the scarcity of supply, and we vary the parameter $\gamma$ in our experiments.

Varying $T\in\{2000,4000,8000,16000\}$, $P_0\in\{0.1,0.3\}$, $\gamma\in\{0.6,0.8\}$, we obtain $16$ parameter configurations. We generate a test problem for each parameter configuration, and for each test problem, we vary the balancing parameter $\alpha\in\{0.25,0.5,0.75\}$.

\noindent\underline{\textbf{Our Policy} \textsc{(Pol.)}\textbf{:}} In Algorithm \ref{alg:policy_1}, we solve the \ref{prob:inventory_upper_bound} problem using an FPTAS, but in practice we can also reformulate the \ref{prob:inventory_upper_bound} problem as an integer program and solve the problem efficiently using an integer programming solver. We use an integer programming solver to solve the \ref{prob:inventory_upper_bound} problem exactly and obtain its optimal solution $\bm{x}^*$. There are two main purposes for solving the \ref{prob:inventory_upper_bound} exactly: first, it enables us to obtain the exact optimal value of the \ref{prob:inventory_upper_bound} problem, which serves as a benchmark for the performance of our policy; second, it allows us to construct a policy from an optimal solution to the \ref{prob:inventory_upper_bound} instead of an approximate solution. After obtaining an optimal solution $\bm{x}^*$ to the \ref{prob:inventory_upper_bound} problem, we compute $\hat{\bm{x}}$ in the policy based on $\bm{x}^*$ using Algorithm \ref{alg:policy_1}, where the precision $\epsilon_2$ of the bisection search is set to be $10^{-3}$.

\noindent\underline{\bf Benchmark Policies:} We present two resolving heuristics with balanced market shares. The two heuristics first obtain an optimal solution $\bm{x}^*$ to~the \ref{prob:inventory_upper_bound} problem, and define \mbox{$\bar{S}=\{i\in\mathcal{N}:x^*_i>0\}$}. We denote $\tilde{\bm{x}}^{(t)}$ as the purchase probabilities for time period $t$. For $t=1$, we set $\tilde{\bm{x}}^{(1)}=\bm{x}^*$. For some preselected time period $t$ depending on the history of sales, the resolving heuristics resolve a variant of the \ref{prob:inventory_upper_bound} problem and obtain its optimal solution $\tilde{\bm{x}}^{(t)}$. The problem resolved by the resolving heuristics incorporates the history of sales and is referred to as \ref{prob:resolving} (the ``H" stands for history of sales). Its formal definition is given by

\begin{equation}\label{prob:resolving}
\tag{\texttt{BMS-H}}
\begin{aligned}
&\max_{\bm{x}}&&{\sum_{i\in\bar{S}}r_ix_i}\\
&\text{s.t.}&&{x_0+\sum_{i\in\bar{S}}x_i=1},\\
&&&{0\leq x_i\leq v_ix_0,\ \forall i\in\bar{S}},\\
&&&\sum_{s=1}^{t-1}\tilde{x}_i^{(s)}+(T-t+1)x_i\geq \alpha\cdot\max_{j\in\mathcal{N}}\sum_{s=1}^{t-1}\tilde{x}_j^{(s)}+(T-t+1)x_j,\ \forall i\in\bar{S},\\
&&&{(T-t+1)x_{i}\leq c_i-X_{i,t-1},\ \forall i\in\bar{S}.}
\end{aligned}
\end{equation}

Here we would like to emphasize that $\tilde{x}_i^{(s)}$ and $X_{i,t-1}$ are input parameters of \ref{prob:resolving} instead of decision variables. The first and second constraints guarantee that the purchase probabilities are valid under MNL. The third constraint guarantees that the market share balancing constraint is satisfied if the purchase probability is maintained throughout the remaining $T-t+1$ time periods. The last constraint guarantees that maintaining the same purchase probability throughout the remaining $T-t+1$ time periods does not violate the inventory constraints in expectation. If the policy does not resolve~\ref{prob:resolving}~at time period $t$, then the purchase probabilities for time period $t$ is set to be $\tilde{\bm{x}}^{(t)}=\tilde{\bm{x}}^{(t-1)}$. After defining $\tilde{x}^{(t)}$, we construct a distribution over assortments $\pi_t$ following Equation~\eqref{eq:sales_to_distribution} such that the purchase probability of product $i$ is $\tilde{x}_i$ for all $i\in\mathcal{N}$. With the distribution over assortments $\pi_t$, we sample an assortment $S$ with probability $\pi_t(S)$ and offer the assortment to customers.

We provide the details for the time periods in which the two resolving heuristics resolve~\ref{prob:resolving} respectively.

\noindent\underline{\textsc{Heuristic 1 (Hr.1):}} We resolve \ref{prob:resolving} if either $t=k\Big\lceil\sqrt{T}\Big\rceil$ for some $k\in\mathbb{Z}^+$, or some product runs out of inventory in the previous time period. We set the frequency of resolving to be every $\lceil\sqrt{T}\rceil$ time steps because we would like the time intervals between different resolving to increase sub-linearly in $T$ as $T$ increases. In this case, the resolving heuristic resolves more and the time intervals between different resolving also increase as $T$ gets large.

\noindent\underline{\textsc{Heuristic 2 (Hr.2):}} We resolve \ref{prob:resolving} only when some product runs out of inventory in the previous time period.

In Appendix~\ref{sec:prop:heuristic}, we prove that the definitions of \textsc{Heuristic 1} and \textsc{Heuristic 2} are valid, i.e., \ref{prob:resolving}~is always feasible when resolved, and both heuristics satisfy the market share balancing constraint as well as inventory constraints (see Proposition~\ref{prop:heuristic}~in Appendix \ref{sec:prop:heuristic}).

\begin{table}[t]\label{table:synthetic_data}
\centering
{\scriptsize
\begin{tabular}{|c|c|c|c c c|}
\hline
\multirow{2}{*}{\makecell{Params.\\$(T,P_0,\gamma,\alpha)$}}&\multirow{2}{*}{$\bar{c}$}&\multirow{2}{*}{$K$}&\multicolumn{3}{c|}{Total Exp. Rev.}\\
\cline{4-6}
~&~&~&\textsc{Pol.}&\textsc{Hr.1}&\textsc{Hr.2}\\
\hline
$(2000,0.1,0.6,0.25)$&16&20&0.935&0.763&0.701\\
$(2000,0.1,0.6,0.50)$&28&16&0.925&0.820&0.747\\
$(2000,0.1,0.6,0.75)$&46&11&0.944&0.855&0.820\\
$(2000,0.1,0.8,0.25)$&26&14&0.953&0.887&0.793\\
$(2000,0.1,0.8,0.50)$&45&12&0.947&0.894&0.837\\
$(2000,0.1,0.8,0.75)$&66&10&0.955&0.891&0.856\\
\hline
$(2000,0.3,0.6,0.25)$&15&22&0.929&0.771&0.658\\
$(2000,0.3,0.6,0.50)$&19&17&0.866&0.739&0.637\\
$(2000,0.3,0.6,0.75)$&25&16&0.923&0.768&0.721\\
$(2000,0.3,0.8,0.25)$&11&33&0.916&0.642&0.581\\
$(2000,0.3,0.8,0.50)$&22&24&0.931&0.747&0.690\\
$(2000,0.3,0.8,0.75)$&22&24&0.918&0.719&0.726\\
\hline
\hline
Avg.&&&0.929&0.791&0.731\\
\hline
\end{tabular}
~~~~~~~~~
\begin{tabular}{|c|c|c|c c c|}
\hline
\multirow{2}{*}{\makecell{Params.\\$(T,P_0,\gamma,\alpha)$}}&\multirow{2}{*}{$\bar{c}$}&\multirow{2}{*}{$K$}&\multicolumn{3}{c|}{Total Exp. Rev.}\\
\cline{4-6}
~&~&~&\textsc{Pol.}&\textsc{Hr.1}&\textsc{Hr.2}\\
\hline
$(4000,0.1,0.6,0.25)$&48&13&0.962&0.933&0.833\\
$(4000,0.1,0.6,0.50)$&75&11&0.959&0.916&0.858\\
$(4000,0.1,0.6,0.75)$&75&11&0.953&0.905&0.890\\
$(4000,0.1,0.8,0.25)$&72&11&0.972&0.961&0.877\\
$(4000,0.1,0.8,0.50)$&129&9&0.972&0.955&0.899\\
$(4000,0.1,0.8,0.75)$&165&8&0.971&0.945&0.916\\
\hline
$(4000,0.3,0.6,0.25)$&28&21&0.951&0.889&0.758\\
$(4000,0.3,0.6,0.50)$&37&19&0.931&0.821&0.772\\
$(4000,0.3,0.6,0.75)$&37&19&0.935&0.830&0.804\\
$(4000,0.3,0.8,0.25)$&30&21&0.945&0.868&0.754\\
$(4000,0.3,0.8,0.50)$&42&18&0.925&0.809&0.778\\
$(4000,0.3,0.8,0.75)$&42&18&0.938&0.829&0.828\\
\hline
\hline
Avg.&&&0.951&0.888&0.830\\
\hline
\end{tabular}

\vspace{3em}

\begin{tabular}{|c|c|c|c c c|}
\hline
\multirow{2}{*}{\makecell{Params.\\$(T,P_0,\gamma,\alpha)$}}&\multirow{2}{*}{$\bar{c}$}&\multirow{2}{*}{$K$}&\multicolumn{3}{c|}{Total Exp. Rev.}\\
\cline{4-6}
~&~&~&\textsc{Pol.}&\textsc{Hr.1}&\textsc{Hr.2}\\
\hline
$(8000,0.1,0.6,0.25)$&89&13&0.973&0.943&0.875\\
$(8000,0.1,0.6,0.50)$&157&11&0.964&0.940&0.898\\
$(8000,0.1,0.6,0.75)$&200&9&0.972&0.961&0.937\\
$(8000,0.1,0.8,0.25)$&96&15&0.970&0.969&0.892\\
$(8000,0.1,0.8,0.50)$&222&13&0.976&0.963&0.901\\
$(8000,0.1,0.8,0.75)$&222&13&0.974&0.941&0.912\\
\hline
$(8000,0.3,0.6,0.25)$&40&27&0.961&0.895&0.795\\
$(8000,0.3,0.6,0.50)$&78&21&0.952&0.881&0.819\\
$(8000,0.3,0.6,0.75)$&78&21&0.942&0.878&0.828\\
$(8000,0.3,0.8,0.25)$&53&23&0.951&0.902&0.824\\
$(8000,0.3,0.8,0.50)$&90&20&0.946&0.877&0.839\\
$(8000,0.3,0.8,0.75)$&152&13&0.969&0.940&0.891\\
\hline
\hline
Avg.&&&0.963&0.924&0.867\\
\hline
\end{tabular}
~~~~~~~~
\begin{tabular}{|c|c|c|c c c|}
\hline
\multirow{2}{*}{\makecell{Params.\\$(T,P_0,\gamma,\alpha)$}}&\multirow{2}{*}{$\bar{c}$}&\multirow{2}{*}{$K$}&\multicolumn{3}{c|}{Total Exp. Rev.}\\
\cline{4-6}
~&~&~&\textsc{Pol.}&\textsc{Hr.1}&\textsc{Hr.2}\\
\hline
$(16000,0.1,0.6,0.25)$&162&15&0.980&0.966&0.904\\
$(16000,0.1,0.6,0.50)$&265&13&0.978&0.961&0.926\\
$(16000,0.1,0.6,0.75)$&369&10&0.980&0.968&0.947\\
$(16000,0.1,0.8,0.25)$&329&12&0.985&0.989&0.947\\
$(16000,0.1,0.8,0.50)$&394&11&0.983&0.965&0.943\\
$(16000,0.1,0.8,0.75)$&594&9&0.985&0.981&0.953\\
\hline
$(16000,0.3,0.6,0.25)$&83&23&0.972&0.931&0.858\\
$(16000,0.3,0.6,0.50)$&160&17&0.973&0.947&0.890\\
$(16000,0.3,0.6,0.75)$&220&15&0.975&0.937&0.895\\
$(16000,0.3,0.8,0.25)$&151&19&0.979&0.976&0.898\\
$(16000,0.3,0.8,0.50)$&227&17&0.976&0.950&0.915\\
$(16000,0.3,0.8,0.75)$&378&12&0.981&0.962&0.925\\
\hline
\hline
Avg.&&&0.979&0.961&0.916\\
\hline
\end{tabular}
\vspace{1em}
}
\caption{Relative expected revenue of our policy and two benchmark policies under synthetic instances}\label{table:synthetic_data}
\end{table}

\begin{table}[t]\label{table:ratio_expected_sales}
\centering
{\scriptsize
\begin{tabular}{|c|c c c|}
\hline
\multirow{2}{*}{\makecell{Params.\\$(T,P_0,\gamma,\alpha)$}}&\multicolumn{3}{c|}{Min. Max. Ratio}\\
\cline{2-4}
~&\textsc{Pol.}&\textsc{Hr.1}&\textsc{Hr.2}\\
\hline
$(2000,0.1,0.6,0.25)$&0.246&0.290&0.254\\
$(2000,0.1,0.6,0.50)$&0.491&0.577&0.491\\
$(2000,0.1,0.6,0.75)$&0.739&0.820&0.735\\
$(2000,0.1,0.8,0.25)$&0.247&0.273&0.247\\
$(2000,0.1,0.8,0.50)$&0.500&0.532&0.501\\
$(2000,0.1,0.8,0.75)$&0.750&0.791&0.747\\
\hline
$(2000,0.3,0.6,0.25)$&0.276&0.343&0.293\\
$(2000,0.3,0.6,0.50)$&0.496&0.554&0.501\\
$(2000,0.3,0.6,0.75)$&0.743&0.842&0.720\\
$(2000,0.3,0.8,0.25)$&0.243&0.360&0.244\\
$(2000,0.3,0.8,0.50)$&0.490&0.627&0.486\\
$(2000,0.3,0.8,0.75)$&0.733&0.844&0.732\\
\hline
\end{tabular}
~~~~~~~~
\begin{tabular}{|c|c c c|}
\hline
\multirow{2}{*}{\makecell{Params.\\$(T,P_0,\gamma,\alpha)$}}&\multicolumn{3}{c|}{Min. Max. Ratio}\\
\cline{2-4}
~&\textsc{Pol.}&\textsc{Hr.1}&\textsc{Hr.2}\\
\hline
$(4000,0.1,0.6,0.25)$&0.250&0.280&0.255\\
$(4000,0.1,0.6,0.50)$&0.497&0.547&0.496\\
$(4000,0.1,0.6,0.75)$&0.746&0.807&0.735\\
$(4000,0.1,0.8,0.25)$&0.312&0.324&0.323\\
$(4000,0.1,0.8,0.50)$&0.529&0.562&0.538\\
$(4000,0.1,0.8,0.75)$&0.745&0.761&0.750\\
\hline
$(4000,0.3,0.6,0.25)$&0.270&0.317&0.278\\
$(4000,0.3,0.6,0.50)$&0.499&0.571&0.495\\
$(4000,0.3,0.6,0.75)$&0.735&0.805&0.741\\
$(4000,0.3,0.8,0.25)$&0.249&0.290&0.252\\
$(4000,0.3,0.8,0.50)$&0.495&0.566&0.493\\
$(4000,0.3,0.8,0.75)$&0.736&0.823&0.745\\
\hline
\end{tabular}

\vspace{3em}

\begin{tabular}{|c|c c c|}
\hline
\multirow{2}{*}{\makecell{Params.\\$(T,P_0,\gamma,\alpha)$}}&\multicolumn{3}{c|}{Min. Max. Ratio}\\
\cline{2-4}
~&\textsc{Pol.}&\textsc{Hr.1}&\textsc{Hr.2}\\
\hline
$(8000,0.1,0.6,0.25)$&0.251&0.273&0.250\\
$(8000,0.1,0.6,0.50)$&0.497&0.527&0.500\\
$(8000,0.1,0.6,0.75)$&0.747&0.767&0.744\\
$(8000,0.1,0.8,0.25)$&0.250&0.249&0.249\\
$(8000,0.1,0.8,0.50)$&0.500&0.531&0.502\\$(8000,0.1,0.8,0.75)$&0.745&0.776&0.749\\
\hline
$(8000,0.3,0.6,0.25)$&0.249&0.281&0.248\\
$(8000,0.3,0.6,0.50)$&0.498&0.541&0.499\\
$(8000,0.3,0.6,0.75)$&0.749&0.791&0.737\\
$(8000,0.3,0.8,0.25)$&0.250&0.265&0.251\\
$(8000,0.3,0.8,0.50)$&0.496&0.535&0.500\\
$(8000,0.3,0.8,0.75)$&0.747&0.783&0.750\\
\hline
\end{tabular}
~~~~~~~~
\begin{tabular}{|c|c c c|}
\hline
\multirow{2}{*}{\makecell{Params.\\$(T,P_0,\gamma,\alpha)$}}&\multicolumn{3}{c|}{Min. Max. Ratio}\\
\cline{2-4}
~&\textsc{Pol.}&\textsc{Hr.1}&\textsc{Hr.2}\\
\hline
$(16000,0.1,0.6,0.25)$&0.250&0.261&0.249\\
$(16000,0.1,0.6,0.50)$&0.498&0.519&0.499\\
$(16000,0.1,0.6,0.75)$&0.749&0.761&0.747\\
$(16000,0.1,0.8,0.25)$&0.362&0.368&0.366\\
$(16000,0.1,0.8,0.50)$&0.498&0.511&0.498\\
$(16000,0.1,0.8,0.75)$&0.749&0.757&0.749\\
\hline
$(16000,0.3,0.6,0.25)$&0.250&0.275&0.251\\
$(16000,0.3,0.6,0.50)$&0.501&0.514&0.497\\
$(16000,0.3,0.6,0.75)$&0.748&0.765&0.749\\
$(16000,0.3,0.8,0.25)$&0.307&0.315&0.309\\
$(16000,0.3,0.8,0.50)$&0.500&0.522&0.501\\
$(16000,0.3,0.8,0.75)$&0.780&0.811&0.778\\
\hline
\end{tabular}
\vspace{1em}
}
\caption{Ratio between the minimum over maximum nonzero mean sales of our policy and two benchmark policies under synthetic instances}
\label{table:ratio_expected_sales}
\end{table}

\noindent\underline{\bf Results:} We compute the total expected revenue of our policy as well as the two resolving heuristics using simulation, where for each instance, we use $400$ replicates to estimate the total expected revenue for each of the three policies. We normalize the total expected revenue obtained by the three policies using the optimal value of~the \ref{prob:inventory_upper_bound} problem. We define $\bar{c}=\min\{c_i:x_i^*>0\}$, and $K=|\{i\in\mathcal{N}:x_i^*>0\}|$, where $\bm{x}^*$ is an optimal solution to the \ref{prob:inventory_upper_bound} problem. The normalized expected revenue of all three policies as well as $\bar{c}$ and $K$ under each instance are provided in Table~\ref{table:synthetic_data}. We present $\bar{c}$ in Table~\ref{table:synthetic_data} to numerically validate our theoretical result that our policy is asymptotically optimal as initial inventories increase. We present $K$ to verify that our policy is practical and offers a large enough variety of products. By doing so we verify that our policy does not fall into an undesirable situation where a balanced market share is achieved by dropping most of the products and offering only a very small number of products.

When comparing the performance between the three policies, one can observe from Table~\ref{table:synthetic_data} that our policy outperforms both \textsc{Heuristic 1} and \textsc{Heuristic 2} by significant margins in most of the instances. One possible reason for this observation is that for these two heuristics, the purchase probabilities of a product can be affected not only by its own availability, but also by the availability of other products as well. Specifically, suppose under the resolving heuristics, product $i$ runs out of inventory at time period $t_0$, then after time period $t_0$ the seller cannot offer product $i$, thus the cumulative purchase probability of product $i$ will always stay at $\sum_{t=1}^{t_0}\tilde{x}_i^{(t)}$ afterwards. Due to the third set of constraints in~\ref{prob:resolving}, the cumulative purchase probability of any other product will be limited to at most $\sum_{t=1}^{t_0}\tilde{x}_i^{(t)}/\alpha$. If $\bar{S}$ is large, it is likely that one of the products runs out of inventory relatively early, resulting in a relatively small cumulative purchase probability of the product, and this could significantly limit the cumulative purchase probabilities of all other products. In contrast, with a set of appropriate purchase probabilities $\hat{\bm{x}}$, our policy is able to maintain a fixed purchase probability $\hat{x}_i$ for each product $i\in\mathcal{N}$ as long as the product itself has remaining inventory. In our policy, whether product $i$ can be offered only depends on whether product $i$ still has remaining inventory, and is independent from the availability of any other product. By removing the dependence of purchase probabilities on the availability of other products, our policy is able to achieve a better empirical performance than both \textsc{Heuristic 1} and \textsc{Heuristic 2}.

One can also observe from Table~\ref{table:synthetic_data}~that when either $T$ or $\gamma$ increases, $\bar{c}$ increases and the performance of all three policies, with respect to the upper bound on the total expected revenue, improves. In particular, the relative expected revenue of our policy increases as $T$ increases from $2000$ to $16000$. This is consistent with Theorem~\ref{thm:inventory_constant_1}, which states that our policy is asymptotically optimal as initial inventories increase.

To numerically verify that all three policies satisfy the market share balancing constraint, we also present the ratio of the minimum nonzero expected sales over the maximum nonzero expected sales in Table~\ref{table:ratio_expected_sales}. For all three policies, we compute the ratio using the expected sales estimated through $400$ independent simulations. From Table~\ref{table:ratio_expected_sales}~one can see that for \textsc{Heuristic 1}, the ratios of the minimum nonzero expected sales over the maximum expected sales are greater than $\alpha$ under all instances. For our policy \textsc{Heuristic 2}, the ratios under some instances are marginally smaller than $\alpha$, but this can be explained by the sampling error due to simulation. This suggests that all three policies satisfy the market share balancing constraint, which is consistent with Theorem~\ref{thm:inventory_constant_1}~and Proposition~\ref{prop:heuristic}.}

\section{Feasibility of Resolving Heuristics}\label{sec:prop:heuristic}

In this section, we prove that the definitions of \textsc{Heuristic 1} and \textsc{Heuristic 2} are valid, i.e., \ref{prob:resolving}~is always feasible when resolved, and both heuristics satisfy the market share balancing constraint as well as inventory constraints. 

\begin{proposition}\label{prop:heuristic}
For both \textsc{Heuristic 1} and \textsc{Heuristic 2}, \ref{prob:resolving}~is always feasible when resolved, and both heuristics are feasible to~\ref{prob:inventory}.
\end{proposition}

\begin{proof}
We denote the set of all possible history of sales up till time period $t$ as $\mathcal{H}_t$. Recall that $\tilde{\bm{x}}^{(t)}$ gives the purchase probabilities at time period $t$. Since $\tilde{\bm{x}}^{(t)}$ is determined by the history of sales up till time period $t-1$, we consider $\tilde{\bm{x}}^{(t)}$ as a random variable in the history of sales space $\mathcal{H}_{t-1}$. As a generalization of both \textsc{Heuristic 1} and \textsc{Heuristic 2}, we consider any policy $\pi$ in the following form: the policy always resolves~\ref{prob:resolving}~whenever some product runs out of inventory exactly in the previous time period, but the policy may also resolve~\ref{prob:resolving}~at other time periods depending on the history of sales.

Recall that $\bar{S}=\{i\in\mathcal{N}:x_i^*>0\}$, where $\bm{x}^*$ is an optimal solution to the \ref{prob:inventory_upper_bound} problem. We first prove that in any of the aforementioned policies, for any $t\in\{1,2,\dots,T\}$, \ref{prob:resolving} is always feasible up till time period $t$, and for all $i,j\in\bar{S}$, the following inequality holds with probability one,
\begin{equation}\label{eq:balance_t}
\sum_{s=1}^t \tilde{x}_i^{(s)}\geq \alpha\sum_{s=1}^t\tilde{x}_j^{(s)}.
\end{equation}
We prove the result by induction. Letting $\bm{x}^*$ be an optimal solution to the~\ref{prob:inventory_upper_bound}~problem, since by definition $\tilde{\bm{x}}^{(1)}=\bm{x}^*$, it is easy to verify that when $t=1$, \ref{prob:resolving}~is always feasible, and by the feasibility of $\bm{x}^*$ we have $\tilde{\bm{x}}^{(1)}$ satisfies $\tilde{x}_i^{(1)}\geq \alpha\tilde{x}_j^{(1)}$ for all $i,j\in\bar{S}$, therefore the result holds for $t=1$. 

Suppose \ref{prob:resolving} is always feasible up till time step $t_0$, and \eqref{eq:balance_t} holds for all $t\leq t_0$ for some $t_0<T$. Consider any history of sales $\bm{H}_{t_0}\in\mathcal{H}_{t_0}$ up till time period $t_0$. Suppose under history of sales $\bm{H}_{t_0}$, the policy resolves~\ref{prob:resolving} at time period $t_0+1$. By the induction assumption we have that for all $i,j\in\bar{S}$, $\sum_{s=1}^{t_0}\tilde{x}_i^{(s)}\geq \alpha\sum_{s=1}^{t_0}\tilde{x}_j^{(s)}$. Therefore when resolving~\ref{prob:resolving}~at time period $t_0+1$, it is easy to verify that $x_i=0$ for all $i\in\bar{S}$, is feasible to~\ref{prob:resolving}, therefore~\ref{prob:resolving}~is always feasible up till time period $t_0+1$.

Then we prove that for all $i,j\in\bar{S}$, we have $\sum_{s=1}^{t_0+1}\tilde{x}_i^{(s)}\geq \alpha\sum_{s=1}^{t_0+1}\tilde{x}_j^{(s)}$. Consider any history of sales $\bm{H}_{t_0}\in\mathcal{H}_{t_0}$. We define $\tau$ as the last time period up till $t_0+1$ in which the policy resolves~\ref{prob:resolving} under history of sales $\bm{H}_{t_0}$. If the policy resolves~\ref{prob:resolving}~at time period $t_0+1$ then we define $\tau=t_0+1$. By the induction assumption we have that for all $i,j\in\bar{S}$
\begin{equation}\label{eq:resolving_1}
\sum_{s=1}^{\tau-1}\tilde{x}_i^{(s)}\geq \alpha\sum_{s=1}^{\tau-1}\tilde{x}_j^{(s)}.
\end{equation}
Since $\tilde{\bm{x}}^{(\tau)}$ is an optimal solution when resolving~\ref{prob:resolving}~at time period $\tau$, by the third set of constraints in~\ref{prob:resolving}~we get
\begin{equation}\label{eq:resolving_2}
\sum_{s=1}^{\tau-1}\tilde{x}_i^{(s)}+(T-\tau+1)\tilde{x}_i^{(\tau)}\geq \alpha\sum_{s=1}^{\tau-1}\tilde{x}_j^{(s)}+\alpha(T-\tau+1)\tilde{x}_j^{(\tau)}.
\end{equation}
By definition of $\tau$ we have that no resolving happens between time period $\tau+1$ and $t_0+1$, thus $\tilde{\bm{x}}_i^{(s)}=\tilde{\bm{x}}^{(\tau)}$ for all $\tau\leq s\leq t_0+1$. Therefore $\sum_{s=1}^{t_0+1}\tilde{x}_i^{(s)}=\sum_{s=1}^{\tau-1}\tilde{x}_i^{(s)}+(t_0-\tau+2)\tilde{x}_i^{\tau}$ for all $i\in\bar{S}$. Combining \eqref{eq:resolving_1} and \eqref{eq:resolving_2} we have that for all $i,j\in\bar{S}$,
\begin{align*}
&\sum_{s=1}^{t_0+1}\tilde{x}_i^{(s)}=\sum_{s=1}^{\tau-1}\tilde{x}_i^{(s)}+(t_0-\tau+2)\tilde{x}_i^{(\tau)}=\dfrac{T-t_0-1}{T-\tau+1}\sum_{s=1}^{\tau-1}\tilde{x}_i^{(s)}+\dfrac{t_0-\tau+2}{T-\tau+1}\Big(\sum_{s=1}^{\tau-1}\tilde{x}_i^{(s)}+(T-\tau+1)\tilde{x}_i^{(\tau)}\Big)\\
\geq&\dfrac{T-t_0-1}{T-\tau+1}\alpha\sum_{s=1}^{\tau-1}\tilde{x}_j^{(s)}+\dfrac{t_0-\tau+2}{T-\tau+1}\alpha\Big(\sum_{s=1}^{\tau-1}\tilde{x}_j^{(s)}+(T-\tau+1)\tilde{x}_j^{(\tau)}\Big)=\alpha\sum_{s=1}^{t_0+1}\tilde{x}_j^{(s)}.
\end{align*}
Here the inequality is due to $t_0<T$ and $\tau\leq t_0+1$, thus $T-t_0-1\geq 0$ and $t_0-\tau+2\geq 0$. Therefore we conclude that the result holds for $t=t_0+1$. By induction \ref{prob:resolving} is always feasible, and \eqref{eq:balance_t} holds for all $t\in\{1,2,\dots,T\}$. Therefore~\ref{prob:resolving}~is feasible whenever resolved. We further have $\sum_{t=1}^T\tilde{x}_i^{(t)}\geq \alpha\sum_{t=1}^T\tilde{x}_j^{(t)}$ holds with probability one for all $i,j\in\bar{S}$.

We claim that
\begin{equation*}
\mathbb{E}[X_{iT}^\pi]=\sum_{t=1}^T\mathbb{E}[\tilde{x}_i^{(t)}].
\end{equation*}
Consider any $t\in\{1,2,\dots,T\}$ and any $i\in\bar{S}$, we have that
\begin{equation*}
\mathbb{E}[X_{it}^\pi-X_{i,t-1}^\pi|\bm{H}_{t-1}]=\tilde{x}_i^{(t)}.
\end{equation*}
Taking the expectation over $\bm{H}_{t-1}$ we get
\begin{equation*}
\mathbb{E}[X_{it}^\pi-X_{i,t-1}^\pi]=\mathbb{E}[\tilde{x}_i^{(t)}].
\end{equation*}
Since $X_{i0}=0$ for all $i\in\bar{S}$, we get
\begin{equation*}
\mathbb{E}[X_{iT}^\pi]=\sum_{t=1}^T\mathbb{E}[X^\pi_{it}-X^\pi_{i,t-1}]=\sum_{t=1}^T\mathbb{E}[\tilde{x}_i^{(t)}].
\end{equation*}
We have proven that $\sum_{t=1}^T\tilde{x}_i^{(t)}\geq \alpha\sum_{t=1}^T\tilde{x}_j^{(t)}$ holds with probability one for all $i,j\in\bar{S}$. Thus for all $i,j\in\bar{S}$,
$$\mathbb{E}[X_{iT}^\pi]=\sum_{t=1}^T\mathbb{E}[\tilde{x}_i^{(t)}]\geq\alpha\sum_{t=1}^T\mathbb{E}[\tilde{x}_j^{(t)}]=\alpha\mathbb{E}[X_{jT}^\pi].$$
It is easy to verify that $\mathbb{E}[X_{iT}^\pi]=0$ if $i\notin\bar{S}$, therefore for any $i\in\mathcal{N}$, we get
\begin{equation*}
\mathbb{E}[X_{iT}^\pi]\in\{0\}\cup\Big[\alpha\cdot\max_{j\in\mathcal{N}}\mathbb{E}[X_{jT}^\pi],\infty\Big).
\end{equation*}
Furthermore, as soon as any product $i\in\bar{S}$ runs out of inventory, we resolve~\ref{prob:resolving}, and the purchase probability for product $i$ is always zero afterwards. Therefore $X_{it}^\pi\leq c_i$ with probability one. Then we conclude that the policy is feasible to~\ref{prob:inventory}. In particular, both \textsc{Heuristic 1} and \textsc{Heuristic 2} are well-defined policies and feasible to~\ref{prob:inventory}.
\end{proof}

\end{appendices}

\end{document}